\documentclass{article}
\usepackage{arxiv}
\pdfoutput=1

\usepackage[utf8]{inputenc} 
\usepackage[T1]{fontenc}    
\usepackage{microtype}      
\usepackage{hyperref}       
\usepackage{url}            
\usepackage{graphicx}
\usepackage{tikz}
\usepackage{makecell, multirow, tabularx}
\usepackage{booktabs}       
\usepackage{caption}
\usepackage{subcaption}
\usepackage{rotating}
\usepackage{placeins}        
\usepackage{afterpage}       
\usepackage[export]{adjustbox}
\usepackage[ruled]{algorithm2e}
\usepackage{amsmath, amssymb, amsthm, mathtools}
\usepackage{dsfont, mathrsfs, bm}
\usepackage{nicefrac}   

\usepackage[authoryear]{natbib}
\usepackage{doi}

\newcolumntype{C}{>{\centering\arraybackslash}X}

\newcommand{\1}[1]{\mathds{1}\!\left(#1\right)}          % indicator
\newcommand{\abs}[1]{\left\lvert #1 \right\rvert}        % absolute value
\DeclarePairedDelimiterX{\norm}[1]{\lVert}{\rVert}{#1}   % norm

\DeclareMathOperator{\X}{\mathbf{X}}        % point process
\DeclareMathOperator{\x}{\mathbf{x}}        % point pattern

\DeclareMathOperator{\N}{\mathbb{N}}        % Natural numbers
\DeclareMathOperator{\R}{\mathbb{R}}        % Reals
\DeclareMathOperator{\Rt}{\mathbb{R}^2}     % R^2
\DeclareMathOperator{\So}{\mathbb{S}^1}        % 1-sphere
\DeclareMathOperator{\borel}{\mathcal{B}}   % borel sigma algebra

\newtheorem{thm}{Theorem}
\theoremstyle{definition}
\newtheorem{defn}{Definition}

\title{Nonparametric isotropy test for spatial point processes using random rotations}

\author{Chiara~Fend\\
	Department of Mathematics\\
	University of Kaiserslautern-Landau (RPTU)\\
	Kaiserslautern, Germany \\
	\texttt{chiara.fend@rptu.de} \\
	\And Claudia~Redenbach \\
	Department of Mathematics\\
	University of Kaiserslautern-Landau (RPTU)\\
	Kaiserslautern, Germany \\
	\texttt{claudia.redenbach@rptu.de}
}
\date{}
\renewcommand{\shorttitle}{Nonparametric isotropy tests for spatial point processes}
\hypersetup{
    pdftitle={Nonparametric Isotropy Test for Spatial Point Processes using Random Rotations},
    pdfauthor={Chiara Fend, Claudia Redenbach},
    pdfkeywords={Fry plot, second-order characteristics, Monte Carlo test, resampling, spatial point processes, nonparametric tests}}

\begin{document}
\maketitle

\begin{abstract}
In spatial statistics, point processes are often assumed to be isotropic meaning that their distribution is invariant under rotations. Statistical tests for the null hypothesis of isotropy found in the literature are based either on asymptotics or on Monte Carlo simulation of a parametric null model. Here, we present a nonparametric test based on resampling the Fry points of the observed point pattern. Empirical levels and powers of the test are investigated in a simulation study for four point process models with anisotropy induced by different mechanisms. Finally, a real data set is tested for isotropy.

\end{abstract}


% %%%%%%%%%%%%%  Introduction %%%%%%%%%%%%%%%%%%%%%%%%%%%%%%%%%%%%%%%%%
\section{Introduction}
Most established models for spatial point processes assume that the point process is stationary and isotropic. While non-stationarity is often considered, formal testing of the null hypothesis of isotropy has rarely been discussed. Recently, a review of summary statistics for the directional analysis of point pattern has been published \citep{rajala_review_2018}. In particular, second-order methods were shown to be powerful tools for detecting and quantifying the anisotropy of a point pattern, see, e.g., \cite{rajala_estimating_2016,sormani_second_2020,wong_isotropy_2016,moller_cylindrical_2016}.

Functional summary statistics such as directed versions of Ripley's K-function or the nearest neighbour distance distribution function can be turned into real valued test statistics in numerous ways \citep{rajala_tests_2022}. The constructions include test statistics that remind of the classical $\chi^2$, Kolmogorov-Smirnov or Cramér-von-Mises test statistics.
However, the distribution of such test statistics under isotropy is typically not available. Hence, critical values have to be selected based on asymptotic theory \citep{guan_assessing_2006}, replicated data \citep{redenbach_anisotropy_2009} or Monte Carlo simulation of a parametric null model \citep{rajala_tests_2022}. For a general discussion of Monte Carlo tests in spatial statistics we refer to \citet[Section 10.6]{spatstat}. The rationale of this approach is based on the assumption that the observed test statistic and test statistics obtained from simulations of the null model are exchangeable under the null hypothesis \citep{barnard_monte_carlo_1963}. This means that the joint distribution of all the test statistics is invariant under permutation of its arguments \citep{davison_bootstrap_1997} which is often achieved by using i.i.d. samples from a parametric (isotropic) null model.

In our setting, the choice of a parametric null model from which we can generate isotropic point patterns is not straightforward. Even if a candidate model is available, the estimation of its parameters may be heavily affected by anisotropy of the observed point pattern. Thus, we will be interested in nonparametric resampling techniques to generate samples which can be used in the Monte Carlo setting.

This idea was already followed by \citet{wong_isotropy_2016}, who propose a resampling scheme via the stochastic reconstruction approach of \citet{tscheschel_statistical_2006}. A new point pattern is generated by iteratively adapting an arbitrary initial point pattern until some selected empirical functional summary statistics approximate the ones of the observed point pattern. 
Their isotropy test statistic 
uses the reduced second-order moment measure as functional summary statistic.  Hence, the authors propose to use the $k$th nearest-neighbor distance distribution function together with a morphological summary function called the planar convexity number for the reconstruction. Neither statistic takes directional information into account. Therefore, the reconstructed point pattern can be interpreted as an isotropic counterpart of the observed data. \citet{wong_isotropy_2016} show in a simulation study that the resulting nonparametric isotropy test is comparable in power to the asymptotic test of \citet{guan_assessing_2006} with the best parameter combination while being less sensitive with respect to the choice of tuning parameters. 

Another nonparametric approach is proposed in \citet{sormani_second_2020} under the name \emph{projection method}. The main idea is to project the Fry points with norm in a certain range to $\So$ and test the resulting sample of unit vectors for uniformity on $\So$. A selection of uniformity tests can be found in \citet{mardia_directional_1999}. The projection method works well in case of regularity but \citet{sormani_second_2020} state that the theoretical level of the test when using the Bingham test for uniformity is not met in case of clustering. 

We propose an alternative to the work by \citet{wong_isotropy_2016}, which is based on resampling the Fry points rather than the point pattern itself. 
The remainder of the paper is organized as follows.
The notation is introduced in Section~\ref{sec:notation}. In Section~\ref{sec:directional} we present methods for the directional analysis of point processes. In Section~\ref{sec:test} we define the new nonparametric isotropy tests based on random rotations. A simulation study comparing the new tests in presented in Section~\ref{sec:study}. We conclude with the application of the proposed tests to a real data set of amacrine cells (Section~\ref{sec:real-data}) and a discussion of the results (Section~\ref{sec:discussion}).

\section{Preliminaries}\label{sec:notation}

In this section we introduce the necessary notation which is needed for the rest of the paper. A spatial point process $\X$ on $\R^d$, $d >1$, is a random locally finite counting measure on $(\R^d, \borel)$ where $\borel$ denotes the Borel $\sigma$-algebra on $\mathbb{R}^d$. 
We assume that $\X$ is a simple point process which means that $\X(\{x\}) \leq 1$ almost surely for all $x\in \R^d$. In this case, the process $\X$ is often identified with its support $\{x \in \R^d \mid \X(\{x\}) > 0\} = \{X_1, \dots X_n\}$ with $n \in \N \cup \{\infty\}$ which we call a random point pattern. Throughout the paper we use the two representations interchangeably. 

The point process $\X$ is said to be stationary if its distribution is invariant under translations in $\R^d$.
If the distribution of $\X$ is invariant under rotations around the origin, then we call $\X$ isotropic otherwise anisotropic.

In this paper, we assume stationarity of the point process models.
In this case, the expected number of points in a Borel set $B$ is proportional to the Lebesgue measure of the set, i.e. \begin{equation}
\mathbb{E}\left[\X(B)\right] = \lambda \abs{B} \quad \text{for all } B \in \borel.
\end{equation} The constant $\lambda > 0$ is called the intensity of the stationary point process.

We focus on the case $d=2$, i.e. planar point processes. 
Since we are interested in directional properties it is convenient to consider polar coordinates. The $1$-sphere in $\Rt$ is denoted by $\So = \{ (x,y)^T \in \Rt \mid x^2 + y^2 = 1\}$.  A unit vector $u \in \So$ can be represented as $u = u(\alpha) = (\cos(\alpha), \sin(\alpha))^T$ with $\alpha \in [0, 2\pi)$ being the angle from the $x$-axis in counterclockwise direction. We use the notation $\So$ and its representation as interval $[0, 2\pi)$ interchangeably.

Let $b_r(x) \subset \Rt$ be the closed disk with radius $r \geq 0$ centered in $x\in \Rt$. 
A sector of the disk with radius $r$ and bounded by the two angles $\alpha-\varepsilon$ and $\alpha+\varepsilon$ with $0 \leq 2\varepsilon \leq \pi$ is denoted by $S(\alpha, \varepsilon, r)$. The infinite double spherical cone $DC(\alpha, \varepsilon)$ is defined through the central axis spanned by the unit vector $u(\alpha)$ and the half-opening angle $0 < \varepsilon \leq \nicefrac{\pi}{2}$.
We set $DS(\alpha, \varepsilon, r) = DC(\alpha, \varepsilon) \cap b_r(0)$ as the restricted double cone. Furthermore the $2$D cylinder with main axis spanned by the unit vector $u(\alpha) \in \So$,  
radius $w \geq 0$ and half-length $r$ is denoted by $Cyl(\alpha, w, r)$. Figure~\ref{fig:geometry} shows the construction of the sets.

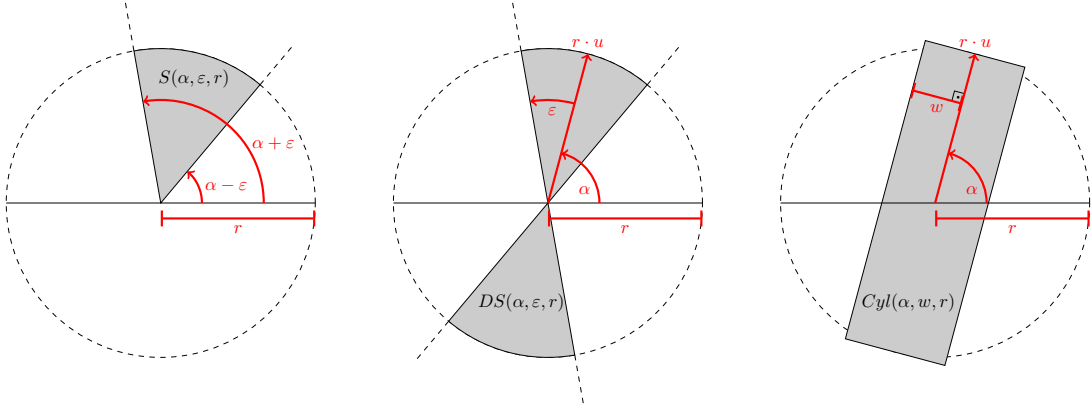
\begin{figure}
	\hspace{1cm}
 \subfloat{\resizebox{0.25\linewidth}{!}{\begin{tikzpicture}
\draw [dashed] circle(3);
\draw [dashed](50:3)-- (50:4);
\draw [dashed](100:3)-- (100:4);
\draw [white](280:3)-- (280:4);
\draw [fill=gray!40](50:3)--(0,0)--(100:3) arc (100:50:3)--cycle;
\draw (-3,0) -- (3,0);
\draw [very thick, ->, red]  (0:0.8)   arc (0:50:0.8)   node[right, pos=.5]{$\alpha-\varepsilon$};
\draw [very thick, ->, red]  (0:2)   arc (0:100:2)   node[right, pos=0.35]{$\alpha +\varepsilon$};
\draw [red, very thick,|-|](0,-0.3) -- (3,-0.3) node[below, pos=.5]{$r$};
\node[very thick]  at (75:2.5) {$S(\alpha, \varepsilon, r)$};
	\end{tikzpicture}}}\hfill
 \subfloat{\resizebox{0.25\linewidth}{!}{\begin{tikzpicture}
\draw [dashed] circle(3);
\draw [dashed](50:3)-- (50:4);
\draw [dashed](100:3)-- (100:4);
\draw [dashed](230:3)-- (230:4);
\draw [dashed](280:3)-- (280:4);
\draw [fill=gray!40](50:3)--(0,0)--(100:3) arc (100:50:3)--cycle;
\draw [fill=gray!40](230:3)--(0,0)--(280:3) arc (280:230:3)--cycle;
\draw (-3,0) -- (3,0);
\draw [very thick, ->, red](0,0)-- (75:3) node[above]{$r\cdot u$};
\draw [very thick, ->, red]  (0:1)   arc (0:75:1)   node[left, pos=.2]{$\alpha$};
\draw [very thick, ->, red]  (75:2)   arc (75:100:2)   node[below, pos=0.5]{$\varepsilon$};
\draw [red, very thick,|-|](0,-0.3) -- (3,-0.3) node[below, pos=.5]{$r$};
\node[very thick]  at (255:2) {$DS(\alpha, \varepsilon, r)$};
	\end{tikzpicture}}}\hfill
\subfloat{\resizebox{0.25\linewidth}{!}{\begin{tikzpicture}
\draw [white] (100:3)-- (100:4);
\draw [white](280:3)-- (280:4);
	\draw [dashed] circle(3);
    \draw [fill=gray!40] ({3*cos(75)-sin(75)}, {3*sin(75)+cos(75)}) -- ({3*cos(75)+sin(75)}, {3*sin(75)-cos(75)}) -- ({-3*cos(75)+sin(75)}, {-3*sin(75)-cos(75)}) -- ({-3*cos(75)-sin(75)}, {-3*sin(75)+cos(75)}) -- cycle;
    \draw (-3,0) -- (3,0);
    \draw ({2*cos(75)-0.2*sin(75)}, {2*sin(75)+0.2*cos(75)}) -- ({2.2*cos(75)-0.2*sin(75)},{2.2*sin(75)+0.2*cos(75)}) -- ({2.2*cos(75)},{2.2*sin(75)});
    \fill ({2.1*cos(75)-0.1*sin(75)}, {2.1*sin(75)+0.1*cos(75)}) circle(0.75pt);
    \draw [very thick, ->, red](0,0)-- (75:3) node[above]{$r\cdot u$};
    \draw [very thick, ->, red]  (0:1)   arc (0:75:1)   node[left, pos=.2]{$\alpha$};
    \draw [very thick, |-|, red]  (75:2)  -- ({acos(2/sqrt(5))+75}:{sqrt(5)})  node[below, pos=.5]{$w$};
    \draw [red, very thick,|-|](0,-0.3) -- (3,-0.3) node[below, pos=.5]{$r$};
    \node[very thick]  at (255:2) {$Cyl(\alpha, w, r)$};
\end{tikzpicture}}}\hspace{1cm}
\caption{Representation of the directed sets used in the directional analysis of planar point patterns. From left to right: Sector $S(\alpha, \varepsilon, r)$, restricted double cone $DS(\alpha, \varepsilon, r)$, 2D cylinder $Cyl(\alpha, w, r)$.}
\label{fig:geometry}
\end{figure}

The translation of a set $A \in \borel$ by a vector $z \in \Rt$ is defined as $A_z = \{ a+z \mid a \in A\}$ and the indicator function of the set is denoted as $\1{ \cdot \in A}$.

% %%%%%%%%%%%%%  Directional methods %%%%%%%%%%%%%%%%%%%%%%%%%%%%%%%%%%
\section{Directional methods for point processes}\label{sec:directional}

We refer to \citet{rajala_review_2018} for a detailed review on available methods for the directional analysis of point process. We focus on using the second-order structure of the process since a previous study \citep{rajala_tests_2022} has shown, that parametric isotropy tests based on second-order characteristics generally perform well for many types of isotropic null models. Other approaches investigated there were not as robust. For instance, directional nearest neighbour methods showed a poor performance when used for clustered point patterns.

In the following sections we let $\x = \X \cap W = \{x_1, \dots, x_n\}$ with $2 \leq n < \infty$ denote the observed point pattern in the observation window $W$ with $0 < \abs{W} < \infty$.

\subsection{Fry points and Fry plot}
Anisotropy can be hard to detect by pure visual inspection of a point pattern $\x$. In some cases, it is easier to see when inspecting the Fry points $F_{\x}$ of $\x$ as proposed in \citet{fry_random_1979}. 
The multiset $F_{\x}$ contains the pairwise difference vectors between all individual points in $\x$, that is, \begin{equation}
    F_{\x} = \{ x_i - x_j \mid x_i, x_j \in \x, \, i \neq j\}.
\end{equation}
By construction, $F_{\x}$ is point symmetric around the origin. The Fry plot displays all Fry points as a scatter plot, see Figure~\ref{fig:fry}. For regular point processes, we generally observe a void around the origin. If the process is isotropic, then the void resembles a disc. Hence, deviations of the void shape from a disc may indicate anisotropy. 
For clustered point processes, many Fry points are found close to the origin and a possible anisotropy becomes visible through the shape of this cluster.

\begin{figure}[th]
\captionsetup[subfloat]{aboveskip=0.25pt}
	\subfloat[isotropic and regular]{\includegraphics[width=0.245\textwidth]{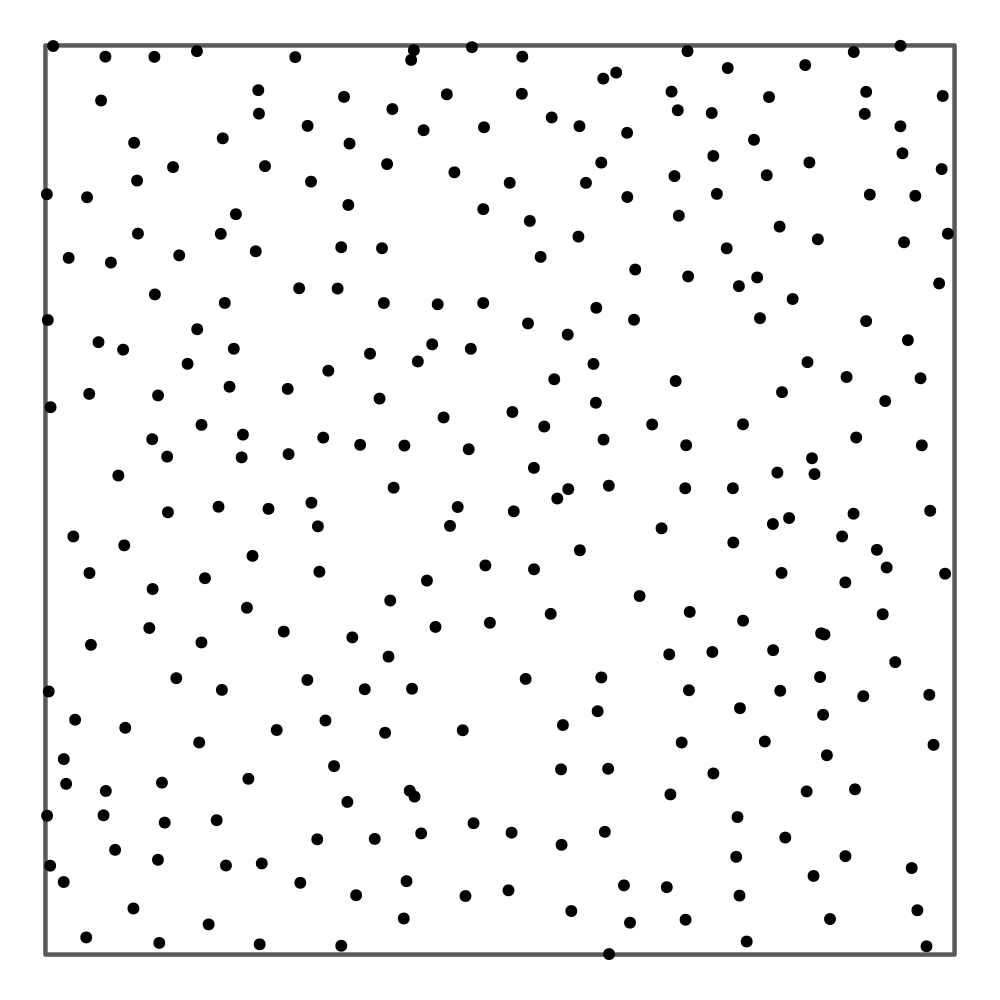}\hfill
 \includegraphics[width=0.245\textwidth]{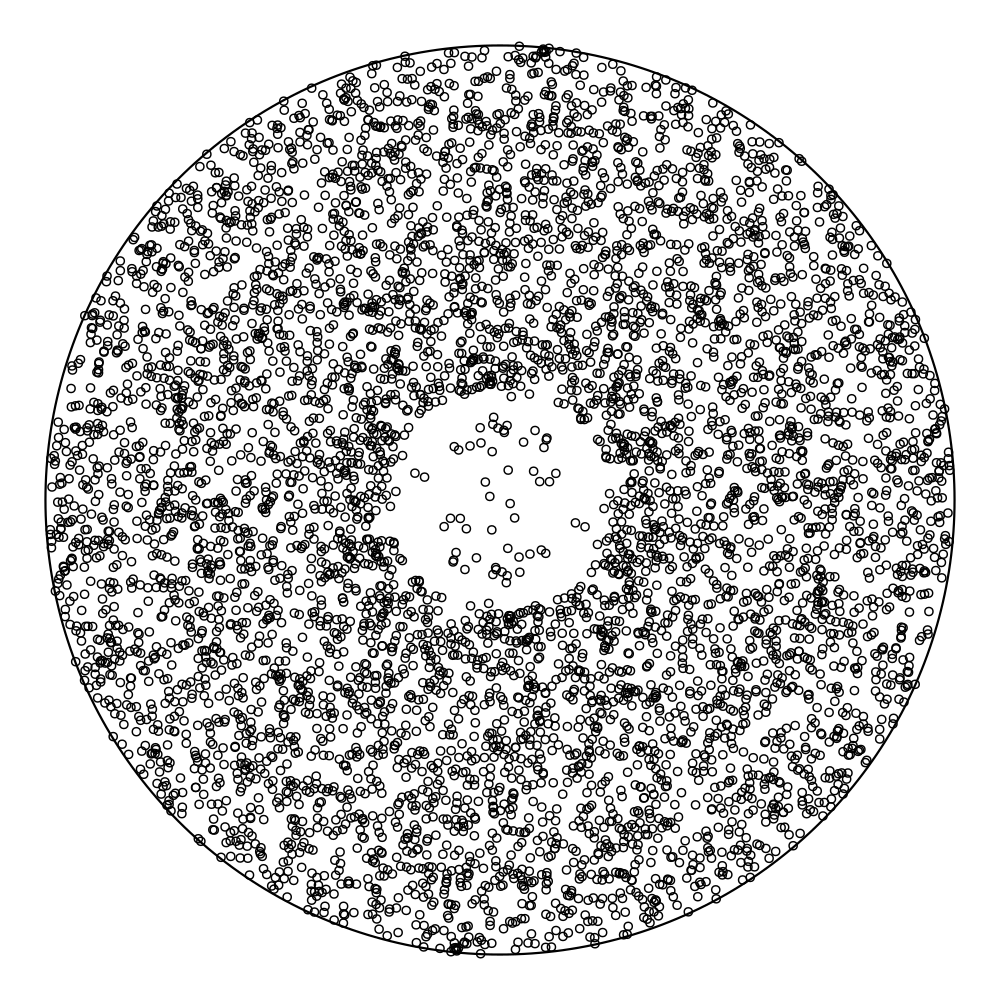}} \hfill
	\subfloat[isotropic and clustered]{\includegraphics[width=0.245\textwidth]{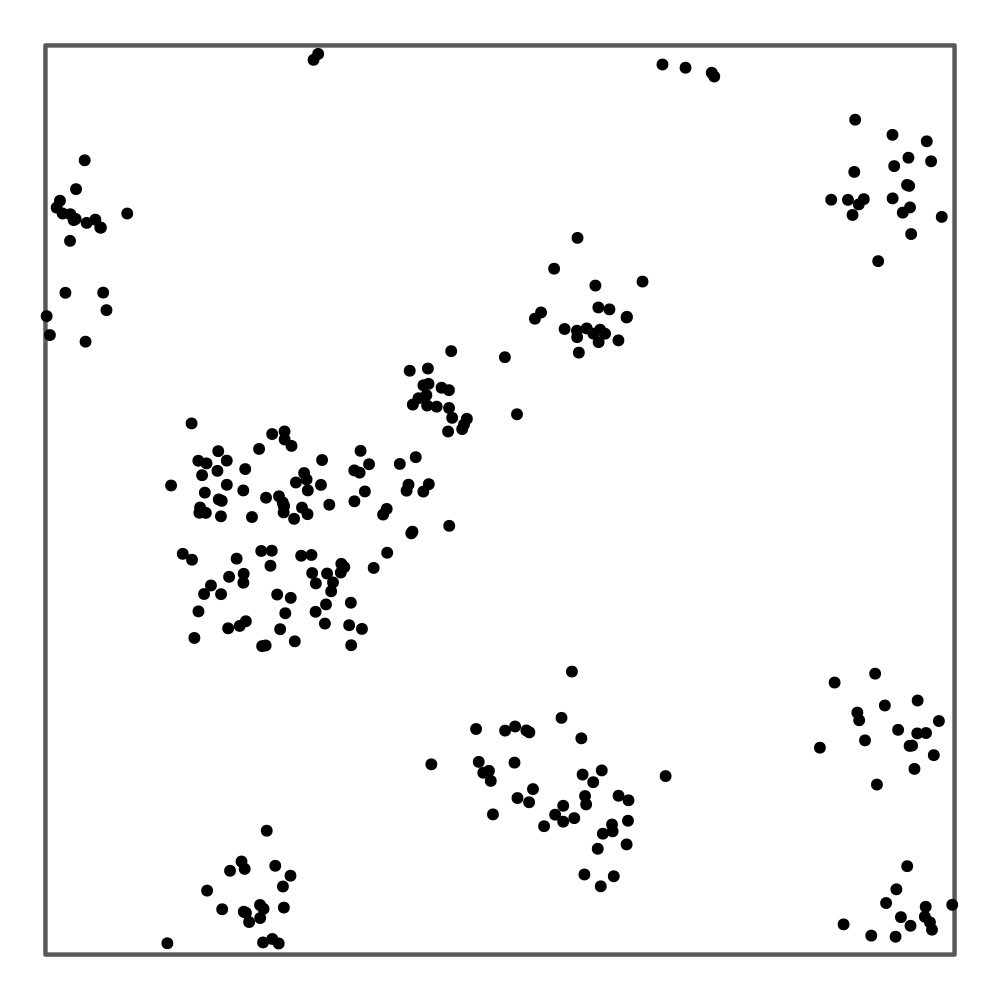}\hfill\includegraphics[width=0.245\textwidth]{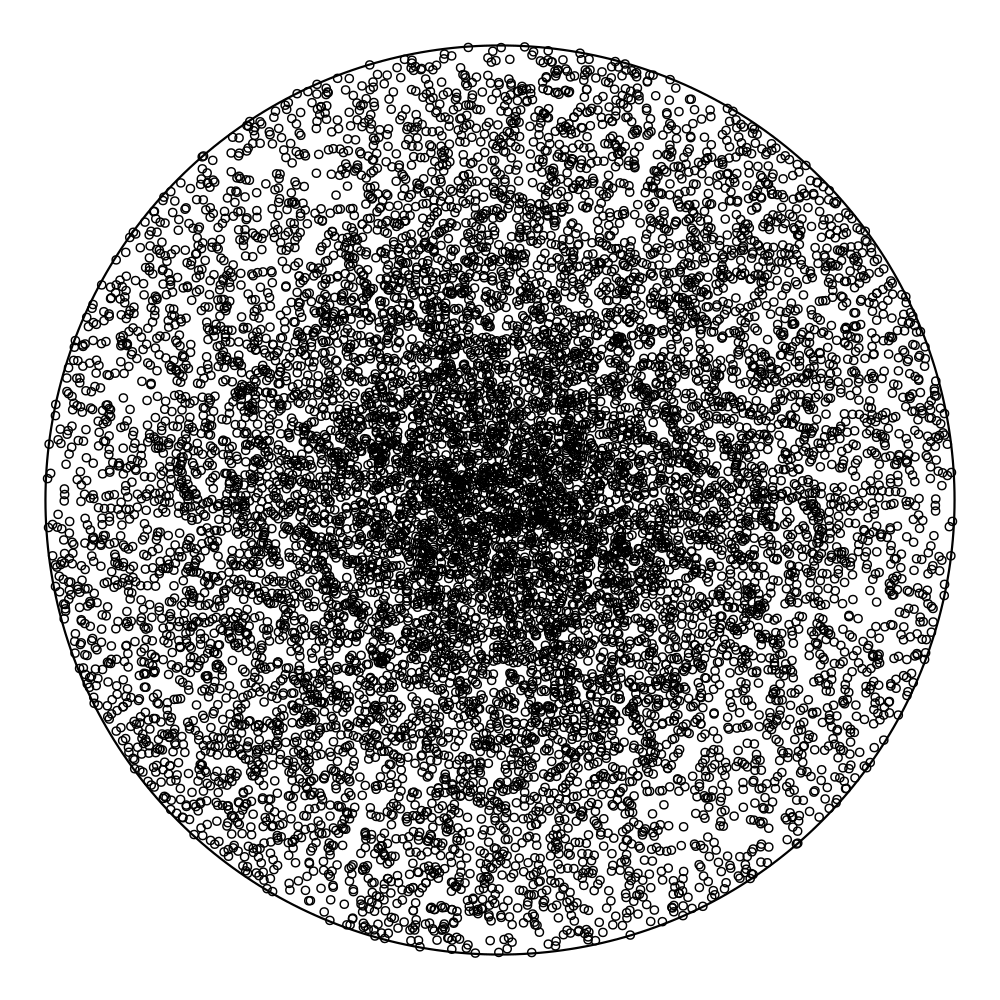}}\\
    \subfloat[anisotropic and regular]{\includegraphics[width=0.245\textwidth]{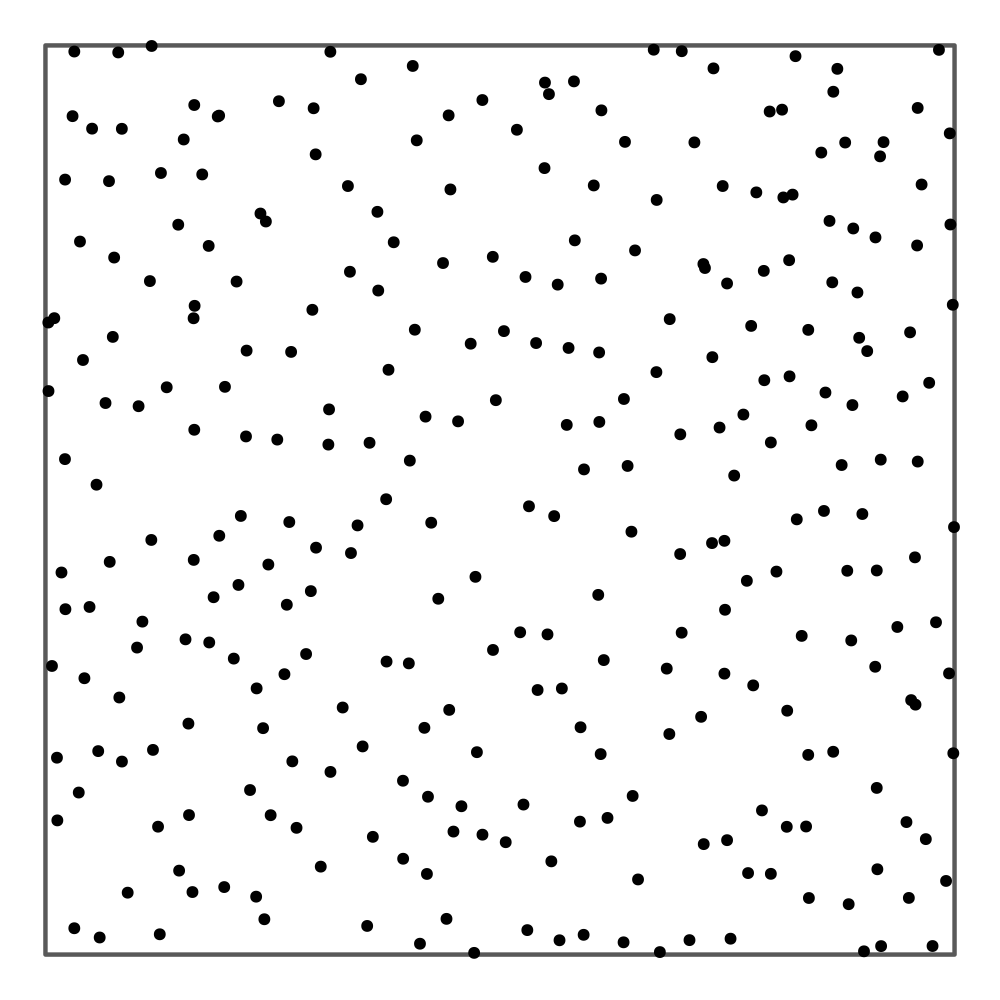}\hfill\includegraphics[width=0.245\textwidth]{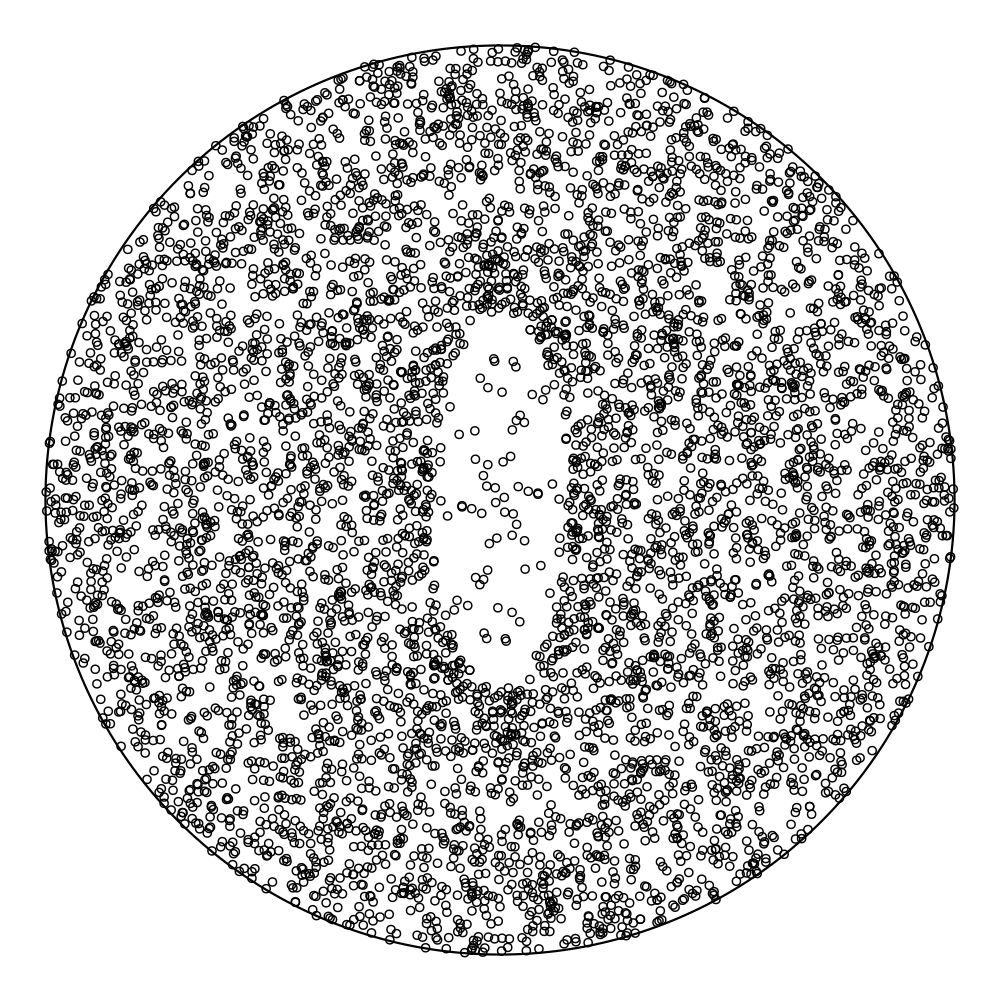}} \hfill
	\subfloat[anisotropic and clustered]{\includegraphics[width=0.245\textwidth]{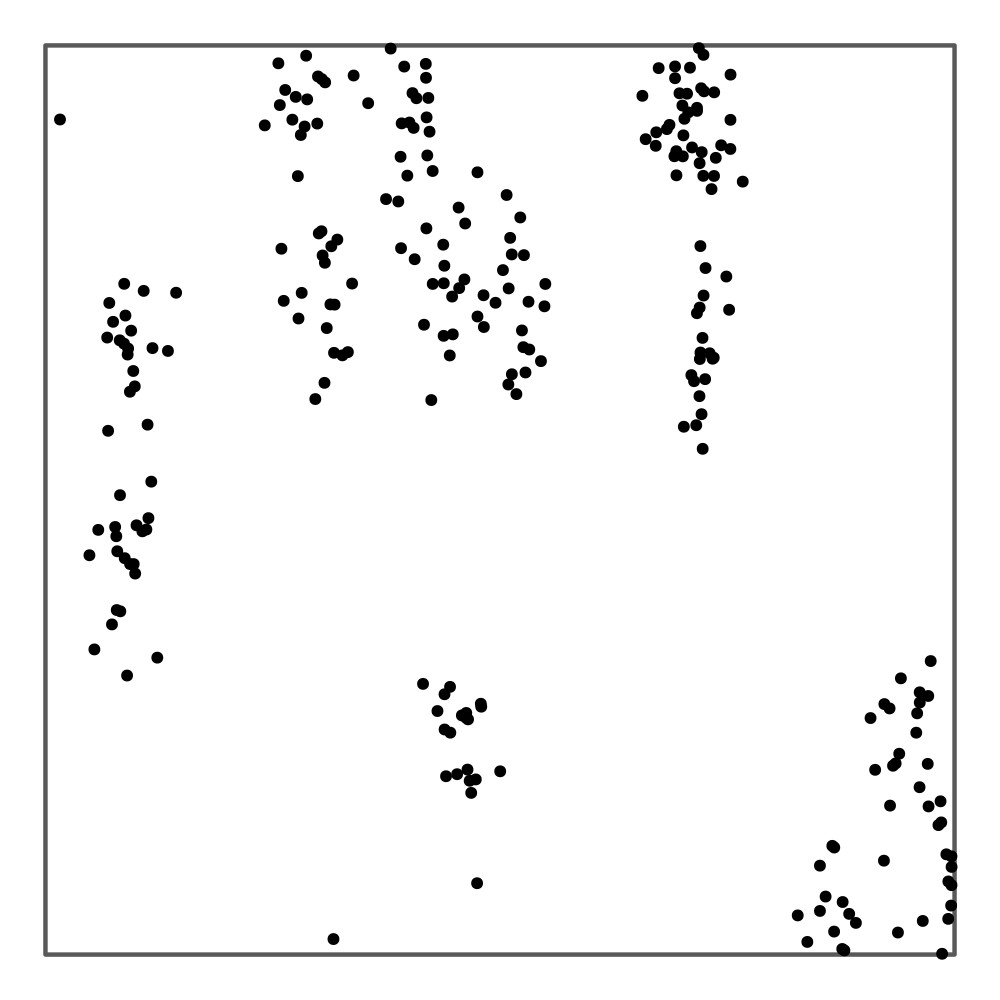}\hfill\includegraphics[width=0.245\textwidth]{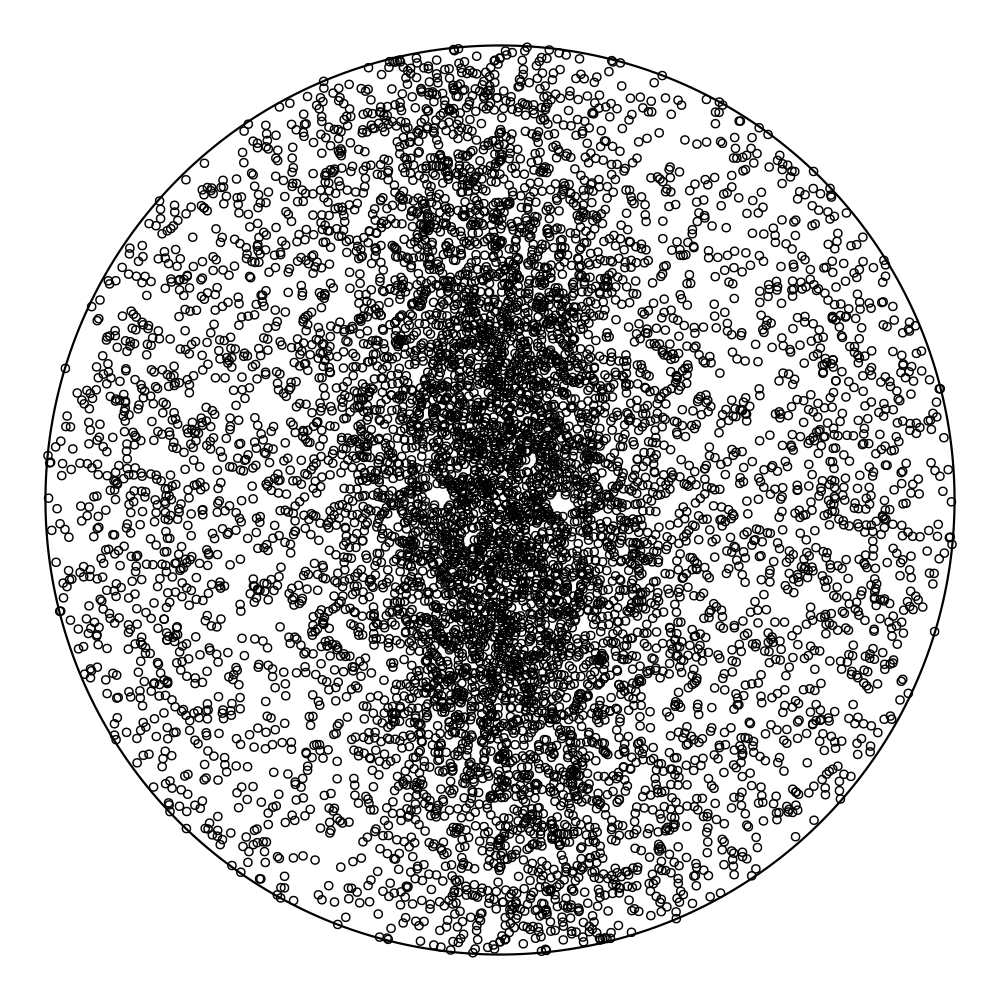}}
	\caption{Examples of point patterns and the corresponding Fry plots. The stationary processes have intensity $\lambda = 50$ and are shown in the observation window $[-\sqrt{6}/2, \sqrt{6}/2]^2$ with each pattern containing approx. $300$ points. The Fry plots show all Fry points with norm up to $0.4$.}
	\label{fig:fry}
\end{figure}

The Fry points are also directly related to the second-order characteristics of a stationary point process which describe the interactions of pairs of points.

\subsection{Directional K-functions}

The reduced second-order moment measure $\mathcal{K}$ of a stationary point process $\X$ with intensity $\lambda$ is defined as
\begin{equation}
\mathcal{K}(B) = \frac{1}{\abs{A}}\mathbb{E}\left[ \sum_{x,y \in \X}^{\neq} \frac{\1{x\in A, x-y \in B}}{\lambda^2}\right] \qquad \text{for} \qquad B \in \borel.
\end{equation}
It does not depend on the choice of the test set $A \in \borel$ with $0 < \abs{A} < \infty$. The $\neq$ sign above the summation indicates that only pairs with $x\neq y$ are considered.
If the point process $\X$ is stationary, we can interpret $\lambda \mathcal{K}(B)$ as the expected number of points in the Borel set $B$ conditional on $\X$ having a point at the origin, not counting the origin if it falls in $B$.

The measure $\mathcal{K}$ is used to define many functional summary statistics. Choosing the Borel set as the ball with radius $r \geq 0$ around the origin, i.e., $B = b_r(o)$, yields Ripley's $K$-function. In the directional analysis of point processes, $B$ is replaced by some directed set. In particular, \citet{redenbach_anisotropy_2009} choose the double cone $DS(\alpha, \varepsilon, r)$ while \citet{moller_cylindrical_2016} use
a directed cylinder $Cyl(\alpha, w, r)$ and \citet{wong_isotropy_2016} the sector $S(\alpha, \varepsilon, r)$.

In this paper we use $B=S(\alpha, \varepsilon, r)$ which defines the sector $K$-function $K_{\textrm{sect}, \alpha, \varepsilon}$ as \begin{equation}
    K_{\textrm{sect}, \alpha, \varepsilon}(r) = \mathcal{K}(S(\alpha, \varepsilon, r)).
\end{equation}

The different types of directional $K$-functions can be estimated from an observed pattern $\x$ in the observation window $W$, $0 < \abs{W} < \infty$, using nonparametric estimators. These estimators account for edge-effects due to the bounded observation window. One such estimator of $\mathcal{K}(B)$ for a stationary process is given through translational edge-correction as 
\begin{equation}\label{eq:K-estimator}
\widehat{\mathcal{K}}(B) = \frac{1}{\widehat{\lambda^2}} \sum_{x,y \in \x}^{\neq} \frac{\1{x-y \in B}}{\abs{W \cap W_{x-y}}} \qquad \text{with} \qquad \widehat{\lambda^2} = \frac{n(n -1)}{\abs{W}^2}.
\end{equation}

The squared intensity estimate $\widehat{\lambda^2}$ is unbiased in case of a Poisson process \citep{moller_statistical_2003}. The combined estimate $\widehat{\mathcal{K}}(B)$ is then ratio-unbiased.

The estimator \eqref{eq:K-estimator} only depends on the difference vectors $x-y$. Hence, it can also be formulated in terms of the Fry points, namely
\begin{equation}
    \widehat{\mathcal{K}}(B) = \frac{1}{\widehat{\lambda^2}} \sum_{z \in F_{\x}} \frac{\1{z \in B}}{\abs{W \cap W_{z}}}.
\end{equation}
The number of Fry points %$\sum_{z \in F_{\x}} \1{z \in B}$ 
in the set $B$  is consequently an unnormalized and uncorrected estimate of the reduced second-order moment measure. 

\subsection{Contrast summary statistics}\label{sec:sumstat}

Under isotropy, estimates of directional summary statistics for different directions should be similar. Hence, anisotropy can be detected by contrasting estimates for two (mostly orthogonal) directions. 
Contrasts can either be in form of pointwise differences or in form of pointwise ratios. 

In the first case of pointwise differences we consider the same shape of the directed set but with two different central axes. Under isotropy, we expect that the differences are small. Thus, large, in absolute terms, contrast values consequently indicate anisotropy. In case of the sector $K$-function, we will denote the functional contrast statistic by \begin{equation}
    K^c_{\mathrm{sect, \alpha_1, \alpha_2, \varepsilon}}(r) = K_{\textrm{sect}, \alpha_1, \varepsilon}(r) - K_{\textrm{sect}, \alpha_2, \varepsilon}(r).
\end{equation}

Alternatively, \citet{wong_isotropy_2016} fix a radius $r \geq 0$ and a main direction given by $\alpha \in [0, \pi]$ and compute the ratio \begin{equation}
    F_{r, \alpha}(\varepsilon) = \frac{\mathcal{K}(S(\alpha, \varepsilon, r))}{\mathcal{K}(S(\frac{\pi}{2}, \frac{\pi}{2}, r))} 
    %= \frac{\mathcal{K}(DS(\alpha, \varepsilon, r))/2}{\mathcal{K}(DS(\frac{\pi}{2}, \frac{\pi}{2}, r))/2} 
    = \frac{K_{\textrm{sect}, \alpha, \varepsilon}(r)}{K_{\textrm{sect}, \frac{\pi}{2}, \frac{\pi}{2}}(r)}, \quad \varepsilon \in [0, \frac{\pi}{2}],
\end{equation}
which results in comparing the value of the reduced second-order moment measure for a sector with the one of the entire ball, which is Ripley's $K$-function. Under isotropy, $F_{r, \alpha}$ is the distribution function of the uniform distribution which we will show in Theorem~\ref{thm:fraction-cdf}. 

\subsection{Second-order properties under isotropy and stationarity}

The reduced second-order moment measure $\mathcal{K}$ of the stationary point process $\X$ was defined through a normalization by the squared intensity $\lambda^2$. Alternatively, one could also investigate the reduced measure $\bar{\mu}^{(2)}$ on the difference vectors defined as 
\begin{equation}
    \bar{\mu}^{(2)}(B) = \lambda^2 \mathcal{K}(B) = \mathbb{E}\left[ \sum_{x,y \in \X}^{\neq} \1{x \in [0,1]^2, \,y-x \in B} \right] \quad \text{for } B \in \borel.
\end{equation}

Note, that also any other test set with unit volume can be chosen instead of $[0,1]^2$ in the expectation.

If we parameterize the difference vector $z = y-x \in \R^d$ using its polar coordinates $z = r (\sin(\theta), \cos(\theta))$, we obtain the reduced measure $\bar{\mu}^{(2)}$ as a measure on the Borel $\sigma$-algebra of $\R_{\geq 0} \times \So$. \citet{daley_intro_2007} use a normalization called the second-order directional rose to investigate if the process is isotropic or not.

\begin{defn}
Let $0 \leq r < \infty$. The second-order directional rose $\Gamma_2(\cdot \mid r)$ of a stationary planar point process $\X$ is a probability measure on $\So$ which is defined as \begin{equation}
    \Gamma_2(\mathrm{d}\theta \mid r) = \frac{\bar{\mu}^{(2)}(\mathrm{d}r \times \mathrm{d}\theta)}{\int_0^{2\pi} \bar{\mu}^{(2)}(\mathrm{d}r \times \mathrm{d}\theta)}.
\end{equation}
It can be interpreted as the probability that the difference vector of two points of $\X$ at distance $r$ points in a direction in the interval $(\theta, \theta + \mathrm{d}\theta)$.
\end{defn}

In case of isotropy one can show the following result.

\begin{thm}{\citep[Prop.~15.2.II]{daley_intro_2007}}\label{thm:uniform}
If the planar point process $\X$ is stationary and isotropic, then for all $0 \leq r < \infty$ and $\theta \in [0, 2\pi]$\begin{equation}
    \Gamma_2(\mathrm{d}\theta \mid r) = (2\pi)^{-1} \mathrm{d}\theta.
\end{equation}
In other words, the second-order directional rose reduces to the uniform distribution on $\So$.
\end{thm}

It is possible to express the directional rose in terms of the measure $\mathcal{K}$, again in polar coordinates, as \begin{equation}
    \Gamma_2(\mathrm{d}\theta \mid r) = \frac{\mathcal{K}(\mathrm{d}r \times \mathrm{d}\theta)}{\int_0^{2\pi} \mathcal{K}(\mathrm{d}r \times \mathrm{d}\theta)}. \label{eq:K-rose}
\end{equation}
Under isotropy, this implies that the angle of a Fry point $z = y-x$ with norm $r > 0$ is uniformly distributed on $[0, 2\pi)$.
This property will be the main ingredient of our resampling approach.

Additionally, these results can be used to prove that under isotropy the functional summary statistic $F_{r, \alpha}$ that appears in the test statistic of \citet{wong_isotropy_2016} indeed coincides with the distribution function of the uniform distribution on $[0, \frac{\pi}{2}]$.

\begin{thm}\label{thm:fraction-cdf}
    Let $\X$ be a stationary and isotropic planar point process, then for all $r \geq 0$, $\alpha \in [0, 2\pi)$ and $\varepsilon \in [0, \frac{\pi}{2}]$ it holds that \begin{equation}
        F_{r, \alpha}(\varepsilon) = \frac{\mathcal{K}(S(\alpha, \varepsilon, r))}{\mathcal{K}(S(\frac{\pi}{2}, \frac{\pi}{2}, r))} = \frac{\varepsilon}{\nicefrac{\pi}{2}}.
    \end{equation}
\end{thm}

\begin{proof}
    Consider first $0 \leq r_1 < r_2$ and $0 \leq \theta_1 < \theta_2 \leq 2\pi$. Let $\mathcal{T}(\mathrm{d}r) = \int_0^{2\pi} \mathcal{K}(\mathrm{d}r \times \mathrm{d}\theta)$ denote the radial measure appearing in the denominator of the directional rose. 
    We can compute \begin{align}
        \mathcal{K}([r_1, r_2] \times [\theta_1, \theta_2]) &= \int_{r_1}^{r_2}\int_{\theta_1}^{\theta_2} \mathcal{K}(\mathrm{d}r \times \mathrm{d}\theta) 
        \stackrel{\mathclap{\eqref{eq:K-rose}}}{=} \int_{r_1}^{r_2}\int_{\theta_1}^{\theta_2} \Gamma_2(\mathrm{d}\theta \mid r) \cdot \mathcal{T}(\mathrm{d}r)  
        \\ &\stackrel{\mathclap{\text{Thm.}~\ref{thm:uniform}}}{=} \int_{r_1}^{r_2}\int_{\theta_1}^{\theta_2} \frac{1}{2\pi}  \mathrm{d}\theta \, \mathcal{T}(\mathrm{d}r)
        = \int_{r_1}^{r_2} \frac{\theta_2 - \theta_1}{2\pi} \mathcal{T}(\mathrm{d}r)\\ &= \frac{\theta_2 - \theta_1}{2\pi} \int_{r_1}^{r_2}\int_0^{2\pi} \mathcal{K}(\mathrm{d}r \times \mathrm{d}\theta) = \frac{\theta_2 - \theta_1}{2\pi} \mathcal{K}([r_1, r_2] \times [0, 2\pi]).
    \end{align}
    The sector $S(\alpha, \varepsilon, r)$ written in polar coordinates results is the rectangular set $[0,r] \times [\alpha-\varepsilon, \alpha+\varepsilon]$. Consequently we obtain \begin{align}
        F_{r, \alpha}(\varepsilon) = \frac{\mathcal{K}(S(\alpha, \varepsilon, r))}{\mathcal{K}(S(\frac{\pi}{2}, \frac{\pi}{2}, r))} = \frac{\frac{2\varepsilon}{2\pi} \mathcal{K}([0, r] \times [0, 2\pi])}{\frac{\pi}{2\pi} \mathcal{K}([0, r] \times [0, 2\pi])} = \frac{\varepsilon}{\nicefrac{\pi}{2}}.
    \end{align}
\end{proof}

% %%%%%%%%%%%%%  Nonparametric anisotropy tests %%%%%%%%%%%%%%%%%%%%%%%
\section{Nonparametric isotropy tests}\label{sec:test}
In this section we propose an isotropy test based on isotropic resampling of the Fry points. 

\subsection{Random rotation of Fry points} 
We aim at resampling Fry-type point patterns that reflect a similar second-order structure as the observed point pattern while being rotational invariant. We will achieve this through random rotations of the Fry points. In the following, we will refer to the resampled patterns as \emph{bootstrap samples} even though we do not purely resample from the observed Fry points as in classical bootstrap. 

The main theoretical ingredient of our approach is that the sum of independent uniformly distributed random variables on $\So$ again follows a uniform distribution on $\So$ \citep{mardia_directional_1999}. 
In theory, Theorem~\ref{thm:uniform} states that under isotropy the angle distribution of the Fry points having any fixed norm $r$ coincides with the uniform distribution. Hence, we can add independent uniformly distributed angles, i.e. rotate the Fry point around the origin, without changing the distribution.

There are three possibilities to draw random uniformly distributed rotations. 

The most straightforward choice is to assign an independent rotation to each individual Fry point. Clearly, this destroys the symmetry of the Fry points. To keep this symmetry, each pair of Fry points, that is $z = x_i - x_j$
and $-z = x_j - x_i$ for $x_i \neq x_j \in \x$, can be assigned the same rotation. 

A third option is what we call the group-wise rotation which can be derived from the way \citet{fry_random_1979} originally created the Fry plot. The group $G_i$ hereby corresponds to all Fry points that have the same origin $x_i \in \x$, i.e.\begin{equation}
    G_i = \{ z = x_j - x_i \mid z \in F_{\x}\}. 
\end{equation}
Each group $G_i$ is rotated by the same random angle. 
This procedure needs the fewest random samples of the rotation but keeps the most of the structure of the point pattern. As for the individual rotations, the resulting bootstrap Fry points are no longer symmetric around the origin.

Note, that one could also theoretically consider rotating the entire Fry plot by the same random rotation. However, this approach does not create any new second-order structure. In particular, we do not obtain samples that are rotational invariant and thus resemble the null hypothesis.

\begin{figure}[th]
    \captionsetup[subfloat]{aboveskip=0.25pt}
    \begin{minipage}{0.3\textwidth}
        \subfloat[Observed point pattern. The numbers give the index of the respective point. Highlighted are the group $G_3$ in orange, and the pairs $x_{8}$, $x_{10}$ and $x_4$, $x_5$ in purple and cyan, respectively.]{\includegraphics[width=\textwidth]{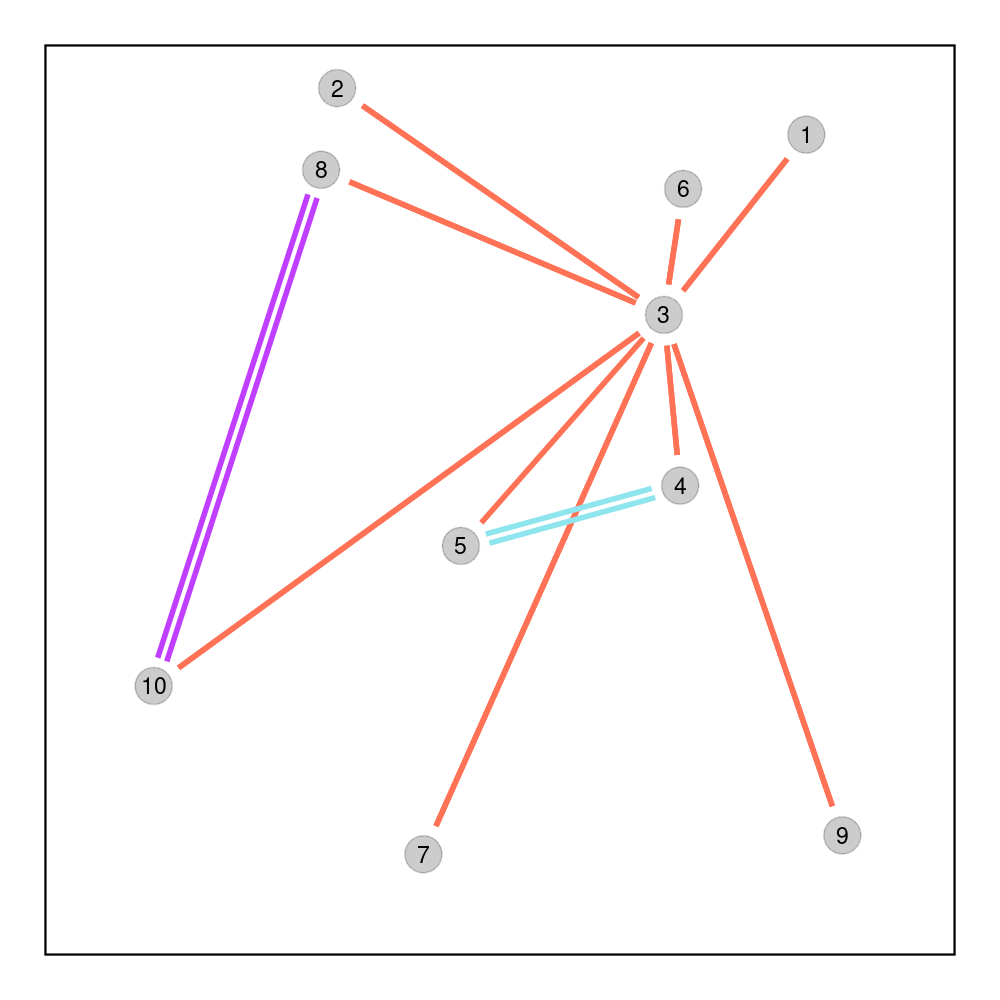}} 
    \end{minipage}
	\hfill
 \begin{minipage}{0.65\textwidth}
     \subfloat[Random group rotation of the group $G_3$ having point $x_3$ as the origin. Observed Fry points on the left, rotated Fry points on the right.]{\includegraphics[width=0.45\textwidth]{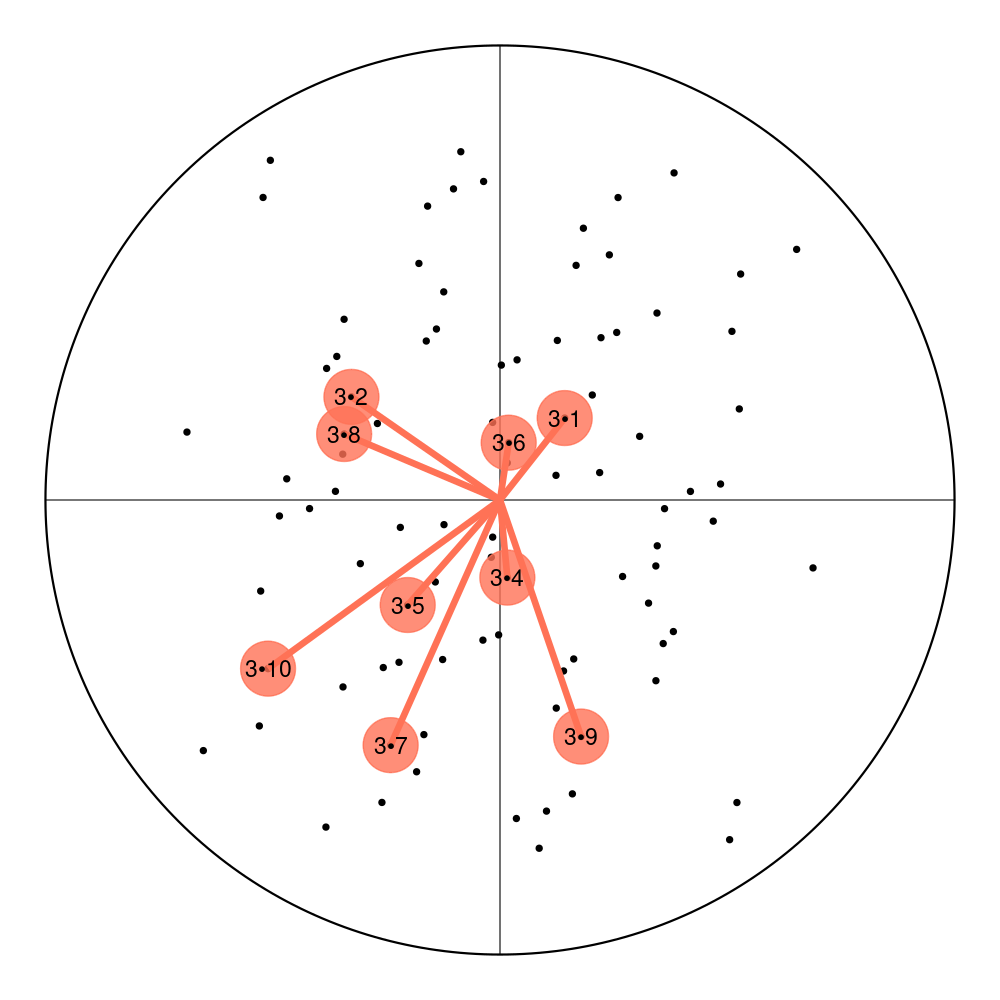}\hfill
 \includegraphics[width=0.45\textwidth]{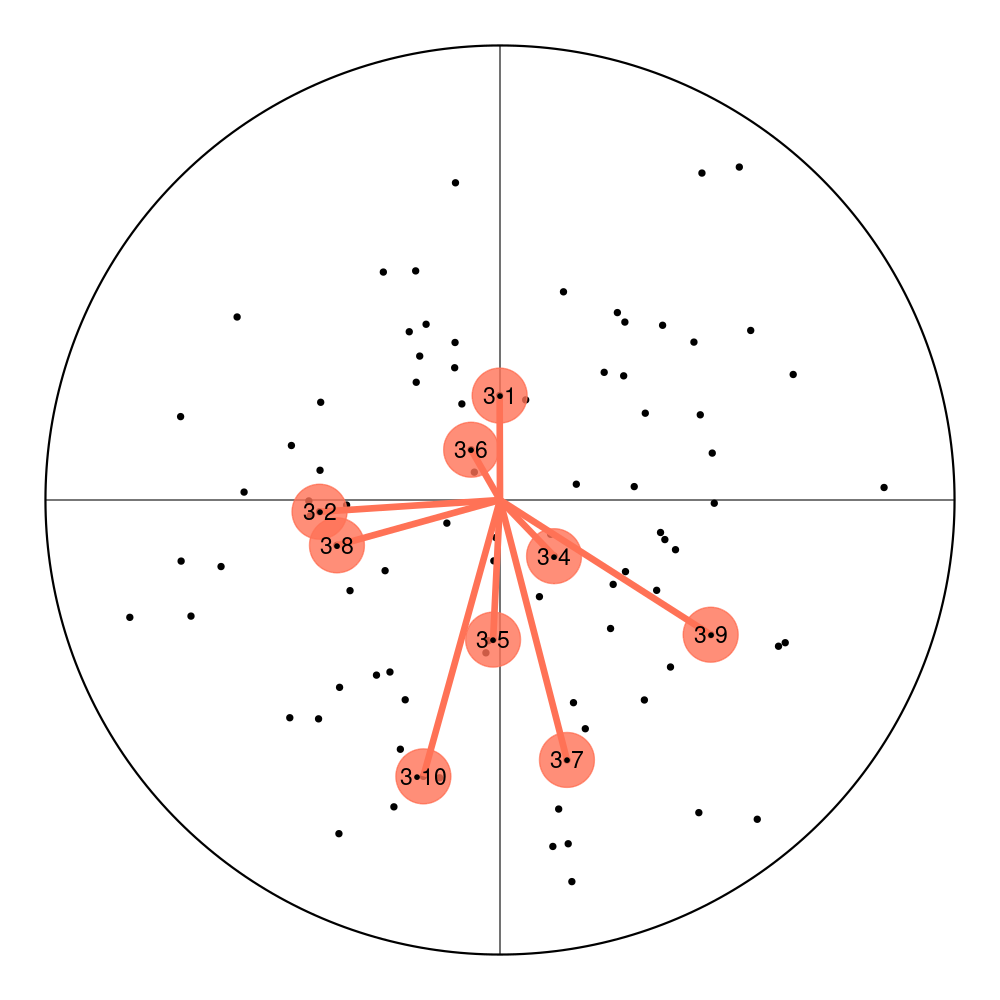}}\\ 
 \subfloat[Random pair rotation highlighted for the pairs $x_{8}$, $x_{10}$ and $x_4$, $x_5$. Observed Fry points on the left, rotated Fry points on the right.]{\includegraphics[width=0.45\textwidth]{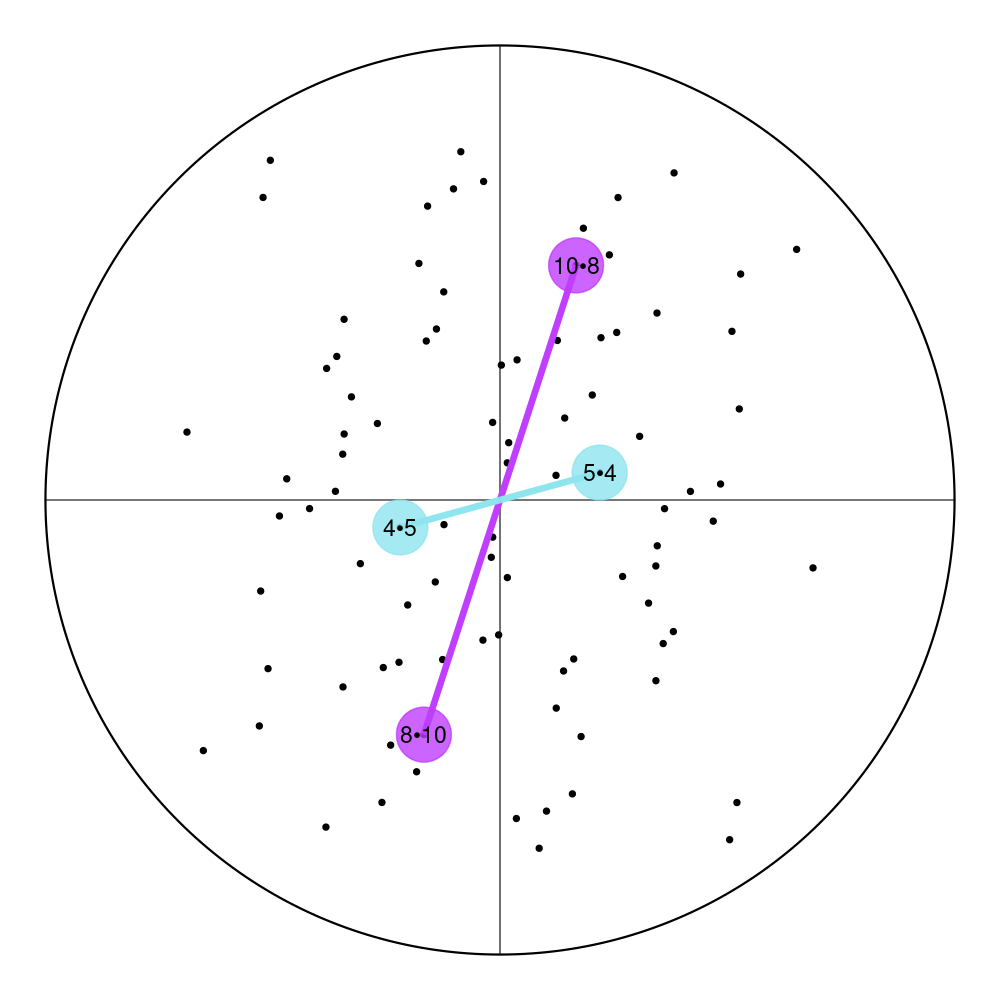}\hfill
 \includegraphics[width=0.45\textwidth]{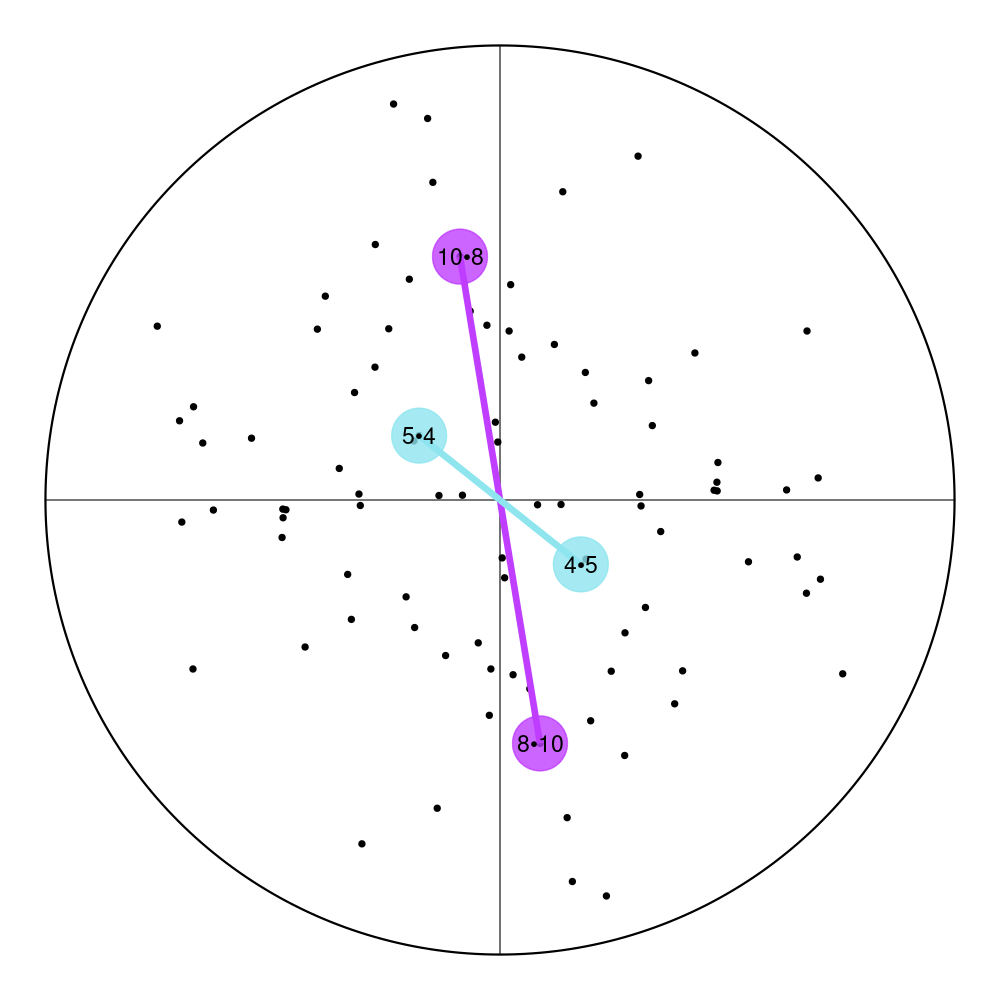}}
 \end{minipage}
	\caption{Visualization of the different random rotations}
	\label{fig:constr-rotations}
\end{figure}

Figure~\ref{fig:constr-rotations} illustrates the constructions while Figure~\ref{fig:fry-group-rot} shows some bootstrap samples obtained using the group-wise rotation for both an isotropic regular process and an anisotropic regular process. Additional examples for other types of processes can be found in the Appendix in Figure~\ref{fig:fry-group-rot-cluster} and Figure~\ref{fig:fry-group-rot-line}.

\begin{figure}[th]
    \captionsetup[subfloat]{aboveskip=0.25pt}
	\subfloat[Observed Fry plot]{\includegraphics[width=0.245\textwidth]{img/fry-points/iso-regular-fry.png}} \hfill
	\subfloat[]{\includegraphics[width=0.245\textwidth]{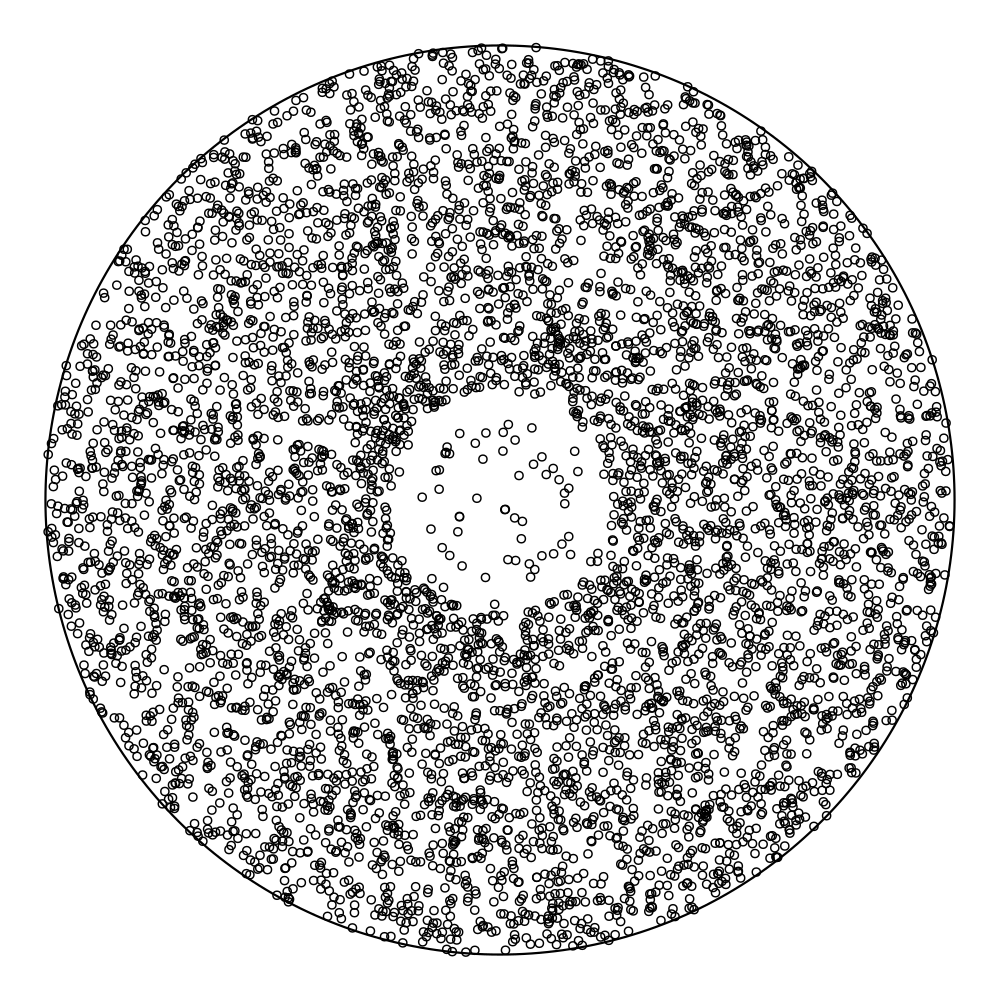}}\hfill 
	\subfloat[]{\includegraphics[width=0.245\textwidth]{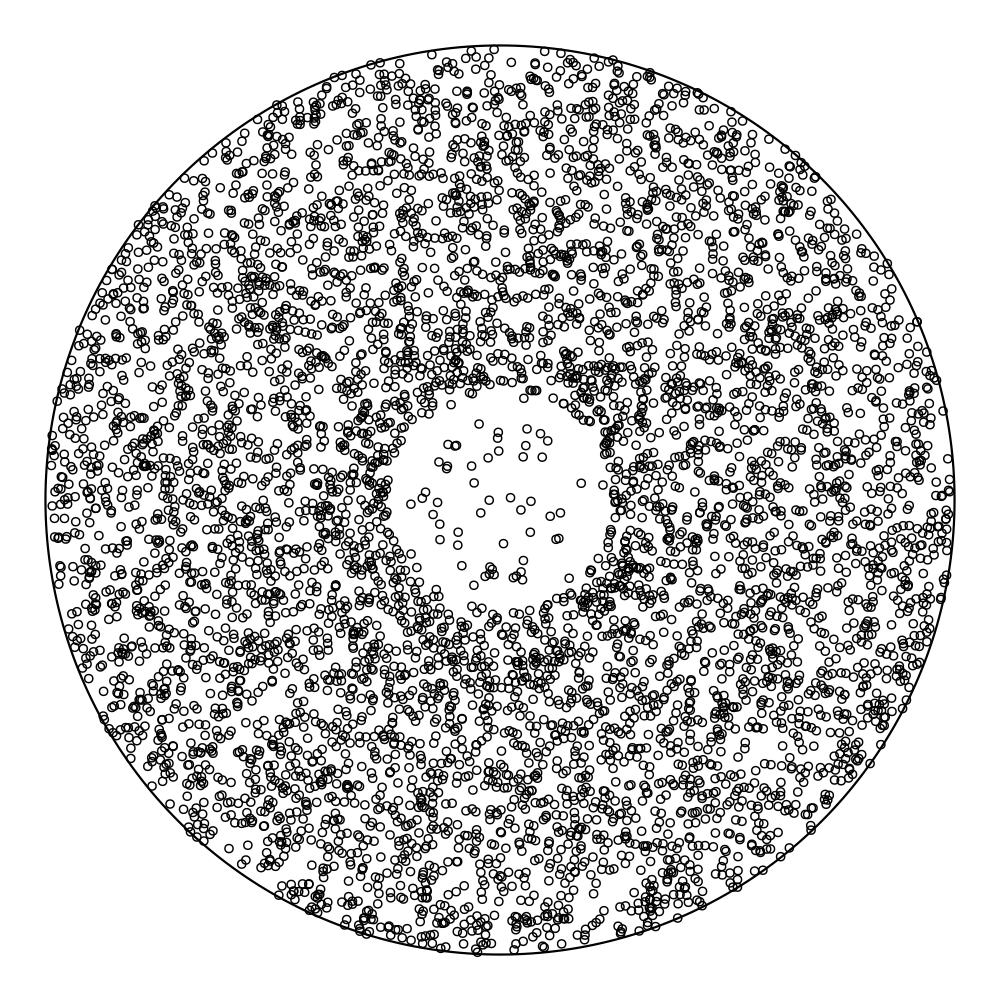}}\hfill
	\subfloat[]{\includegraphics[width=0.245\textwidth]{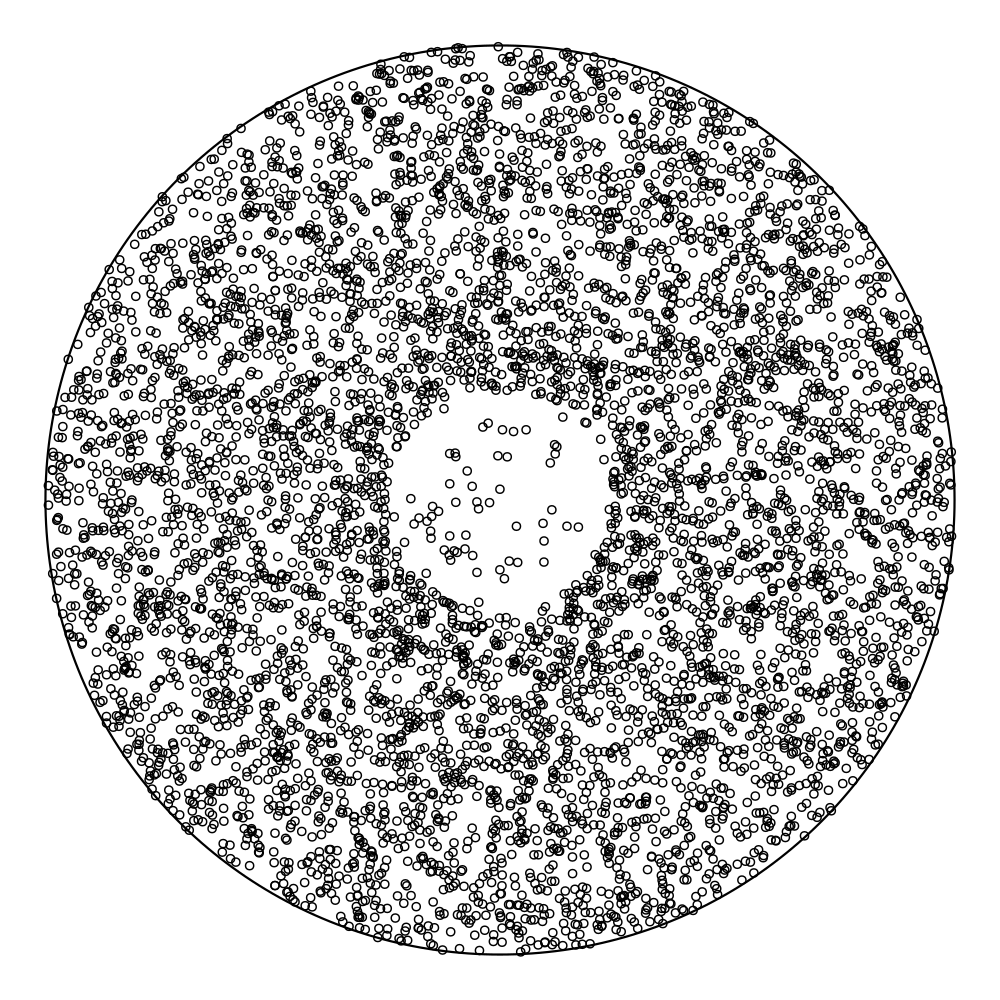}}\\
 
    \subfloat[Observed Fry plot]{\includegraphics[width=0.245\textwidth]{img/fry-points/aniso-regular-fry.png}} \hfill
	\subfloat[]{\includegraphics[width=0.245\textwidth]{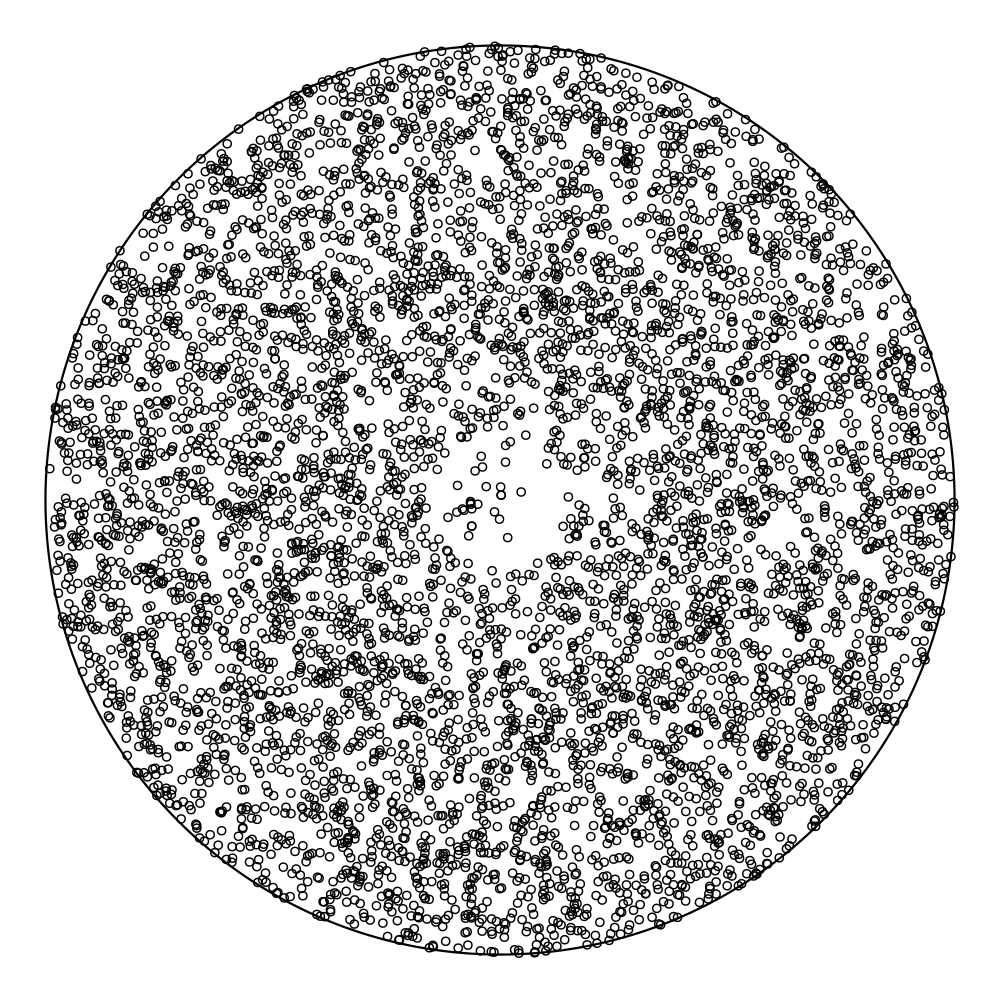}}\hfill 
	\subfloat[]{\includegraphics[width=0.245\textwidth]{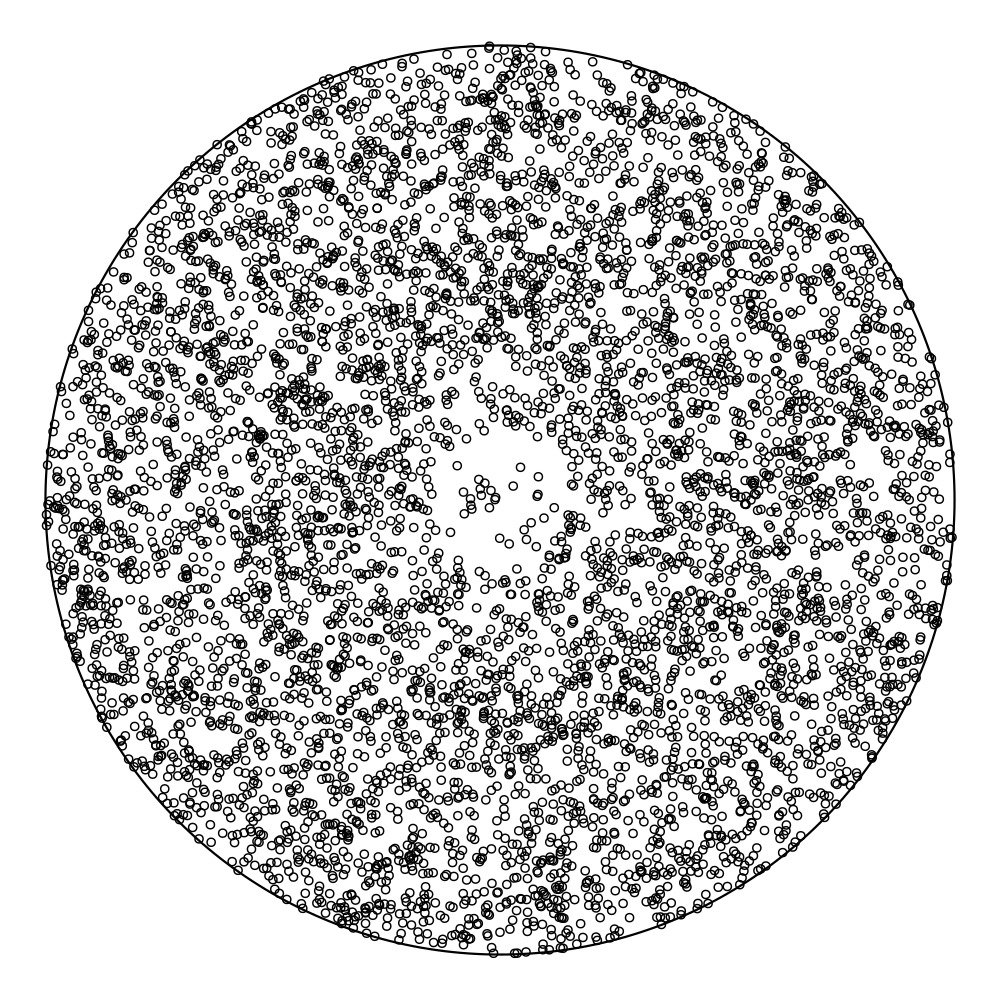}}\hfill
	\subfloat[]{\includegraphics[width=0.245\textwidth]{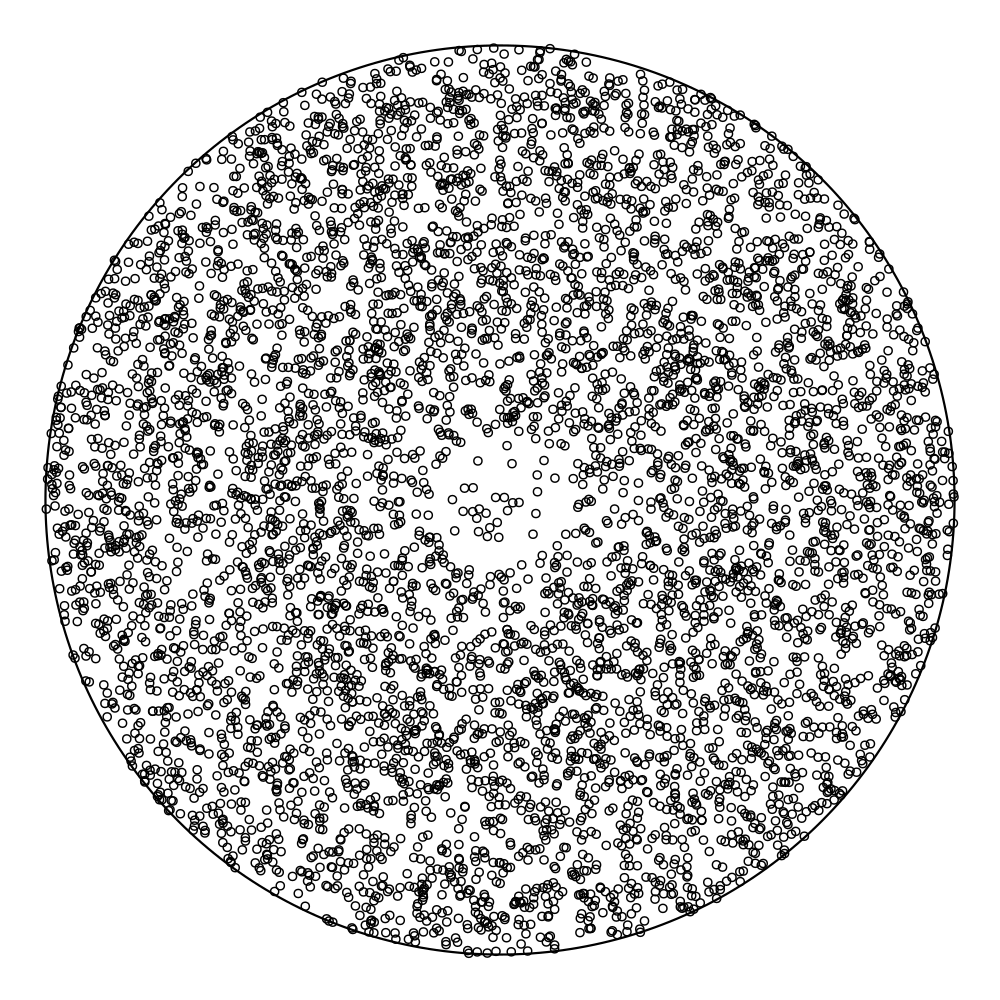}}
	\caption{Examples of Fry plots generated for the Monte Carlo test. (a) and (e) display the observed Fry plots from an isotropic and an anisotropic regular process, respectively. (b) - (d) and (f) - (h) show simulations obtained by random group-wise rotation of the observed Fry points.}
	\label{fig:fry-group-rot}
\end{figure}

\subsection{Test statistic and decision}\label{sec:statistic}

For the Monte Carlo test we first select a functional contrast summary statistic based on the reduced second-order moment measure $\mathcal{K}$.  Then, we compare the estimated statistic for the observed Fry points with the ones obtained from $0 < M < \infty$ independently generated bootstrap Fry point patterns. We will denote the empirical functional contrast summary statistic by $T$ where $T_0$ corresponds to the observed points and $T_1, \dots, T_M$ correspond to the bootstrap samples. 

For the comparison we need a total ordering $\preceq$ where $T_i \preceq T_j$ means that, under the null hypothesis, $T_i$ is at least as extreme as $T_j$. The Monte Carlo $p$-value \citep{davison_bootstrap_1997} as an estimate of the true $p$-value is defined via
\begin{equation}\label{eq:mc-pval}
	p_{\operatorname{MC}} = \frac{1}{M+1} \left(1 + \displaystyle\sum_{i=1}^M \1{T_i \preceq T_0} \right).
\end{equation}

Note, that in order to be able to reject the null hypothesis at the significance level $\alpha$, we need that $\frac{1}{M+1} \leq \alpha$. A common choice of $M$ for $\alpha = 0.05$ is $M=99$. We consequently reject the null hypothesis if the observed value $T_0$ is one of the five most extreme samples.

There are multiple options on how the ordering can be constructed, see \citet{rajala_tests_2022}. 

Consider first the setting when choosing the contrast of the sector $K$-functions as summary statistic, i.e.
\begin{equation}
    T(r) = \widehat{K}^c_{\mathrm{sect, \alpha_1, \alpha_2, \varepsilon}}(r) = \widehat{K}_{\textrm{sect}, \alpha_1, \varepsilon}(r) - \widehat{K}_{\textrm{sect}, \alpha_2, \varepsilon}(r).
\end{equation} 
The parameters $\alpha_1$, $\alpha_2$ and $\varepsilon$ are fixed. 
Under the null hypothesis of isotropy, we expect to see no difference between the two directions $\alpha_1$ and $\alpha_2$ at each distance $r$. In contrast, large absolute values of $T(r)$ point towards anisotropy. It therefore makes sense to construct an ordering based on an $L_p$-norm. We use in our study \begin{equation}
    T_i \preceq_{\mathrm{INT}}T_j \quad \coloneqq  \quad \int_{0}^{r_{\max}} \abs{T_i(r)} \, \mathrm{d}r \, \geq \, \int_{0}^{r_{\max}} \abs{T_j(r)} \, \mathrm{d}r.
\end{equation}
In practise, we can only obtain a discretization of $T_i$ on the interval $[0, r_{\max}]$ with $0 < r_{\max} < \infty$ and hence $T_0, \dots, T_M \in \R^k$ when using $k$ evaluation points.
The integral, i.e. the $L_1$-norm of the functional statistic restricted to $[0, r_{\max}]$, is approximated with the trapezoidal rule. 

A second option for the ordering that we consider in our study is given through global envelope tests \citep{myllymaki_global_2017}. In this approach, statistical depth functions such as the extreme rank length measure (ERL) are used to rank each $T_j$ with respect to the set of all $M+1$ functions. The smallest value of the measure is assigned to the overall most extreme $T_j$. We will denote the corresponding ordering by
\begin{equation}
    T_i \preceq_{\mathrm{ERL}} T_j \quad \coloneqq \quad \operatorname{ERL}(T_i) \leq \operatorname{ERL}(T_j).
\end{equation}
The ERL measure first ranks all $M+1$ functions at each fixed evaluation point from most to least extreme. We can then represent each function as a sorted vector of its pointwise ranks at the $k$ evaluation points. The lexicographical ordering of these sorted vectors gives the overall ordering of $T_0,\dots, T_{M}$.
In our context of differences of sector $K$-functions the pointwise extremeness is two-sided as both large negative and large positive values are extreme under the null hypothesis. 

In the second scenario where we use the ratio $F_{r, \alpha}$ of \citet{wong_isotropy_2016} as functional contrast statistic we have theoretical results under the null hypothesis. This knowledge can be incorporated into the ordering. In particular, one can first compute the pointwise deviations to the known reference value, i.e. \begin{equation}
    D_{r, \alpha}(\varepsilon) = \widehat{F}_{r, \alpha}(\varepsilon) - \frac{\varepsilon}{\nicefrac{\pi}{2}},
\end{equation} which, for each fixed distance $r$ and direction $\alpha$,  results in a functional contrast summary statistic in $\varepsilon$. \citet{wong_isotropy_2016} proceed by computing the $L_{\infty}$-norm of this contrast $D_{r, \alpha}$. Finally, the functional contrast summary statistic in $r$ is given by a maximization over all directions $\alpha$. In other words, we obtain
\begin{equation}
    T(r) = \sup_{\alpha \in [0, \pi]} \sup_{\varepsilon \in [0, \nicefrac{\pi}{2}]} \abs{D_{r, \alpha}(\varepsilon)}.
    \label{eq:wong-chiu}
\end{equation}
Contrary to the contrast of the two sector $K$-functions, this summary statistic already measures absolute differences. Hence, only one-sided extremeness needs to be taken into account in the construction of the orderings.

% %%%%%%%%%%%%%  Simulation study %%%%%%%%%%%%%%%%%%%%%%%%%%%%%%%%%%%%%
\section{Simulation study}\label{sec:study}

We conduct a simulation study to compare the empirical levels and powers of the proposed nonparametric isotropy test using random rotations of the Fry points.

\subsection{Point process models}

In this study we consider four parametric point process models that are derived from different types of point process classes and show different types of anisotropy. In the Strauss process and the Thomas-like cluster process  anisotropy is induced by the geometric anisotropy mechanism. In the Poisson line cluster point process and the Matérn-like cluster point process with elliptical clusters, anisotropy is specified by the von Mises-Fisher distribution for the direction of the lines and clusters. We will use the parametrization introduced in \citet{rajala_tests_2022} where parameters can be interpreted in a similar way for all four models. The authors proposed to use 
\begin{itemize}
    \item a range parameter $0 < R < \infty$ describing the spatial scale,
    \item a strength parameter $\gamma \in [0,1]$ describing the strength of the structural features, where $\gamma=0$ corresponds to the case where the structural feature is most prominent in the observed point pattern,
    \item a degree of anisotropy $a \in (0,1]$ where $a=1$ describes the case of isotropy. 
\end{itemize}

The exact parametrization and model descriptions are explained below. The intensity of all processes is chosen as $\lambda=0.005$. To investigate the influence of the number of observed points, we consider square observation windows $W=[-\nicefrac{L}{2}, \nicefrac{L}{2}]^2$ where the side length $L$ is chosen such that the expected number of observed points $n = \lambda \cdot|W|$  is either $100$, $300$ or $500$. That is, we choose $L \in \{\sqrt{2}\cdot100, \sqrt{6}\cdot100, \sqrt{10}\cdot100\}$.

In our power study, we simulated $1000$ point patterns for each parameter setting of each model in each of the three observation windows. 

\subsubsection{Strauss point process with geometric anisotropy mechanism}
The Strauss point process is an example of a regular point process belonging to the class of Gibbs point processes. When defined on the whole space $\R^2$ it is stationary and, additionally, isotropic. For the anisotropic Strauss process, we consider the geometric anisotropy mechanism in terms of a compression matrix
\begin{equation}
    C = \begin{bmatrix}
        \nicefrac{1}{a} & 0 \\
        0 & a
    \end{bmatrix},  \quad a \in (0,1]
    \label{eq:geom-aniso}
\end{equation}
which describes a stretching along the $x$-axis and a compression along the $y$-axis while being area preserving.
First, an isotropic Strauss process is simulated in $C^{-1}W$. We set the interaction radius of this process to $R$ and the interaction parameter to $\gamma$. Since we use the implementation available in the R-package \texttt{rstrauss} \citep{rstrauss}, it is possible to directly specify the number of points $n$ instead of specifying the intensity parameter $\beta$ that controls but does not coincide with the intensity $\lambda$. The generated pattern is then transformed by applying $C$ to obtain a realization of the anisotropic Strauss point process in $W$. 

\subsubsection{Thomas-like cluster point process with geometric anisotropy mechanism}

The Thomas-like cluster point process is constructed similar to the classical Thomas cluster process. But instead of using Poisson-distributed random variables for the number of clusters as well as the number of offsprings per cluster, we consider here fixed numbers, i.e. Binomial processes. The construction works as follows. As for the classical Thomas cluster point process, we first sample parent points independently and uniformly which define the centers of the clusters. Then, offspring points are generated by sampling random displacements from the parent point following a two-dimensional centered Gaussian distribution with covariance matrix given by multiplying the identity matrix with a scalar $\sigma^2$. The offspring points form the realization of the isotropic process.

The standard deviation $\sigma$ of the Gaussian displacements is chosen as $\sigma = \nicefrac{R}{\chi^2_{p}}$ where $\chi^2_{p}$ denotes the $p$-quantile of a $\chi_2^2$-distribution. In our study we set $p=0.94$ which means that $94~\%$ of the offspring points are contained in a disk with radius $R$ around the parent point. 

The strength parameter $\gamma$ is used to control the number of clusters $n_0$ in the observation window through $n_0 = \gamma n$. Additionally, we restrict the parameter space of $\gamma$ slightly such that $n_0 \geq 5$ holds in all realizations. The desired $n$ offspring points in the observation window are distributed evenly among the parent points, i.e. each cluster will contain $n_1$ points such that $n_0n_1$ is closest to $n$. This means that for small values of $\gamma$ we observe only few clusters containing many points, while $\gamma = 1$ implies that each cluster consists of a single point.

For the anisotropic processes we use again the geometric anisotropy mechanism in terms of the transformation $C$ defined in Equation~\eqref{eq:geom-aniso}. The anisotropy is visible in the contours of the clusters. In case of isotropy, the clusters have spherical contours, while in the anisotropic case the contours are elliptical.

As also clusters with center outside the observation window can have offspring inside the window, we will use a boundary correction. This is done by simulating $\lfloor n_0\abs{W\oplus b_{\sigma}(0)}/\abs{W}\rfloor$ parent points on the dilated observation window $C^{-1}W\oplus b_{\sigma}(0)$. All offspring falling outside the observation window are discarded.

\subsubsection{Poisson line cluster point process with von Mises-Fisher distributed line directions}

This model is based on scattering points in the vicinity of a random line system. Lines are generated according to a Poisson line process conditional on the total line length in the observation window. 
Anisotropy is introduced by specifying a preferred direction of the lines. We simulate the line process on the representation space $\R \times \So$ \citep{skm13}. The pair $(p, \theta) \in \R \times \So$ defines a line that forms a counter-clockwise angle $\theta$ from the $x$-axis and which has the orthogonal signed distance $p$ from the center of the observation window, in our case the origin. As usually with line processes, we first simulate the line process in the bounding disk of the observation window and then clip the lines to the square observation window. Denote the radius of the bounding disk by $r_{d}$.

Both parameters of the line are chosen randomly and independent from each other with $p$ being uniformly distributed on the interval $[-r_{d}, r_{d}]$. The line direction distribution is a von Mises-Fisher distribution on $\So$ with mean direction $\mu$ and concentration parameter $\kappa$ given as \begin{equation}
    \kappa = \kappa_{\max} \cdot \left(1 - \exp(1-1/a)\right)
    \label{eq:von-mises}
\end{equation}
where we take $1/0 = \infty$. Thus, the case $a=0$ corresponds to the highest concentration around the mean, while $a=1$ implies $\kappa=0$ which defines the uniform distribution on $\So$. In our study we fixed the maximal concentration to $\kappa_{\max}=10$. 

Instead of specifying the mean number of lines per unit volume, we will generate independent lines until the total length of the lines within the square observation window reaches a threshold $M = L\cdot5^{1+\gamma}$. 
The structure parameter $\gamma$ therefore controls the number of lines. For $\gamma=0$ we want to achieve at least the length of five times the side length. 

After the lines are generated, we simulate a Binomial process on the line system conditional on observing $n$ points in the observation window. Each point is first placed independently and uniformly directly on the line system. Then, it is displaced orthogonally from the line following a centered Gaussian with standard deviation $R$. If $R=0$ no displacement takes place. In order to account for edge effects, the possible initial line positions are extended slightly to the outside of the observation window. 

\subsubsection{Matérn-like cluster point process with elliptical clusters with preferred direction}

As a second type of cluster process, we use a Matérn-like cluster point process with elliptic clusters. As with the Thomas-like cluster process we again replace all Poisson-distributed random variables with fixed numbers. While cluster shapes are aligned in case of the geometric anisotropy mechanism, we now choose them at random around a preferred direction. The cluster centers are drawn from a Binomial point process in the dilated window. The dilation radius is taken as the longer half-axis of the cluster as explained below. The number of clusters $n_0$ is given as $n_0 = \gamma n$. As before, we restrict the parameter space of $\gamma$ such that $n_0 \geq 5$ holds in all parameter combinations. 

The individual cluster shapes are ellipses where the longer half-axis has length $R$ and the shorter half-axis has length $\tau R$ with $\tau \in (0,1]$ being a constant. The anisotropy induced by the direction of the ellipses is given through the von Mises-Fisher distribution on $\So$ with mean direction $\mu$ and $\kappa$ as defined in Equation~\eqref{eq:von-mises}. In our study we set the aspect ratio $\tau$ of the axes to $\tau = 0.4$, $\mu = \pi/3$ and the maximal concentration $\kappa_{\max}=10$.

% %%%%%%%%%%%%%  Results %%%%%%%%%%%%%%%%%%%%%%%%%%%%%%%%%%%%%%%%%%%%%%
\FloatBarrier
\subsection{Results}\label{sec:results}

We present in the following the results of the nonparametric isotropy tests using the three different random rotation techniques. We considered several parameter settings for each of the point process models. The exact choices are listed in Table~\ref{tab:parameter-setting}. Figure~\ref{fig:samples} shows one realization for four different parameter combinations of each model.

\begin{table*}[th]\centering
\renewcommand{\arraystretch}{1.5}
\begin{tabularx}{\textwidth}{lcccCc} \toprule
& \multicolumn{4}{c}{Model parameters} &  \\ \cmidrule(r){2-5}
Point process model family & $R$ & $\gamma$ & $a$ & other & $n$  \\ \midrule
Strauss  & $5$, $10$ &  $0.00$, $0.40$, $0.80$  & $0.7$, $1.0$ & -- & $100$, $300$, $500$  \\
Thomas-like cluster & $10$, $20$ &  $0.05$, $0.15$, $0.25$  & $0.7$, $1.0$ &  $p=0.94$ & $100$, $300$, $500$  \\
Poisson line cluster & $0$, $1$  & $0.00$, $0.40$, $0.80$ & $0.7$, $1.0$ & $\mu=\frac{\pi}{3}$, $\kappa_{\max} = 10$ & $100$, $300$, $500$  \\  
Matérn-like cluster & $10$, $20$  & $0.05$, $0.15$, $0.25$ & $0.7$, $1.0$ &  $\mu=\frac{\pi}{3}$, $\tau=0.4$, $\kappa_{\max} = 10$ &$100$, $300$, $500$  \\
\bottomrule 
\end{tabularx}
\caption{Overview of the parameters used in the simulation study. The model parameters are the range parameter $R$, the strength parameter $\gamma$, the degree of anisotropy $a$. The parameter $n$ denotes the expected number of points in the point patterns.}
\label{tab:parameter-setting}
\end{table*}

\begin{figure}[ht]
\captionsetup[subfloat]{aboveskip=4pt, belowskip=4pt, labelformat=empty}
	\subfloat[$R=10$, $\gamma=0.4$, $a=1$]{\includegraphics[width=0.245\textwidth, height=0.245\textwidth]{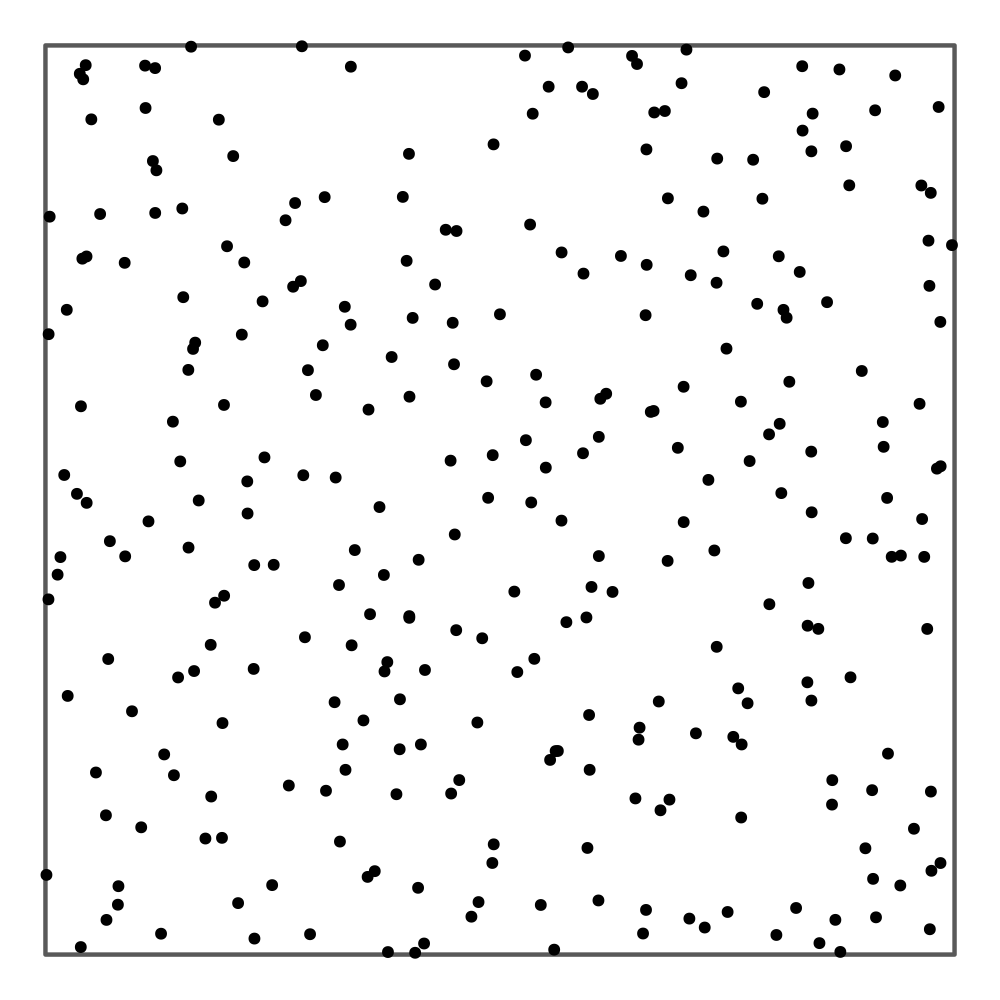}} \hfill
	\subfloat[$R=10$, $\gamma=0$, $a=0.7$]{\includegraphics[width=0.245\textwidth, height=0.245\textwidth]{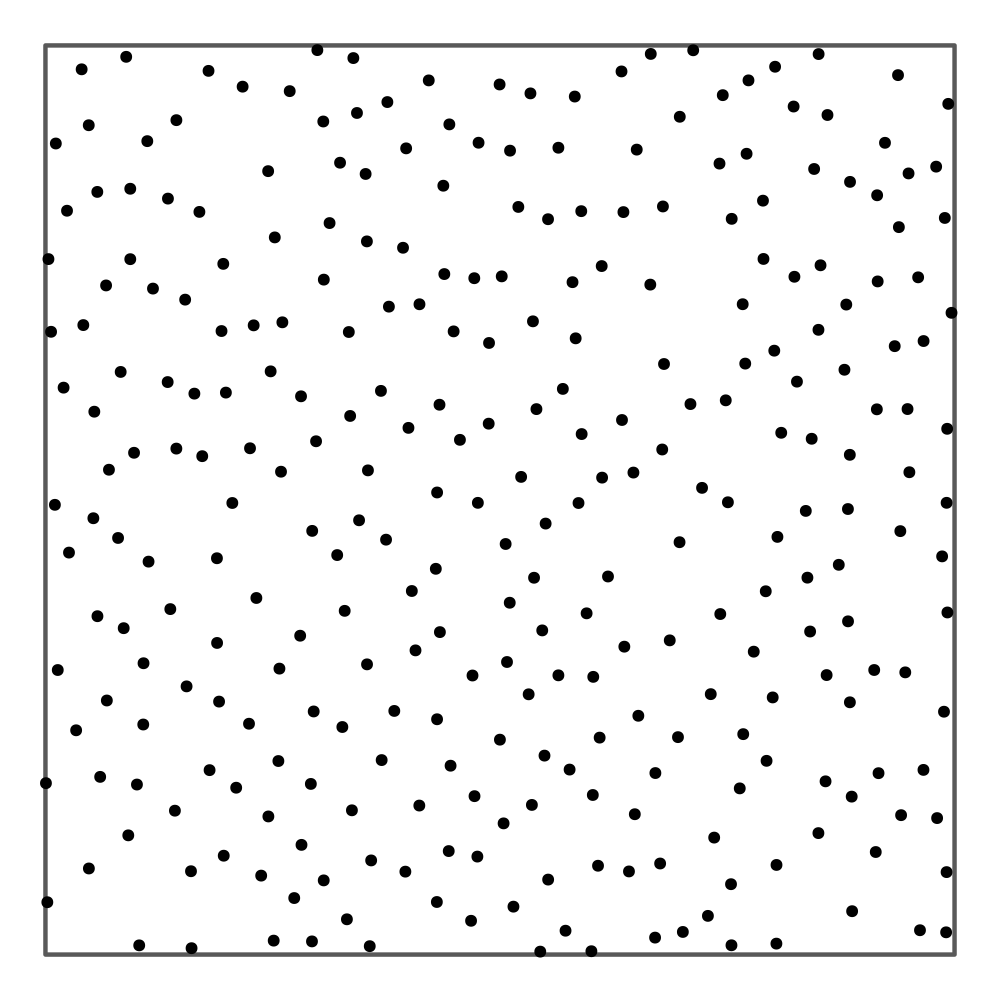}}\hfill 
	\subfloat[$R=10$, $\gamma=0.4$, $a=0.7$]{\includegraphics[width=0.245\textwidth, height=0.245\textwidth]{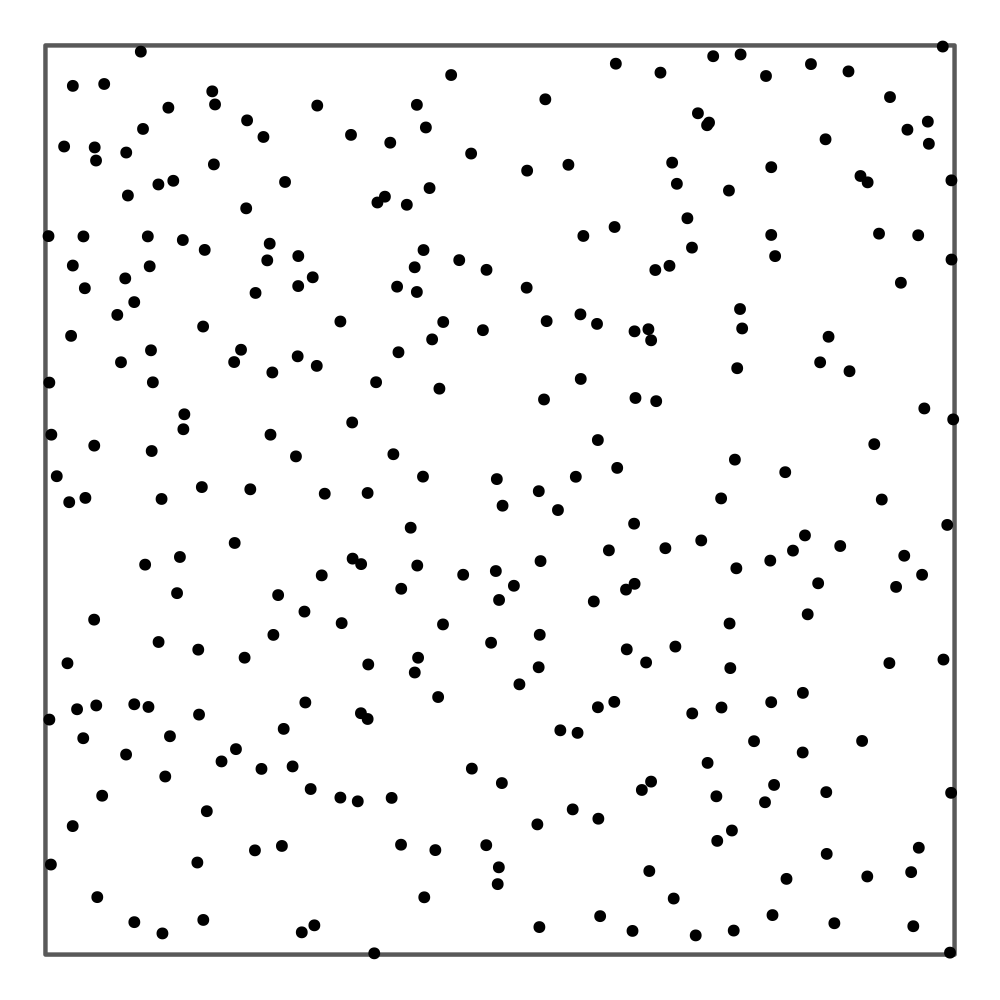}}\hfill
	\subfloat[$R=5$, $\gamma=0.8$, $a=0.7$]{\includegraphics[width=0.245\textwidth, height=0.245\textwidth]{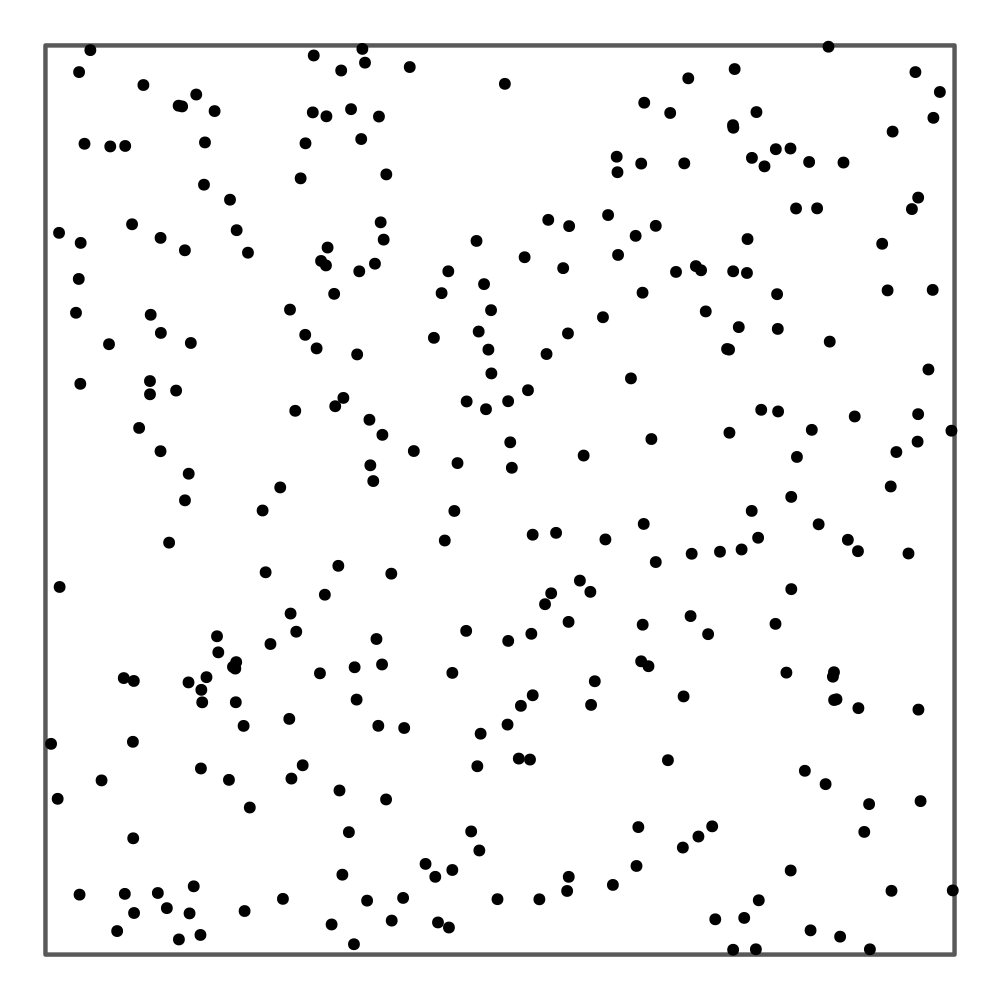}}

 \subfloat[$R=10$, $\gamma=0.15$, $a=1$]{\includegraphics[width=0.245\textwidth, height=0.245\textwidth]{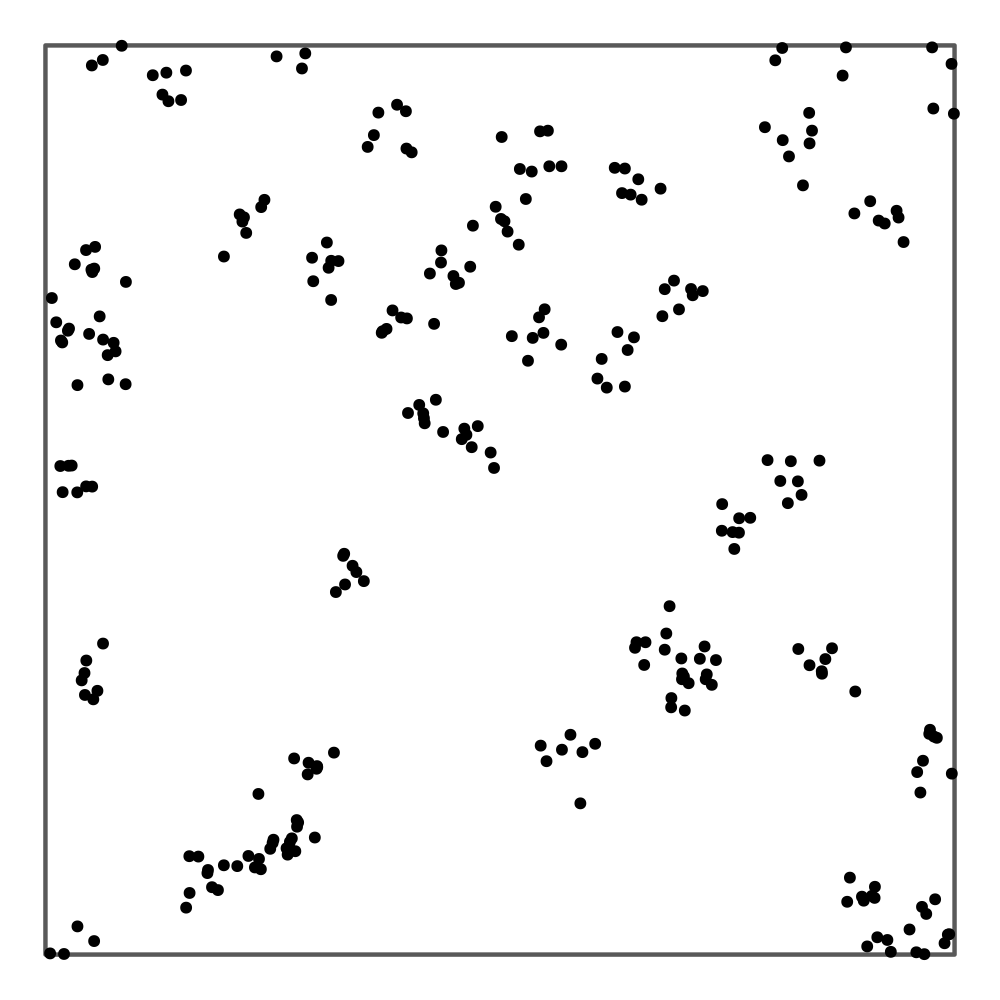}} \hfill
	\subfloat[$R=10$, $\gamma=0.05$, $a=0.7$]{\includegraphics[width=0.245\textwidth, height=0.245\textwidth]{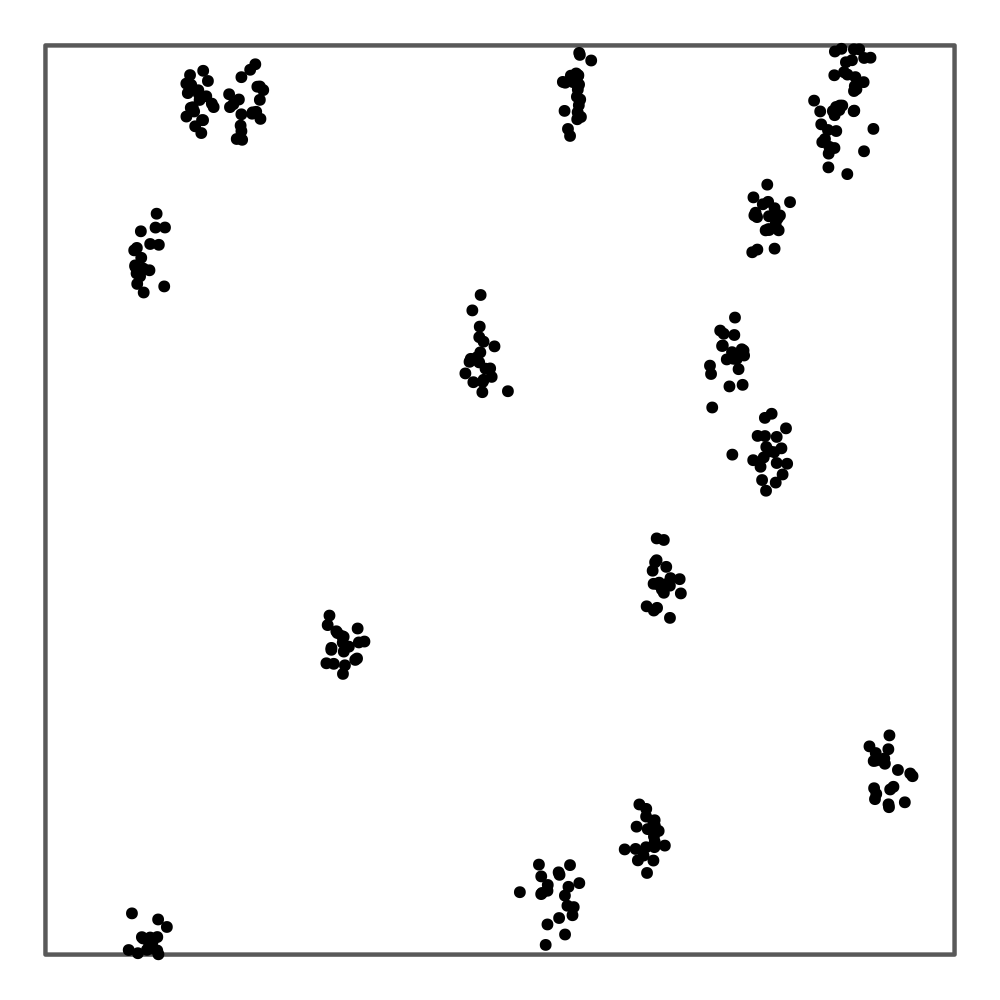}}\hfill 
	\subfloat[$R=20$, $\gamma=0.15$, $a=0.7$]{\includegraphics[width=0.245\textwidth, height=0.245\textwidth]{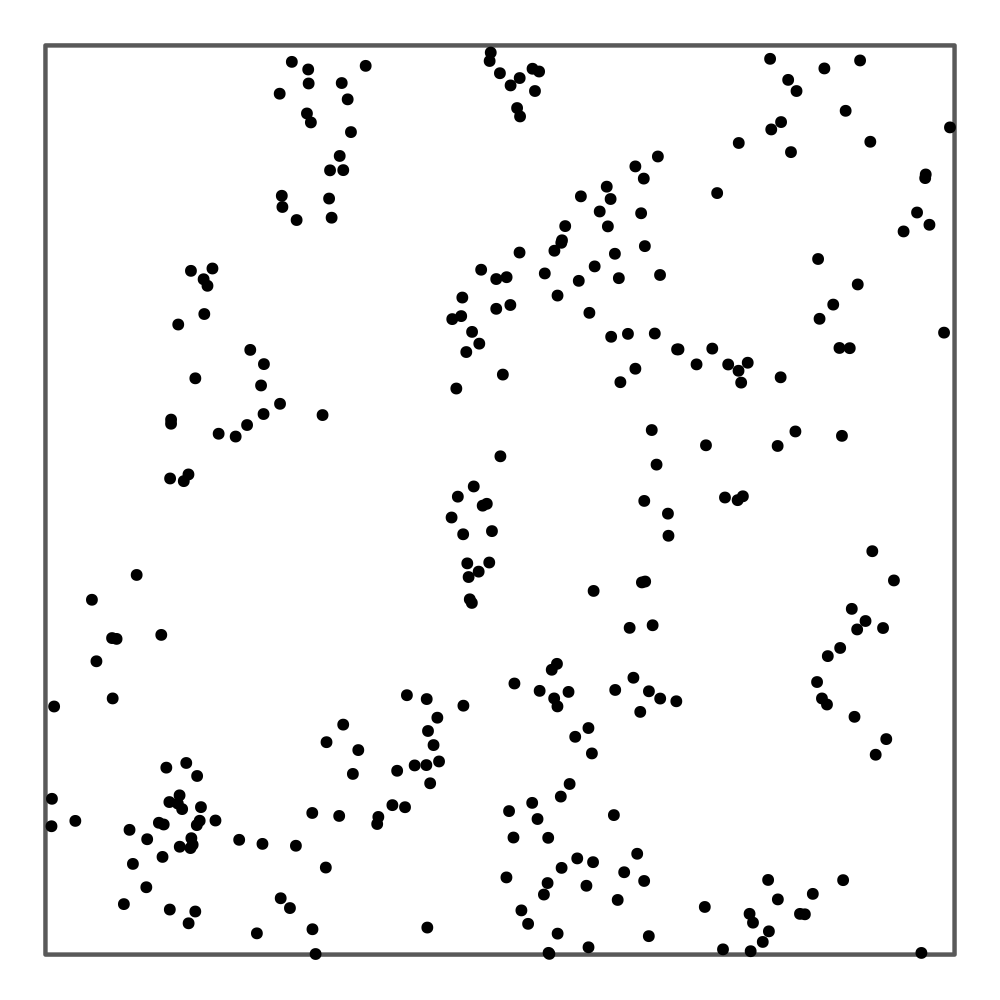}}\hfill
	\subfloat[$R=20$, $\gamma=0.25$, $a=0.7$]{\includegraphics[width=0.245\textwidth, height=0.245\textwidth]{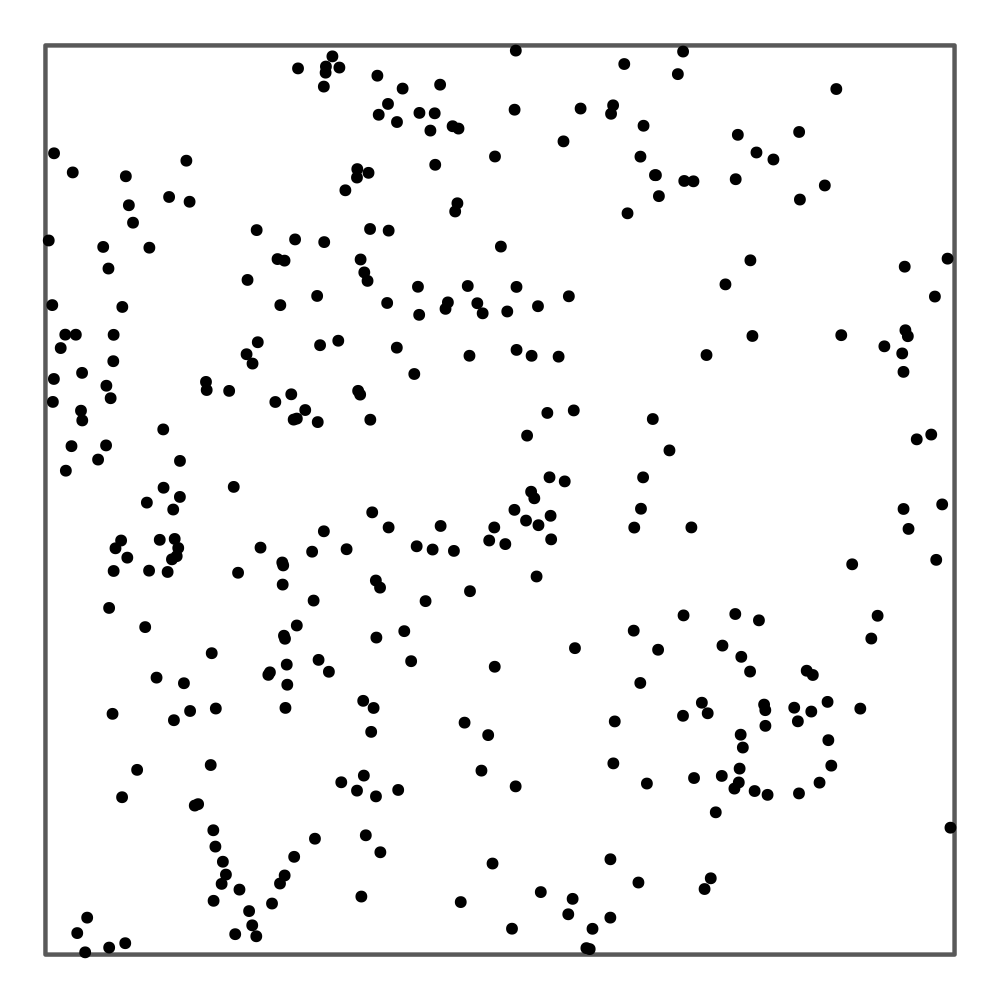}}

 \subfloat[$R=1$, $\gamma=0.8$, $a=1$]{\includegraphics[width=0.245\textwidth, height=0.245\textwidth]{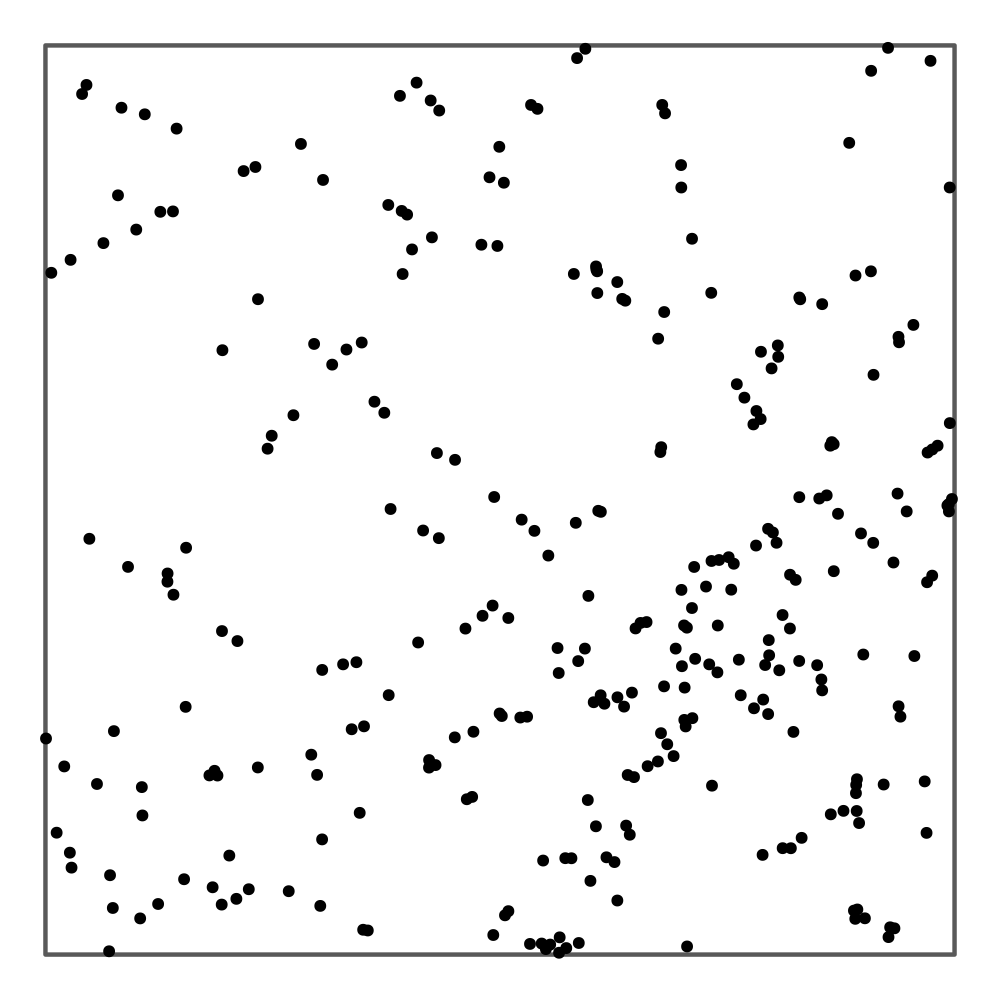}} \hfill
	\subfloat[$R=1$, $\gamma=0$, $a=0.7$]{\includegraphics[width=0.245\textwidth, height=0.245\textwidth]{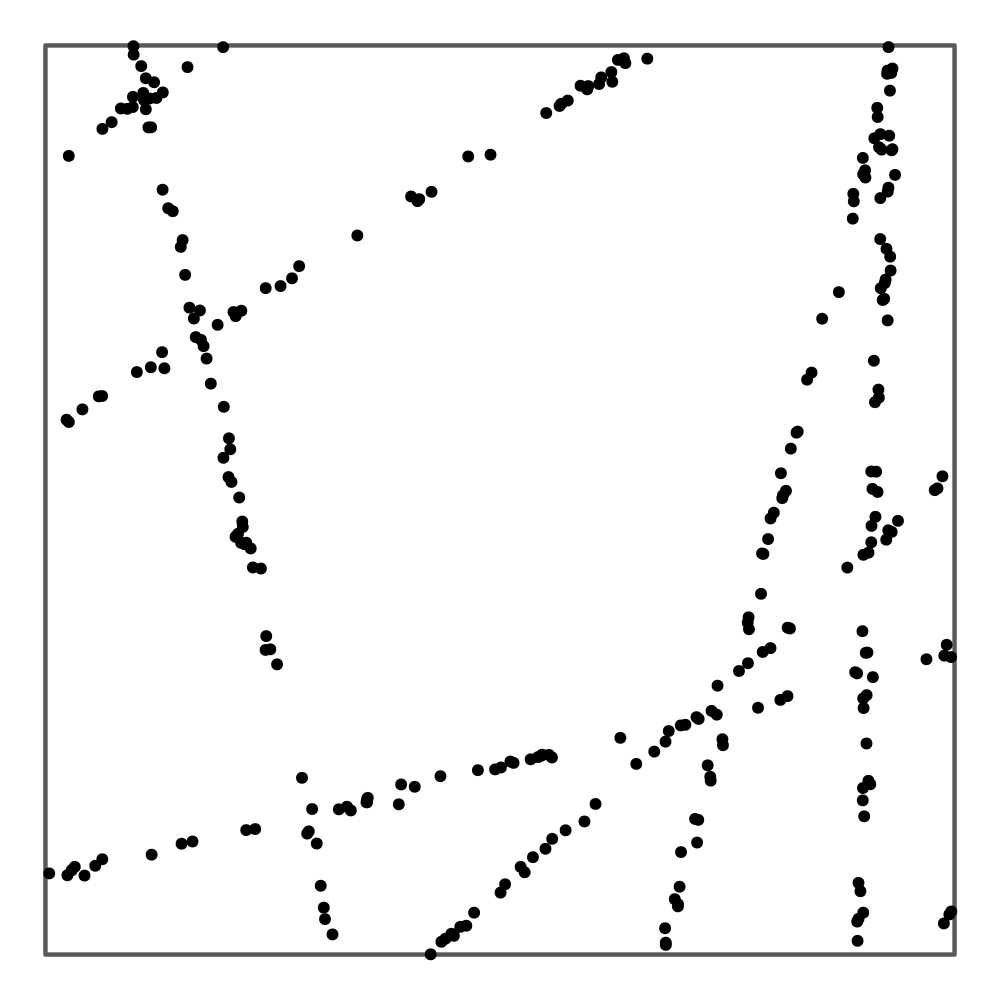}}\hfill 
	\subfloat[$R=0$, $\gamma=0.4$, $a=0.7$]{\includegraphics[width=0.245\textwidth, height=0.245\textwidth]{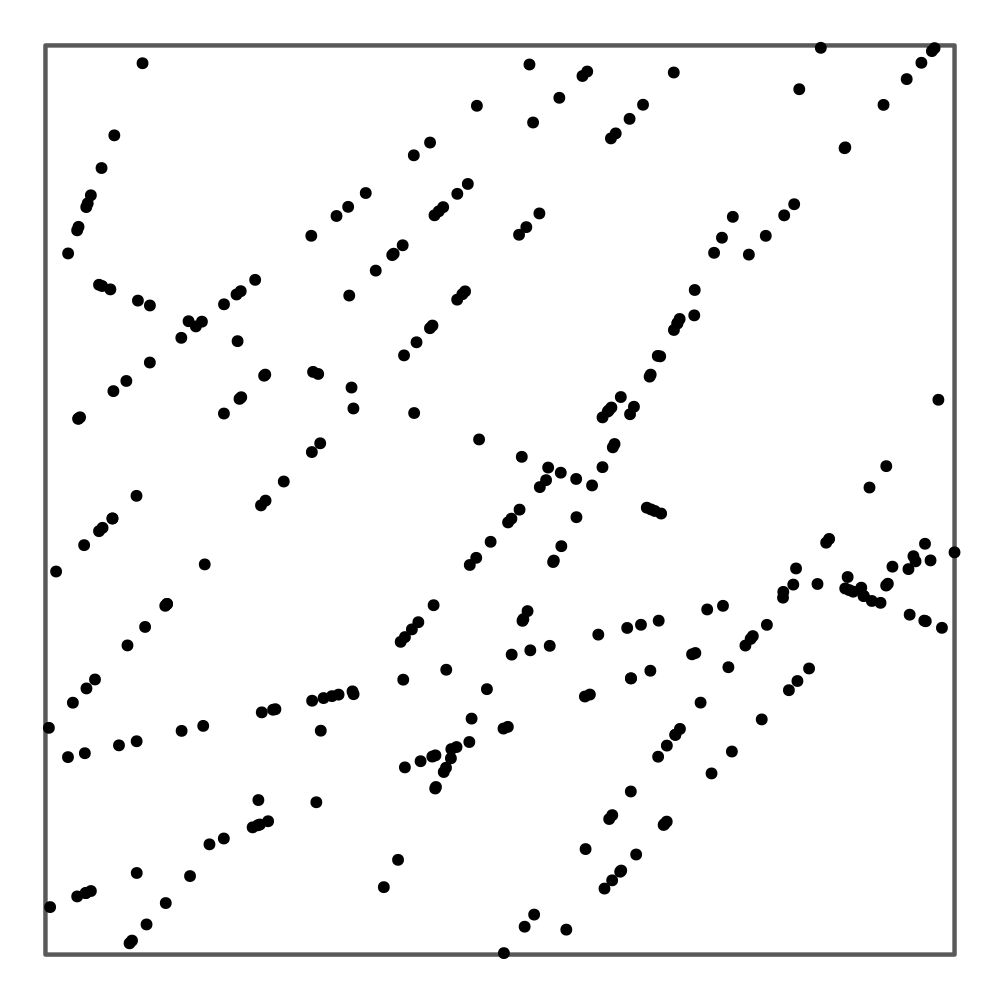}}\hfill
	\subfloat[$R=1$, $\gamma=0.8$, $a=0.7$]{\includegraphics[width=0.245\textwidth, height=0.245\textwidth]{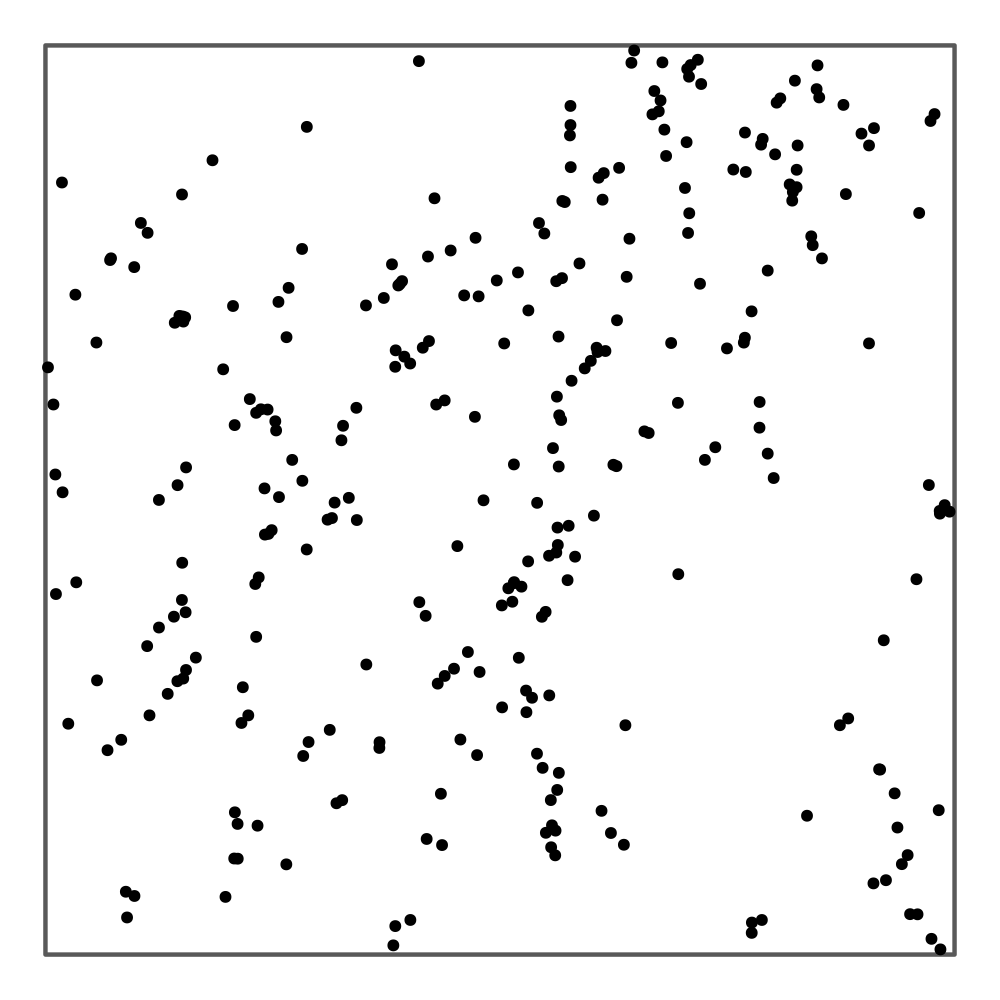}}

 \subfloat[$R=10$, $\gamma=0.15$, $a=1$]{\includegraphics[width=0.245\textwidth, height=0.245\textwidth]{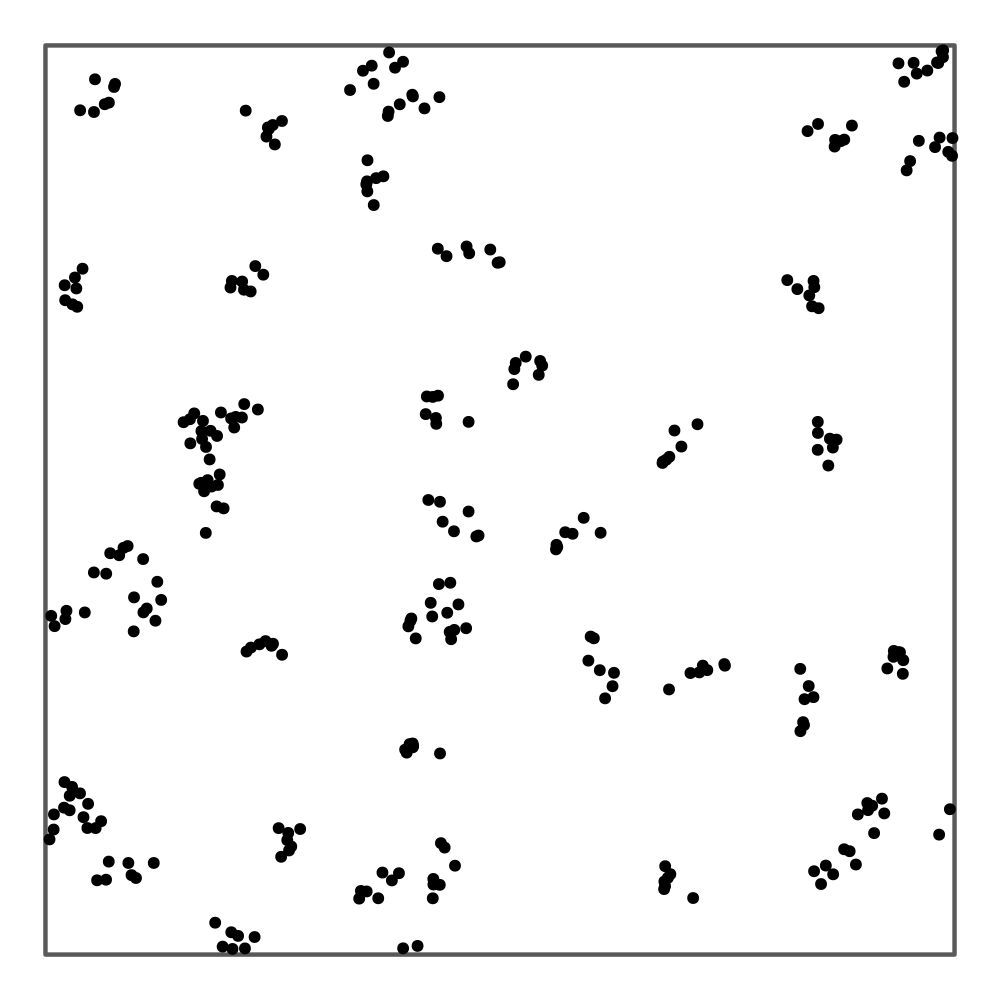}} \hfill
	\subfloat[$R=10$, $\gamma=0.05$, $a=0.7$]{\includegraphics[width=0.245\textwidth, height=0.245\textwidth]{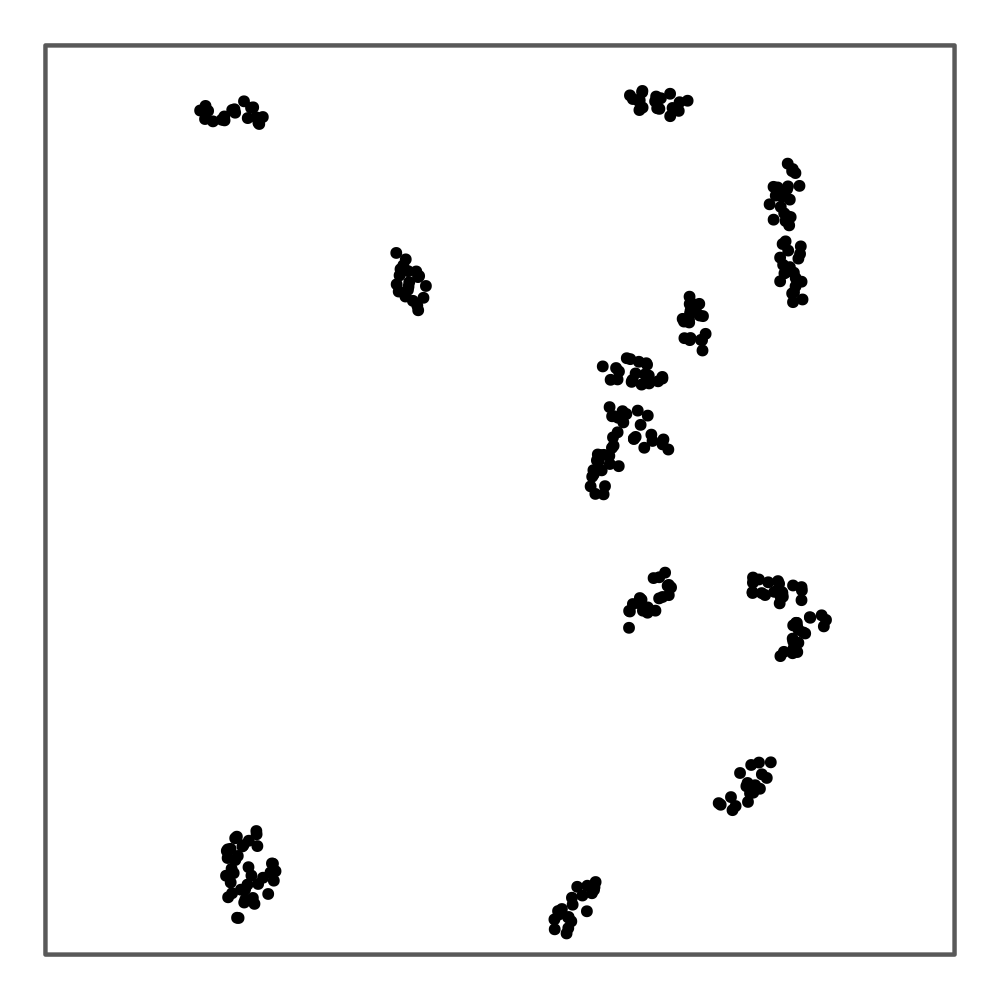}}\hfill 
	\subfloat[$R=20$, $\gamma=0.15$, $a=0.7$]{\includegraphics[width=0.245\textwidth, height=0.245\textwidth]{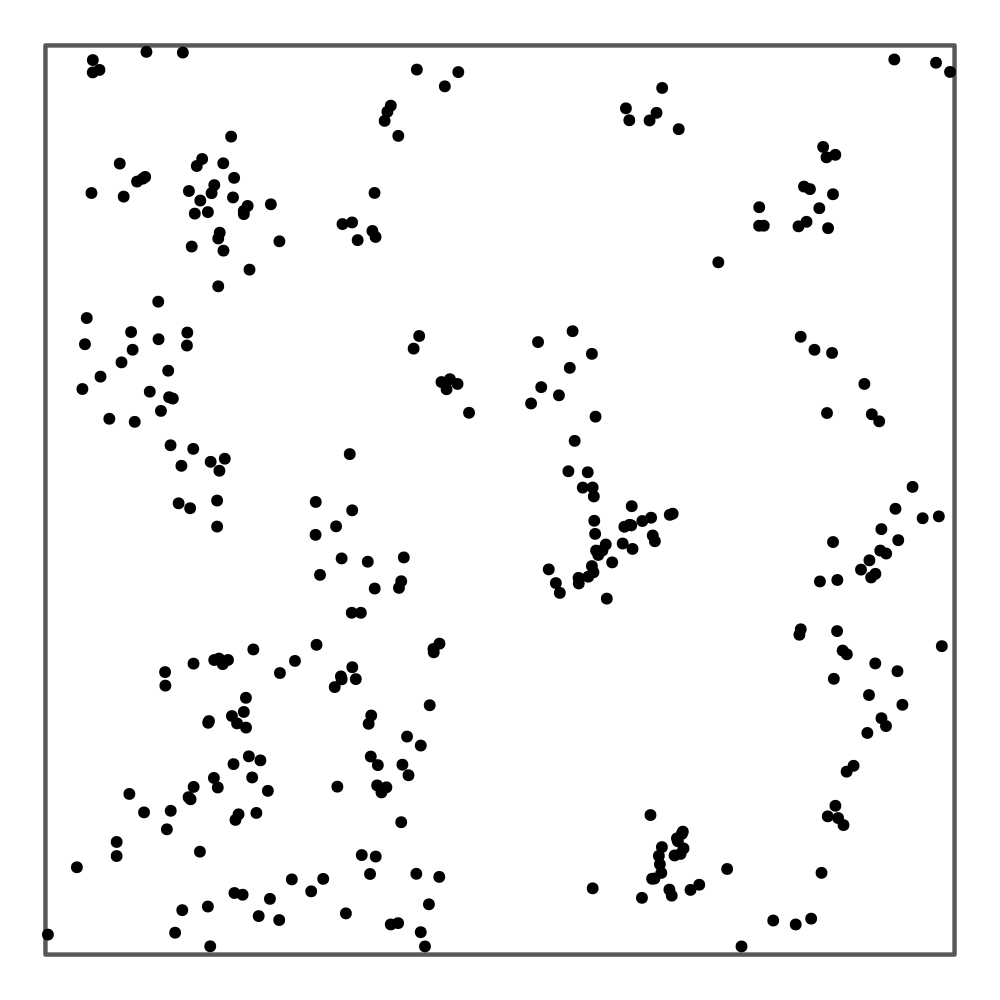}}\hfill
	\subfloat[$R=20$, $\gamma=0.25$, $a=0.7$]{\includegraphics[width=0.245\textwidth, height=0.245\textwidth]{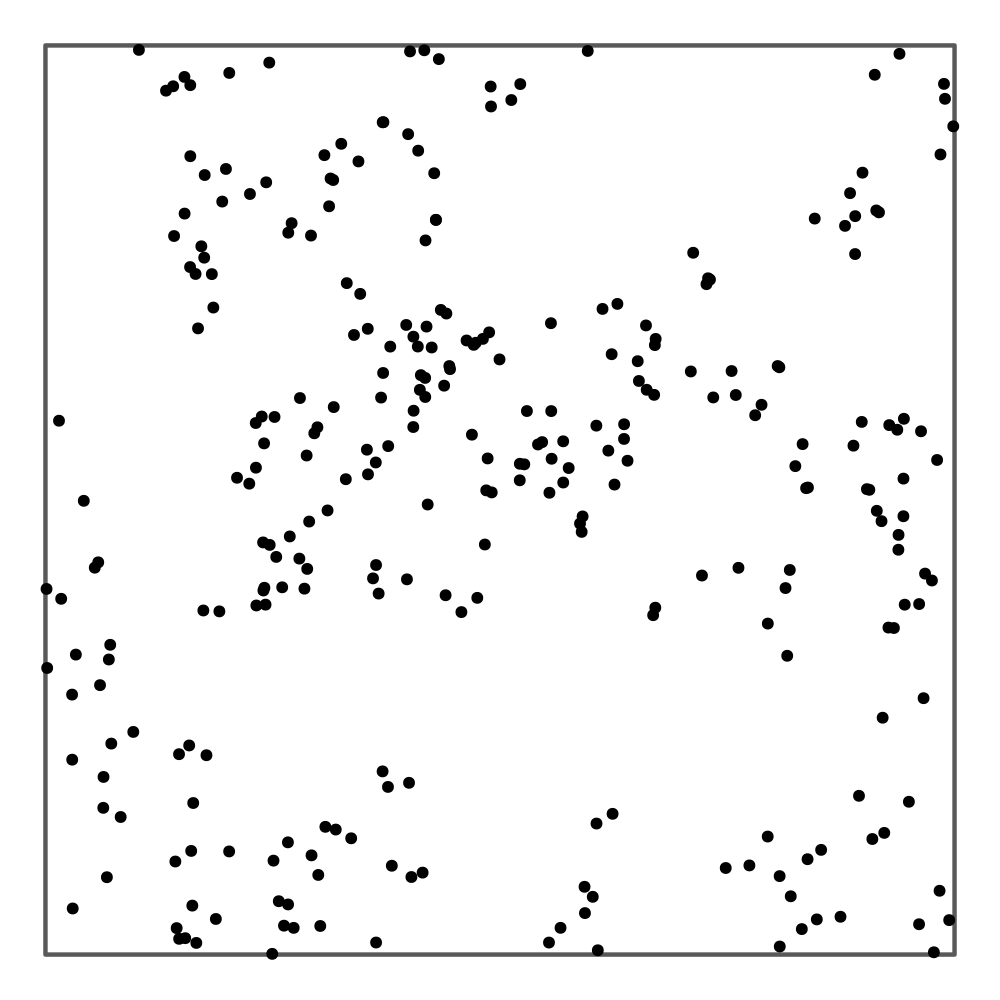}}
	\caption{Realizations of the four point process model families with varying parameters. 
 From top to bottom: Strauss, Thomas-like cluster, Poisson line cluster, Matérn-like cluster. The number of points in each pattern is approximately $n=300$. The leftmost column corresponds to isotropic point processes.}
	\label{fig:samples}
\end{figure}

We use the contrast of the sector $K$-function $K^c_{\mathrm{sect, \alpha_1, \alpha_2, \varepsilon}}$ on the interval $[0, r_{\max}]$ as functional contrast summary statistic. The individual empirical vectors are compared with either the integral ordering or the extreme rank length ordering as defined in Section~\ref{sec:statistic}. When using the contrast of the sector $K$-functions we have to choose four additional tuning parameters $\alpha_1, \alpha_2 \in [0, 2\pi]$, $\varepsilon \in [0, \nicefrac{\pi}{2}]$ and $0 < r_{\max} < \infty$. Additionally, we set the number of evaluation points of the functional statistic to $200$. 

The simulation study in \citet{rajala_tests_2022} showed that the half-opening angle $\varepsilon$ of the sector has little effect on the power of the isotropy tests when all other parameters are fixed. Consequently, we follow their recommendation and set $\varepsilon = \frac{\pi}{4}$. 

The two directions ideally should match the preferred directions in the point pattern. If no expert knowledge regarding the optimal choice of $\alpha_1$ and $\alpha_2$ is available, then a contrast maximization over a set of discrete angles can be incorporated into the ordering similar to the test statistic of \citet{wong_isotropy_2016} given in Equation~\eqref{eq:wong-chiu}. Due to the increased computation effort needed for the maximization, we will fix the two directions in our study. We use $\alpha_1 = 0$ and $\alpha_2 = \frac{\pi}{2}$ corresponding to the $x$- and $y$-axis, respectively, in case of the geometric anisotropic point processes. For the models where anisotropy is introduced by the preferred direction $\mu$ of the von Mises-Fisher distribution, we let $\alpha_1 = \mu$ and $\alpha_2 = \mu + \nicefrac{\pi}{2}$. 

It is well known \citep[see e.g.][]{rajala_tests_2022, wong_isotropy_2016, redenbach_anisotropy_2009}, that for many types of point process models the upper bound $r_{\max}$ needs to be chosen slightly larger than the interaction distance between points when working with second-order characteristics. In our parametrization framework this is given as the parameter $R$. We want to choose $r_{\max} = sR$ where $s > 1$ is a real-valued scaling factor.
\citet{redenbach_anisotropy_2009} state that $s = 1.1$ yields the best performance in case of hard-core processes with hard-core distance $R$. In \citet{wong_isotropy_2016} the choice of $s = 1.1$ results for regular pairwise interaction processes with interaction radius $R$ in high powers but also in high type I error rates in case of isotropic patterns. Hence, they advocate to use a scaling factor of $1.2 \leq s \leq 1.4$. In case of clustered patterns like the classical Thomas or Matérn cluster processes, $r_{\max}$ needs to match the size of the clusters. For a classical Thomas cluster point process which is parameterized through the standard deviation $\sigma$, \citet{wong_isotropy_2016} suggest to use $0.57\cdot 4\sigma \leq r_{\max} \leq 0.67\cdot 4\sigma$ since the diameter of the clusters is approximately $4\sigma$. In our parametrization this recommendation corresponds to $0.57\cdot 2 R \leq r_{\max} \leq 0.67\cdot 2R$ and thus a scaling factor of $1.14 \leq s \leq 1.34$.

In this study we opt for the scaling factor $s = 1.3$ for both clustered and regular point processes. \citet{rajala_tests_2022} additionally recommend that the upper bound $r_{\max}$ should be very large when observing linear structures in the point pattern. Consequently we set $r_{\max} = R + 25$ in the case of the Poisson line cluster process.

Due to the computational complexity we
used $M=99$ for all tests in this study.
The significance level is set to $\alpha=0.05$. Below, we summarize the results from the power study for each type of process separately. 

\subsubsection{Strauss point process with geometric anisotropy mechanism}

\begin{figure}
    \centering
    \includegraphics[width=0.9\textwidth]{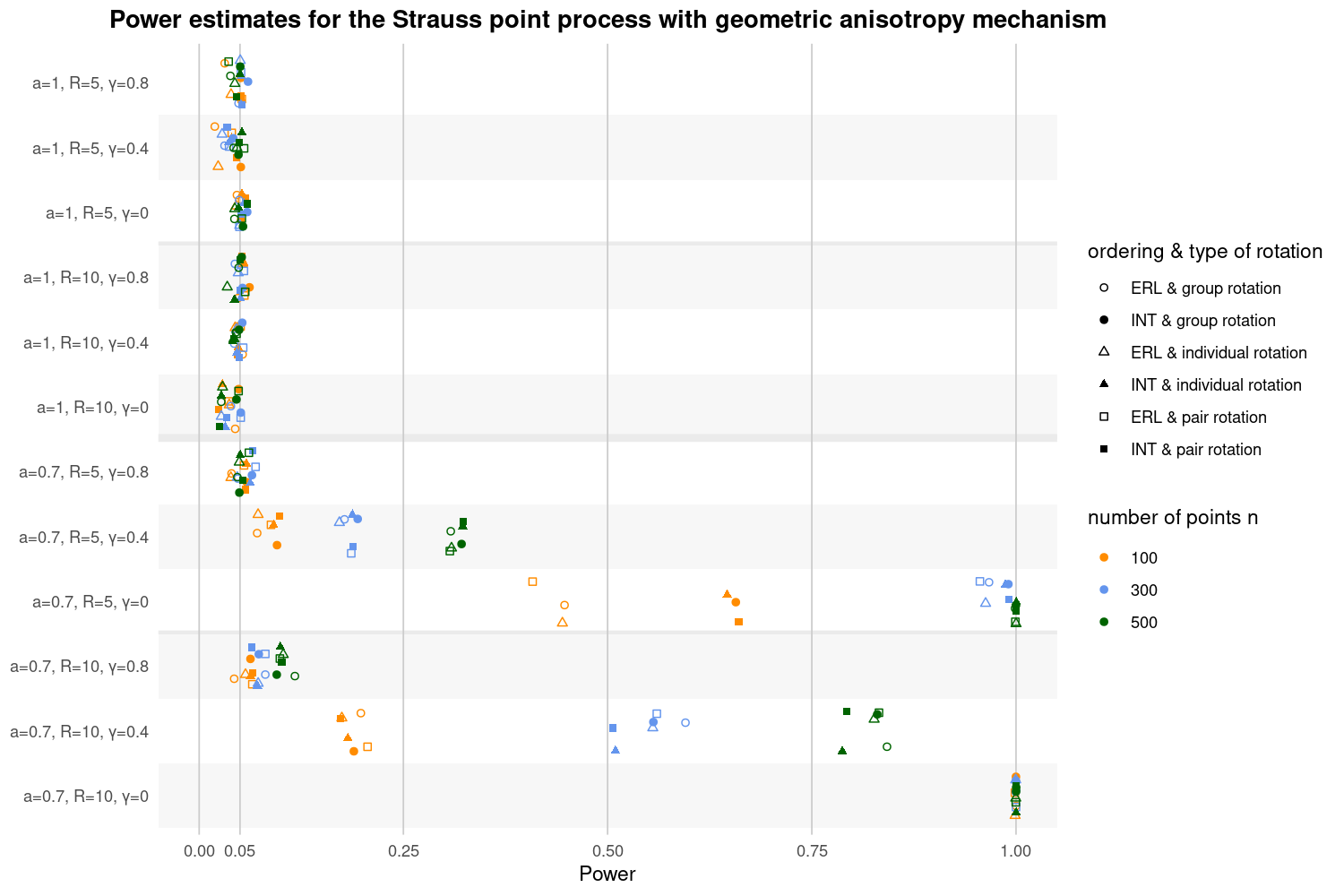}
    \caption{Power estimates for the Strauss point process with geometric anisotropy mechanism. Each parameter combination corresponds to one horizontal block. The upper six combinations represent isotropic models, while the lower six are anisotropic. }
    \label{fig:geomstrauss-results-all}
\end{figure}

All results for the geometric anisotropic Strauss point process are displayed in Figure~\ref{fig:geomstrauss-results-all}. As mentioned above the upper bound was set to $r_{\max} = 1.3\cdot R$ which results in either $6.5$ or $13$. The first six models with $a=1$ correspond to the case of isotropy. We can see that the tests are in some cases conservative, in particular when using the ERL ordering. Almost all tests are conservative for the hardcore Strauss process with $\gamma = 0$ and $R=10$. The strong regularity in this model makes the decision if there is anisotropy or not very easy and thus it is plausible that the empirical size tends to zero. 
Overall the combination of integral ordering and the group rotation indicated by the filled circles in Figure~\ref{fig:geomstrauss-results-all} yield tests with empirical size being close to the nominal level for all six parameter combination.

With regard to the power in case of the anisotropic models we observe as expected that the power increases for decreasing $\gamma$ and increasing $R$. In particular anisotropic hardcore processes are detected almost all the time and even for a small number of points. For $\gamma=0.8$ the isotropic process is close to a Poisson process and thus the geometric anisotropy is hard to detect. This explains the low powers close to the significance level. 

The choice of ordering, i.e. whether we use the integral or the ERL ordering for the comparison, does not have a big influence in most cases, only for the anisotropic hardcore model with $a=0.7$, $R=5$ and $\gamma=0$ the integral ordering is more powerful. 

\begin{figure}[ht]
\captionsetup[subfloat]{aboveskip=4pt, belowskip=4pt, labelformat=empty}
	\subfloat[$100$ points]{\includegraphics[width=0.245\textwidth, height=0.245\textwidth]{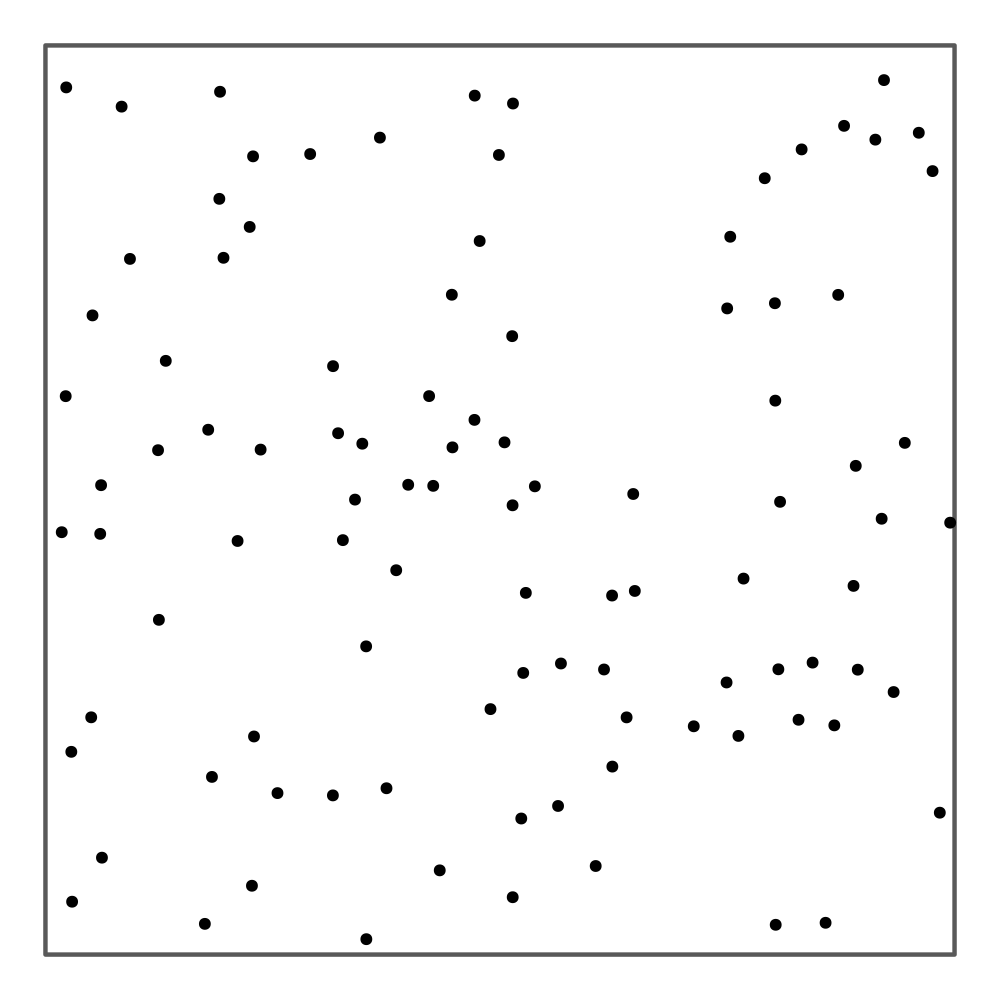}} \hfill
	\subfloat[Fry points in $b_{6.5}(0)$]{\includegraphics[width=0.245\textwidth, height=0.245\textwidth]{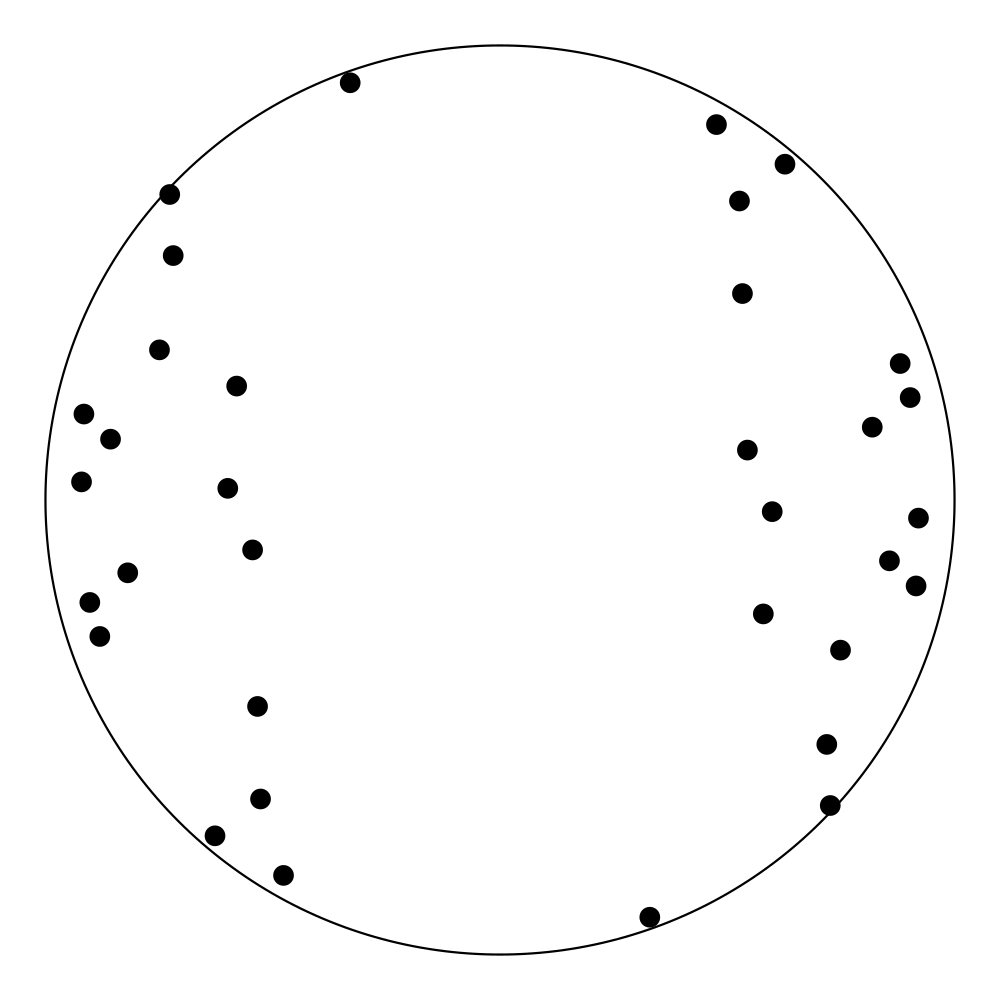}}\hfill 
	\subfloat[$300$ points]{\includegraphics[width=0.245\textwidth, height=0.245\textwidth]{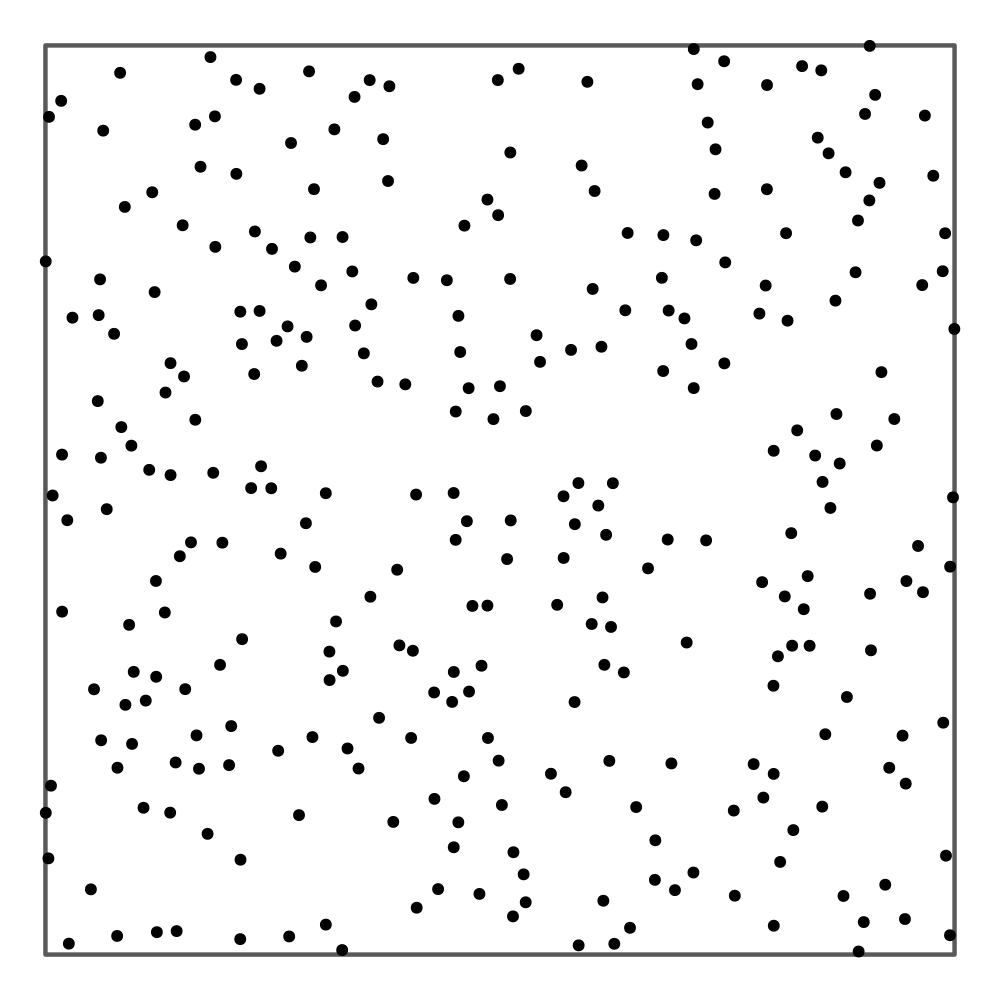}}\hfill
	\subfloat[Fry points in $b_{6.5}(0)$]{\includegraphics[width=0.245\textwidth, height=0.245\textwidth]{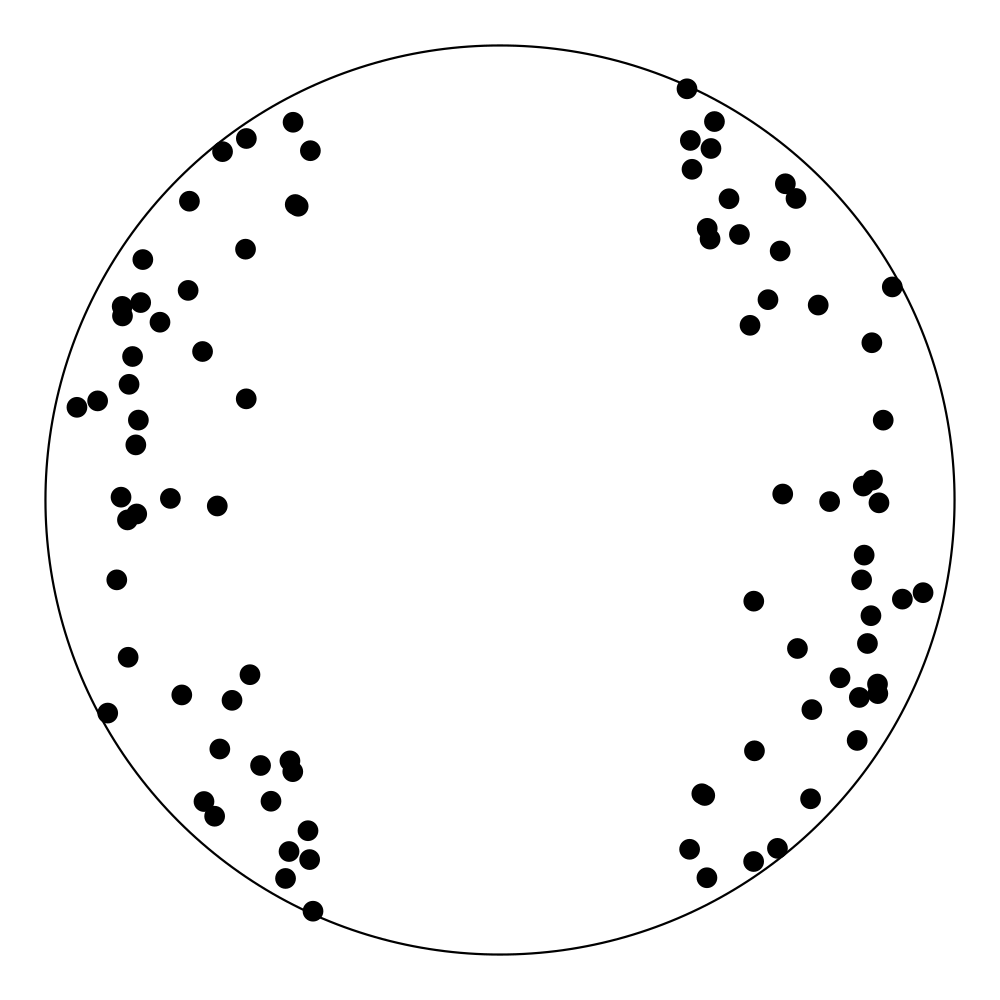}}
	\caption{Two realizations of the Strauss point process with geometric anisotropy mechanism for the parameter setting $a=0.7$, $R=5$, $\gamma=0$. The realization on the left has around $100$ points, the one on the right $300$. In our study we obtained Monte Carlo $p$-values of $0.01$ (integral ordering) and $0.14$ (ERL ordering) for the left pattern and of $0.03$ (integral ordering) and $0.08$ (ERL ordering) for the right pattern.}
	\label{fig:geomstrauss-number}
\end{figure}

This can be explained due to the low number of Fry points that are taken into account for hard-core point processes, see Figure~\ref{fig:geomstrauss-number}. For the pattern on the left that contains $100$ points, only $34$ Fry points lie in the set $b_{6.5}(0)$. Only four out of these points have norm smaller than $4$ and consequently we have a very small number of possible values of the contrast until $r=4$. With the extreme rank length ordering it is possible that contrast statistics obtained from the bootstrap samples are overall more extreme than the observed contrast just because they attain the highest/lowest value at the small distances. 

The three different rotation methods are overall comparable in their power for the anisotropic models and there is no preferred rotation scheme. 

\subsubsection{Thomas-like cluster point process with geometric anisotropy mechanism}

Figure~\ref{fig:geomthomas-results-all} shows all results obtained for the Thomas-like cluster point process with geometric anisotropy mechanism. As for the regular process before, the first six parameter combinations describe isotropic models and the remaining six anisotropic ones. For the clustered process we observe a clear difference in the empirical sizes for the three different rotation mechanisms. Only with the group-wise rotation it is possible to control the type I error probability. For the smaller clusters where the circular shape is very clear, the size matches the level, but for the larger clusters we obtain liberal tests. This is due to having more cases, where several clusters overlap and create even larger clusters whose shape dominates in the Fry plot. An example of this case is shown in Figure~\ref{fig:geomthomas-overlap}. Both patterns stem from isotropic models but have a Monte Carlo $p$-value of $0.01$ for the group-wise rotation with the integral ordering. Patterns like these are more likely in case of larger $R$ which explains the empirical size of more than $0.05$ for $R=20$ while the level is met for $R=10$.

\begin{figure}
    \centering
    \includegraphics[width=0.9\textwidth]{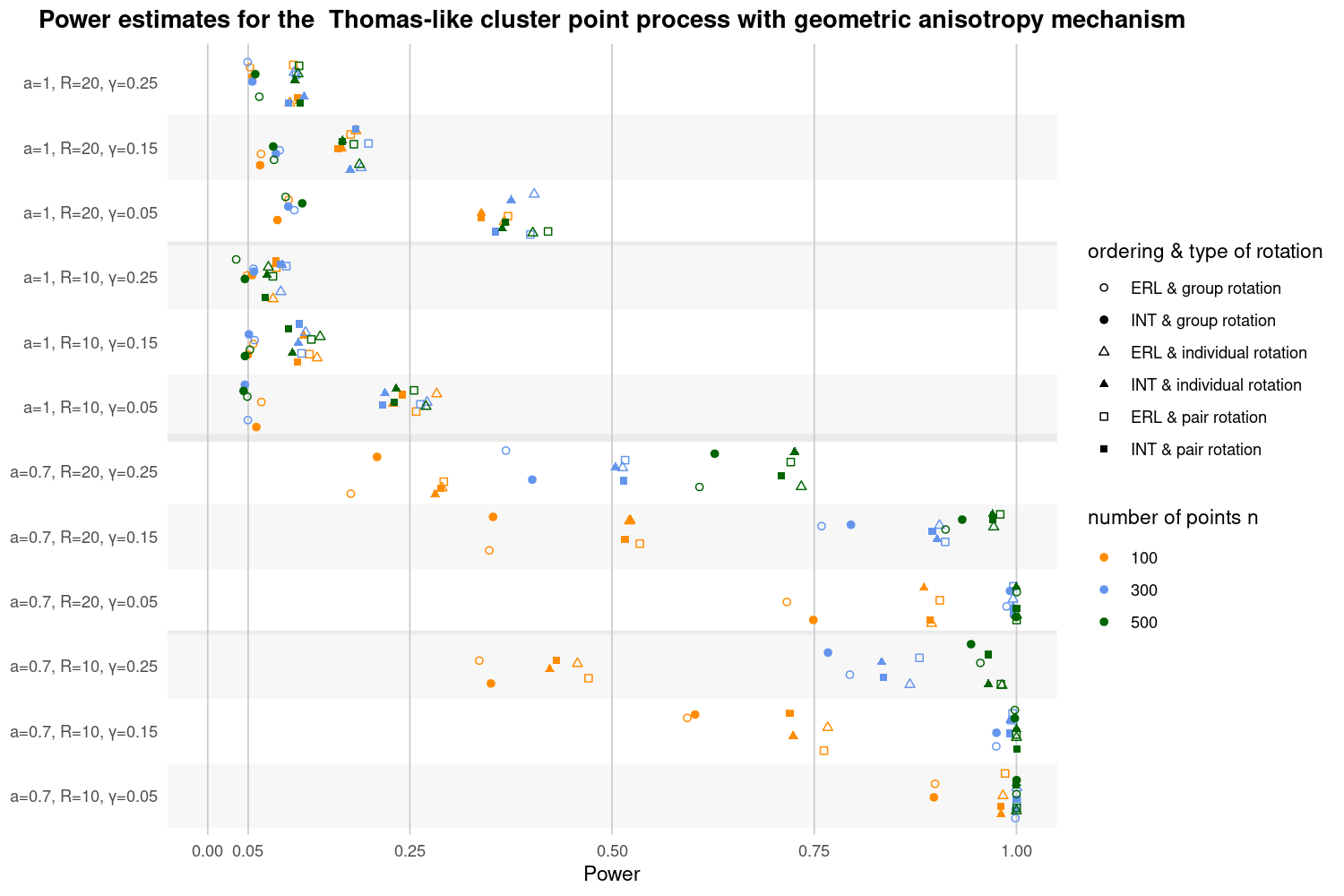}
    \caption{Power estimates for the Thomas-like cluster point process with geometric anisotropy mechanism. Each parameter combination corresponds to one horizontal block. The upper six combinations represent isotropic models, while the lower six are anisotropic.}
    \label{fig:geomthomas-results-all}
\end{figure}

For the discussion of the empirical power, we restrict to the group-wise rotation as the other two methods have no controlled type I error. Consequently, rejections of the null hypothesis are in general more likely for the other two approaches. For all six anisotropic models we observe high powers, in particular when having at least $300$ points. As for the regular process, the difference in power between the two orderings is small. Overall, the proposed tests are able to detect the geometric anisotropy even for many large but sparse clusters. 

\begin{figure}[ht]
\captionsetup[subfloat]{aboveskip=4pt, belowskip=4pt, labelformat=empty}
	\subfloat[large $R=20$]{\includegraphics[width=0.245\textwidth, height=0.245\textwidth]{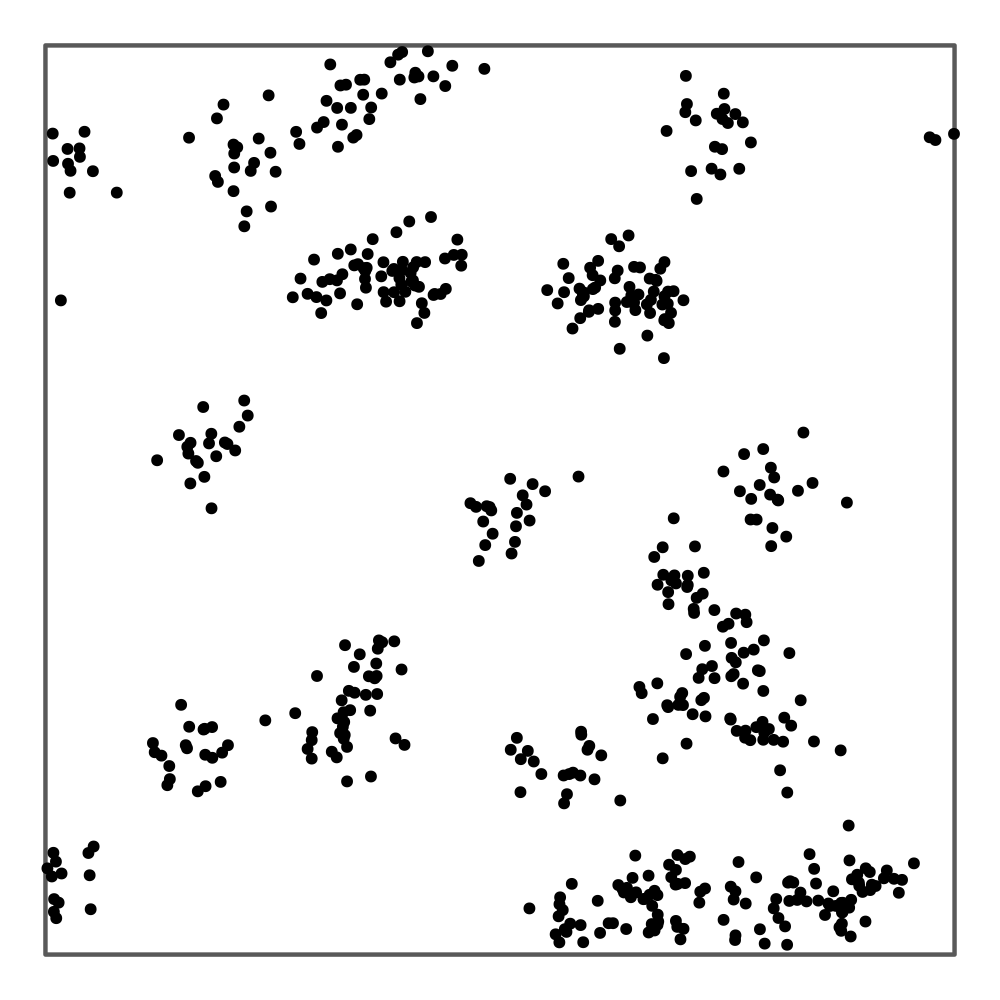}} \hfill
	\subfloat[Fry points in $b_{1.3R}(0)$]{\includegraphics[width=0.245\textwidth, height=0.245\textwidth]{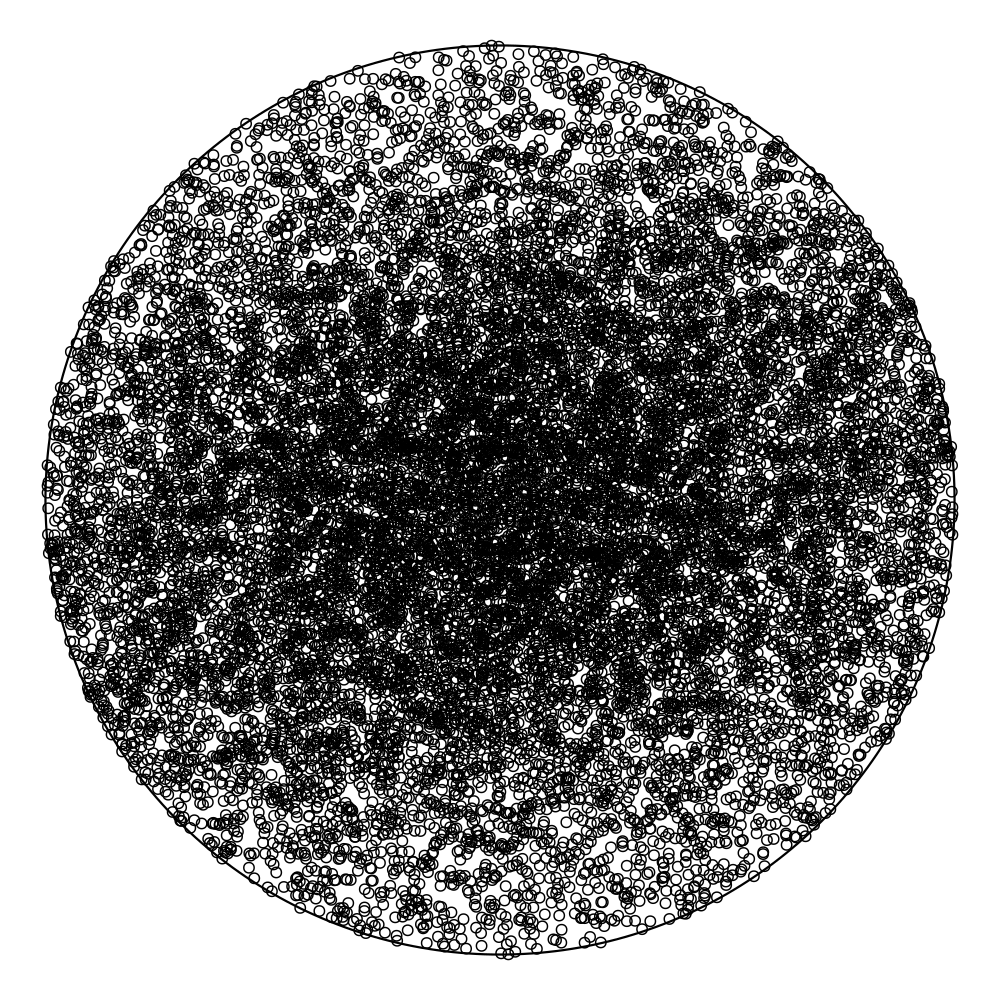}}\hfill 
	\subfloat[small $R=10$]{\includegraphics[width=0.245\textwidth, height=0.245\textwidth]{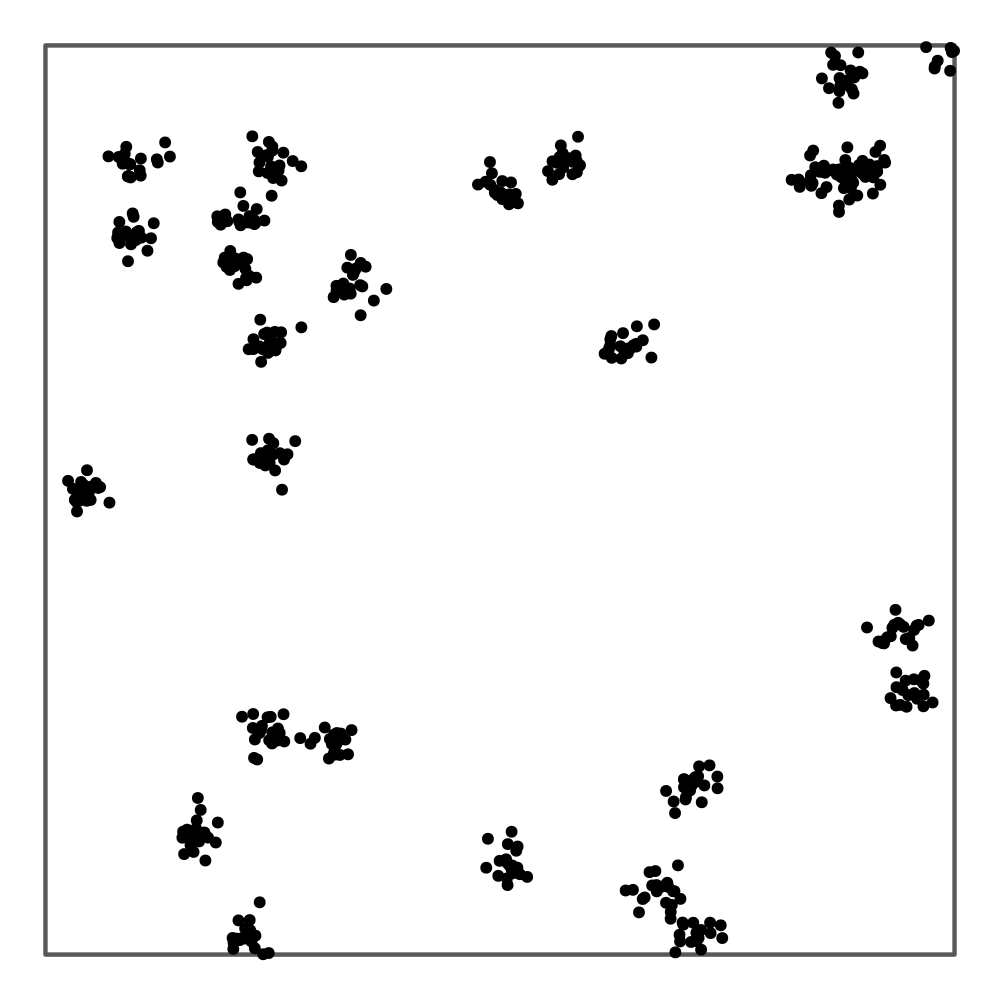}}\hfill
	\subfloat[Fry points in $b_{1.3R}(0)$]{\includegraphics[width=0.245\textwidth, height=0.245\textwidth]{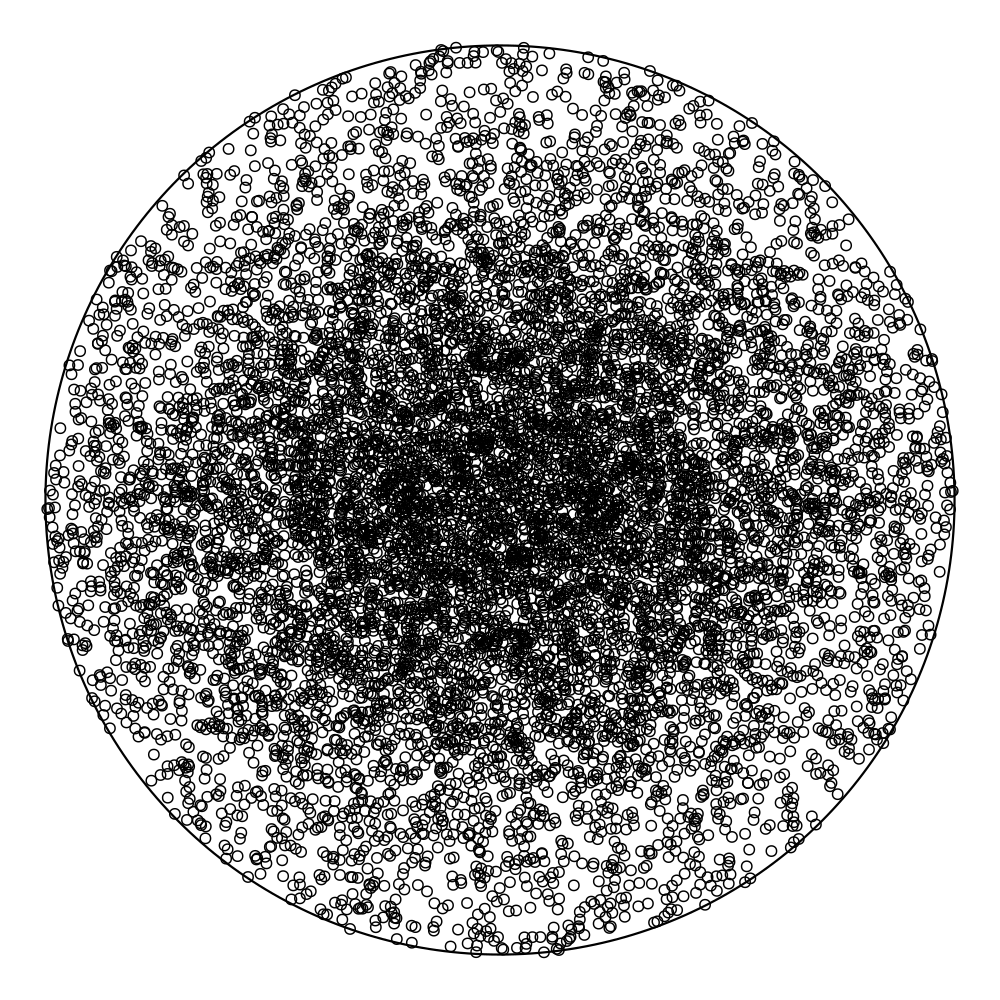}}
	\caption{One realization of the isotropic Thomas-like cluster point process for the parameter setting $a=1$, $R=20$, $\gamma=0.05$ (left) and $a=1$, $R=10$, $\gamma=0.05$ (right). Realizations have around $500$ points. In both Fry plots we observe a horizontal accumulation and thus the visible cluster in the center no longer looks rotation invariant. In both cases, the proposed test rejects the null hypothesis of isotropy with a $p$-value of $0.01$.}
	\label{fig:geomthomas-overlap}
\end{figure}

\subsubsection{Poisson line cluster point process with von Mises-Fisher distributed line directions}

The results for the Poisson line cluster point process are shown in Figure~\ref{fig:linecluster-results-all}. The behaviour compared to the previous two model classes is a bit different. Here, anisotropy needs to be identified through the direction of the lines. The model is isotropic if the directions are uniformly distributed. 
The test based on the rotation of the Fry points clearly fails to keep the nominal level if the number of lines (controlled by $\gamma$) is too small. In this case, the sample of orientations is not sufficient to accurately represent the uniform distribution on $\So$. One can identify the individual lines in the Fry plot such that the distribution looks rather discrete, see Figure~\ref{fig:fry-group-rot-line}. In case of isotropy, the visual lines in the Fry plot are less clustered than in the anisotropic case, but nevertheless it is far from rotational invariance. The samples obtained through the group-wise rotation blur these lines as shown in Figure~\ref{fig:fry-group-rot-line}. Consequently, in the Monte Carlo test even the isotropic patterns often show a more extreme structure compared to the bootstrap samples. For this reason, the rejection rate of the isotropic patterns is extremely high.
The problem gets worse, if we increase the number of points as in our parametrization this only changes the number of points per line. With more points, the linear structures within the Fry plots are even more pronounced as shown in Figure~\ref{fig:linecluster-discretization}.

As for the anisotropic patterns, we are able to detect the anisotropy in almost all the cases considered. However, as discussed, this comes at the cost of not having a controlled type I error. 

Based on this study, one should avoid using the proposed nonparametric test by random rotations when only a few collinear sets of points are visible as one cannot trust the estimated $p$-value at all. The test will only be able to detect the existence of the collinear structures but no differences in the angular distributions.

\begin{figure}[ht]
\captionsetup[subfloat]{aboveskip=4pt, belowskip=4pt, labelformat=empty}
	\subfloat[isotropic with $300$ points]{\includegraphics[width=0.245\textwidth, height=0.245\textwidth]{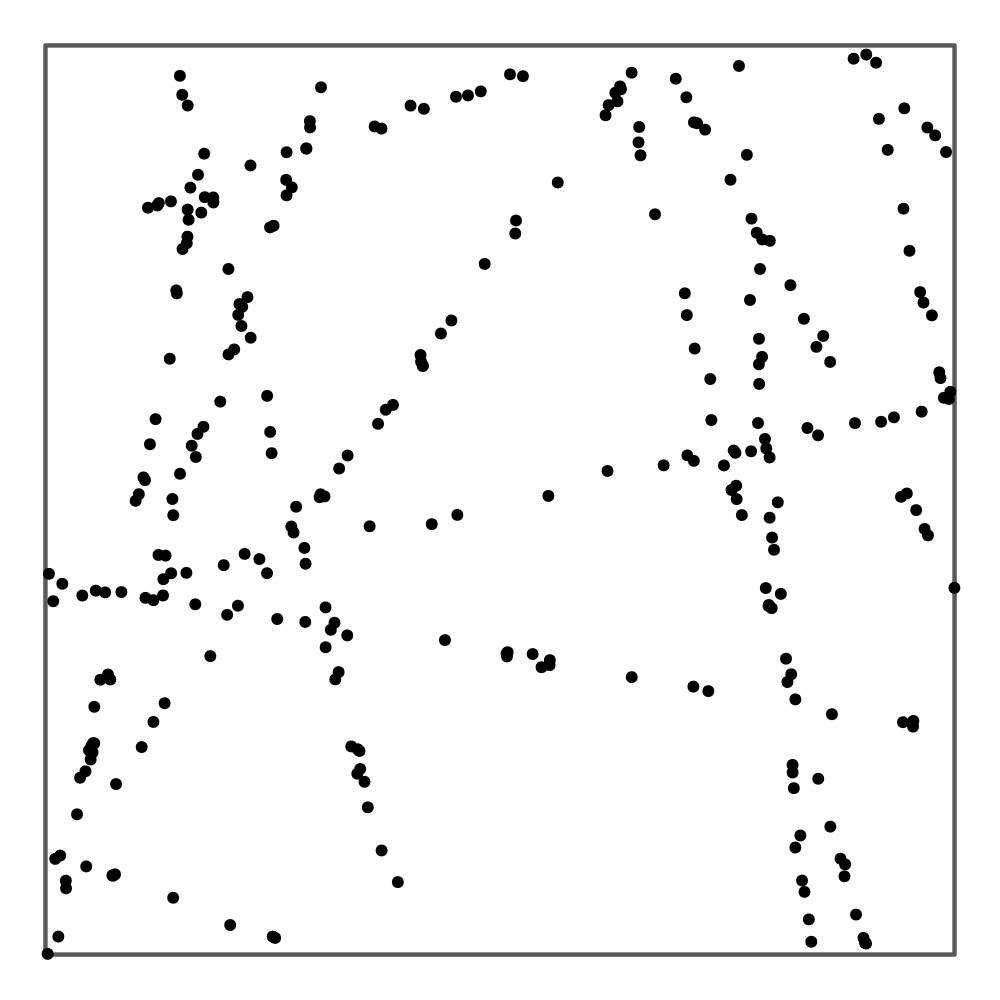}} \hfill
	\subfloat[Fry points in $b_{26}(0)$]{\includegraphics[width=0.245\textwidth, height=0.245\textwidth]{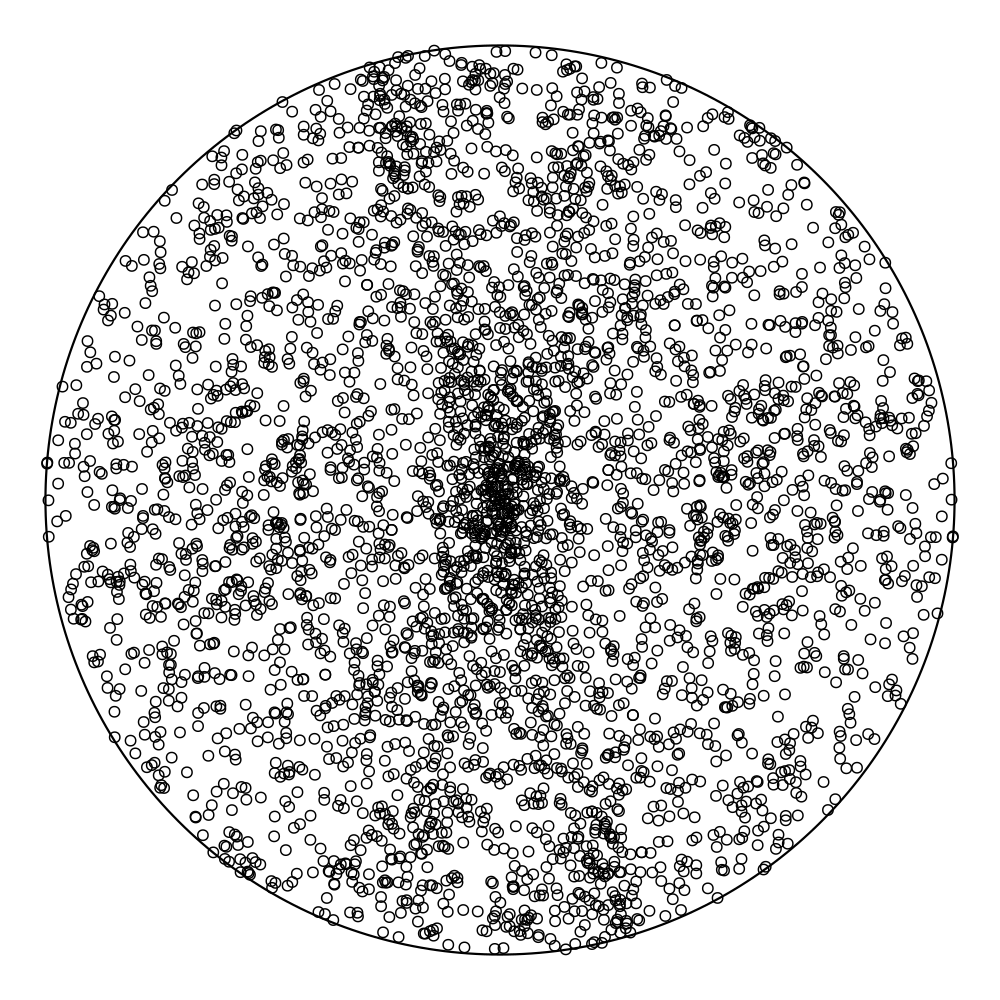}}\hfill 
	\subfloat[isotropic with $500$ points]{\includegraphics[width=0.245\textwidth, height=0.245\textwidth]{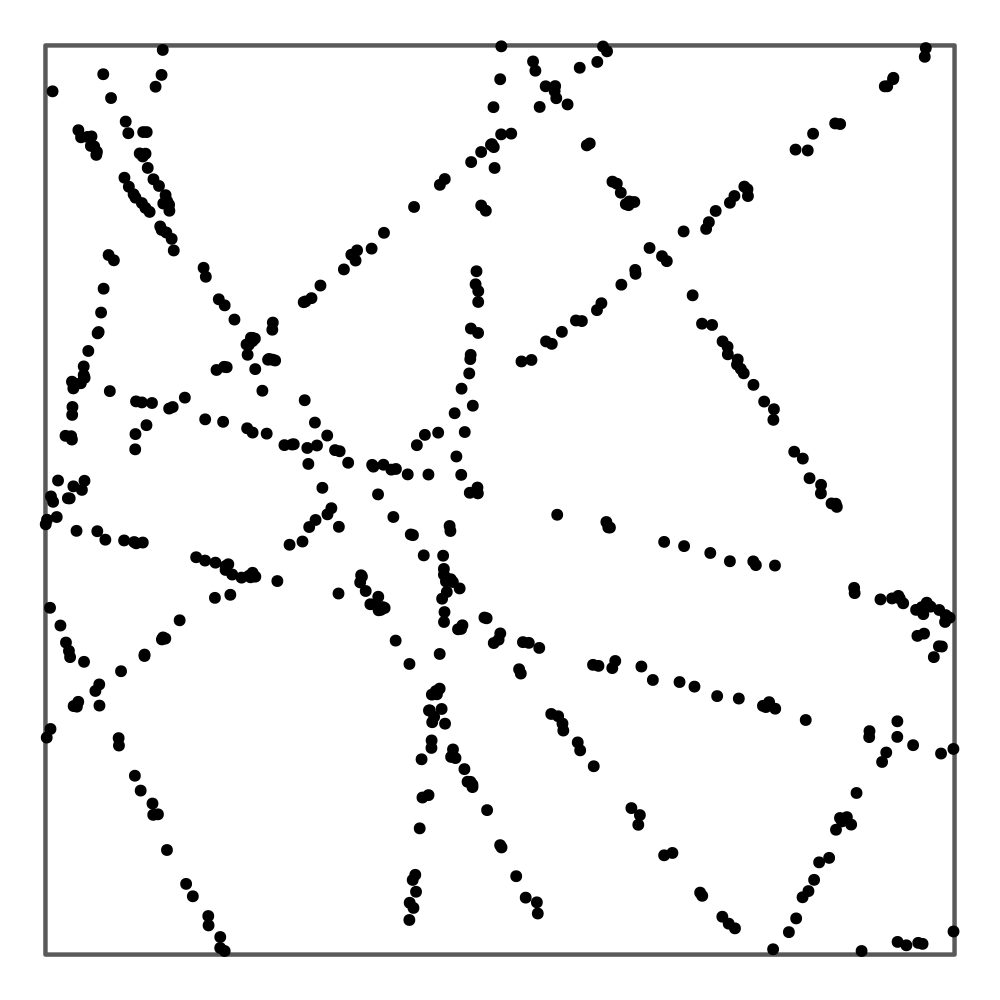}}\hfill
	\subfloat[Fry points in $b_{26}(0)$]{\includegraphics[width=0.245\textwidth, height=0.245\textwidth]{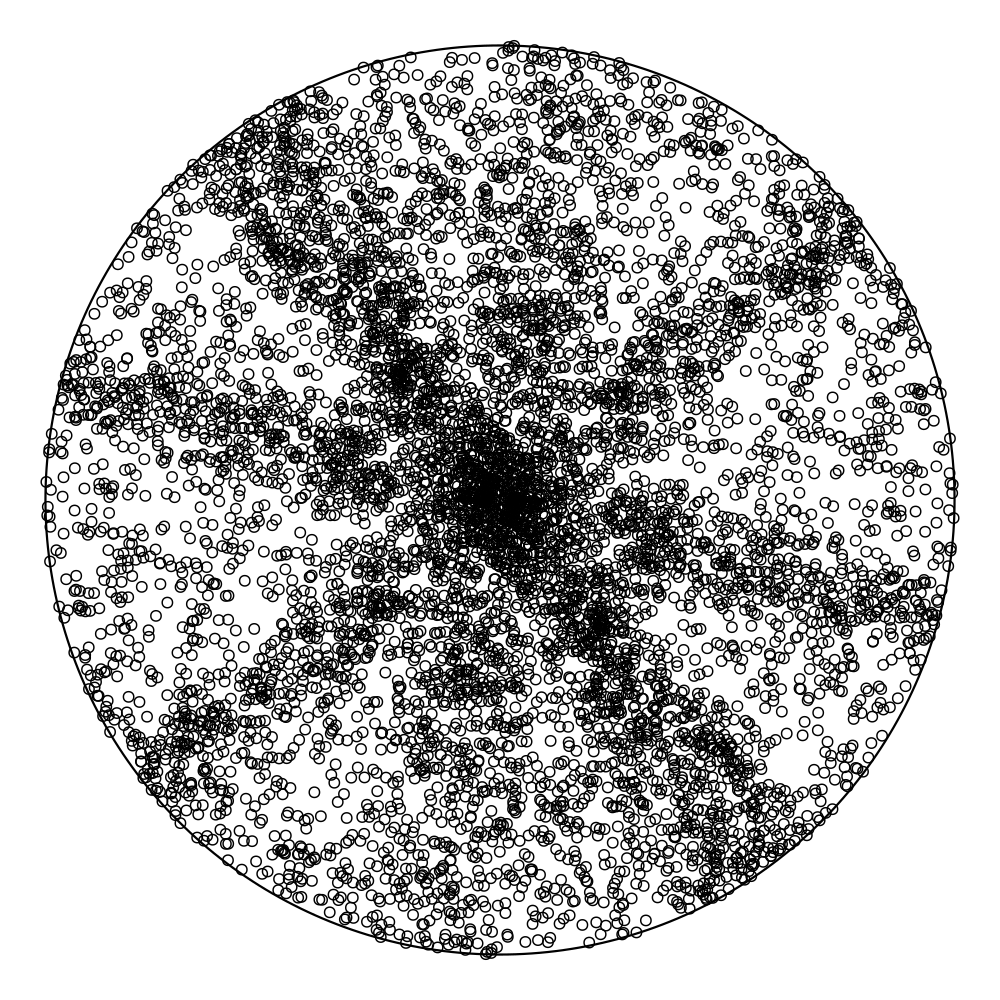}}
	\caption{One realization of the Poisson line cluster point process for the isotropic parameter setting $a=1$, $R=1$, $\gamma=0.4$ for $300$ points (left) and $500$ points (right). In both Fry plots we can identify the lines. The proposed test with group-wise rotation and integral ordering rejects the null hypothesis of isotropy in both cases with a $p$-value of $0.01$.}
	\label{fig:linecluster-discretization}
\end{figure}

\subsubsection{Matérn-like cluster point process with elliptical clusters with preferred direction}

The results of the Matérn-like cluster process with anisotropic clusters having a preferred direction is shown in Figure~\ref{fig:matern-results-all}. This model has elliptical clusters whose main direction is either uniformly distributed ($a=1$) or has a preferred direction of $\pi/3$ ($a=0.7$). 
Similar to the line cluster process, anisotropy in this model is induced by the orientation of the ellipses. Again, in case of only a few clusters ($\gamma=0.05$) the sample of orientations is too small to be able to identify the uniform distribution, see Figure~\ref{fig:matern-discretization}. Consequently, the null hypothesis of isotropy is rejected far too often. With an increasing number of clusters (i.e. increasing $\gamma$), also the empirical size goes down to the nominal level. Thus, if the anisotropy lies in the directional distribution of clusters, we obtain valid tests only if the number of clusters is sufficiently high. This holds for both types of orderings.

As for the cluster point process with geometric anisotropy mechanism the group-wise rotation yields the smallest type I error probabilities which are the closest to the nominal level. But consequently, it results in smaller powers for the anisotropic models in comparison with the two other rotation techniques. Nonetheless, the empirical power using the group-wise rotation is still high, and the anisotropy induced by the von Mises-Fisher distribution is clearly detected.

\begin{figure}[ht]
\captionsetup[subfloat]{aboveskip=4pt, belowskip=4pt, labelformat=empty}
	\subfloat[isotropic $a=1.0$]{\includegraphics[width=0.245\textwidth, height=0.245\textwidth]{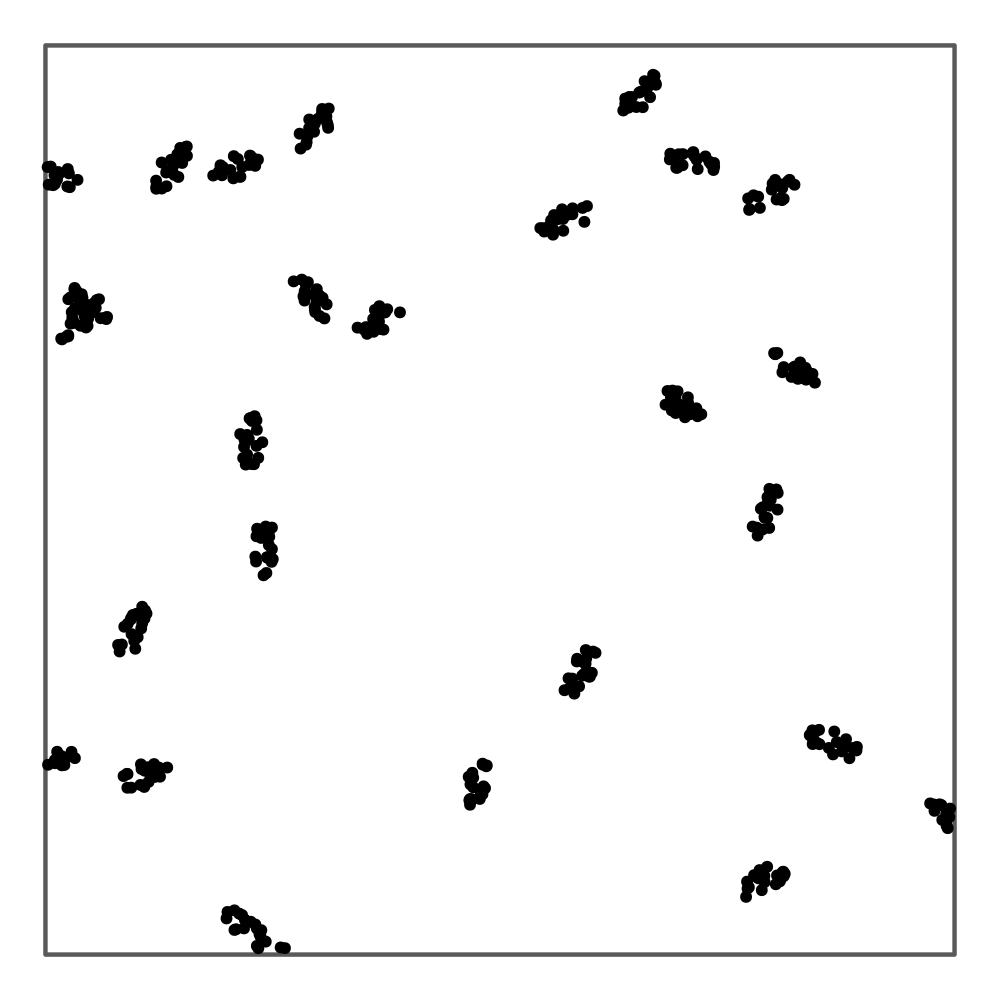}} \hfill
	\subfloat[Fry points in $b_{1.3R}(0)$]{\includegraphics[width=0.245\textwidth, height=0.245\textwidth]{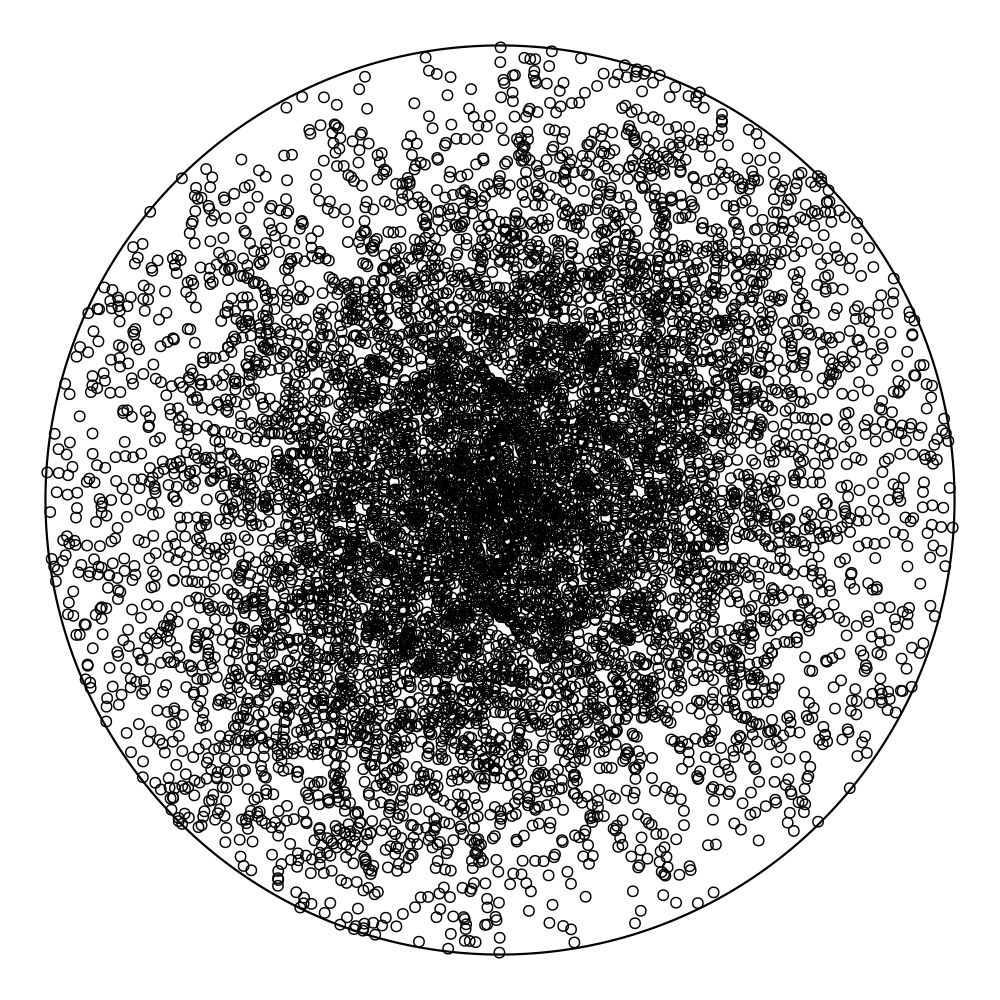}}\hfill 
	\subfloat[anisotropic $a=0.7$]{\includegraphics[width=0.245\textwidth, height=0.245\textwidth]{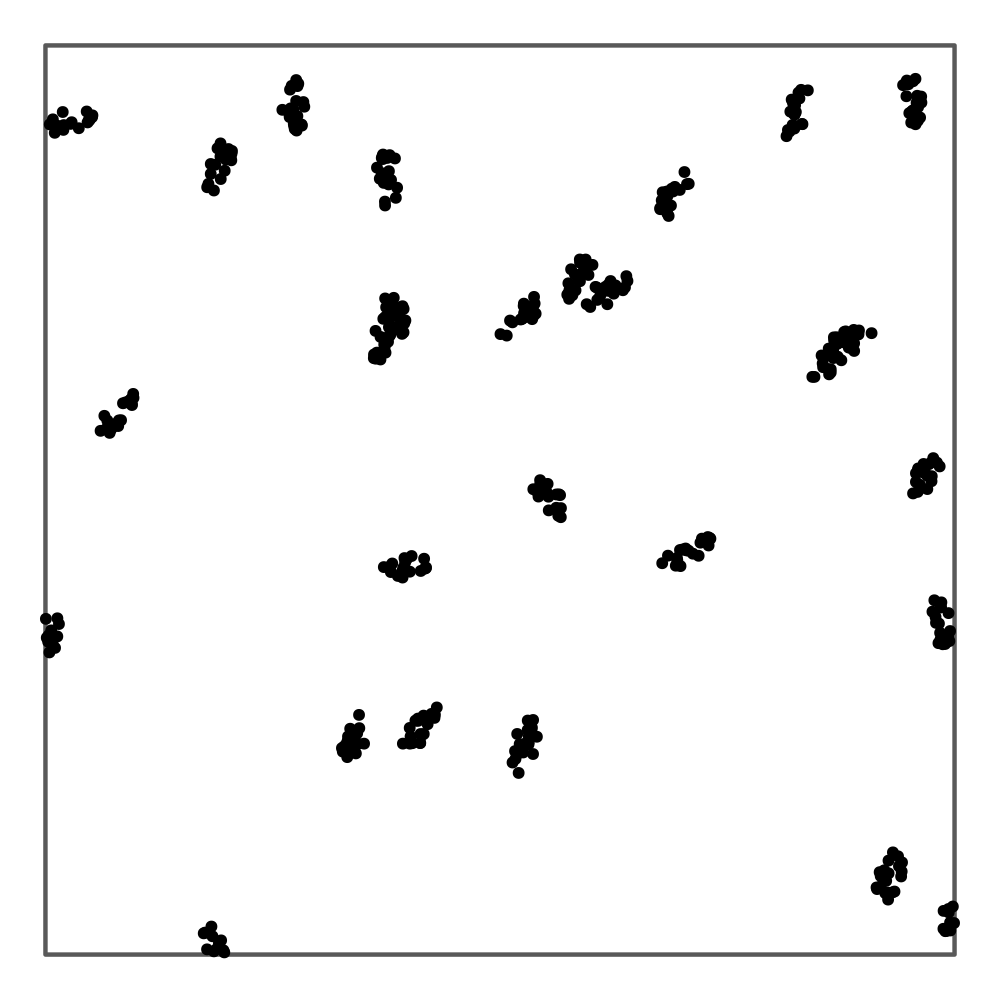}}\hfill
	\subfloat[Fry points in $b_{1.3R}(0)$]{\includegraphics[width=0.245\textwidth, height=0.245\textwidth]{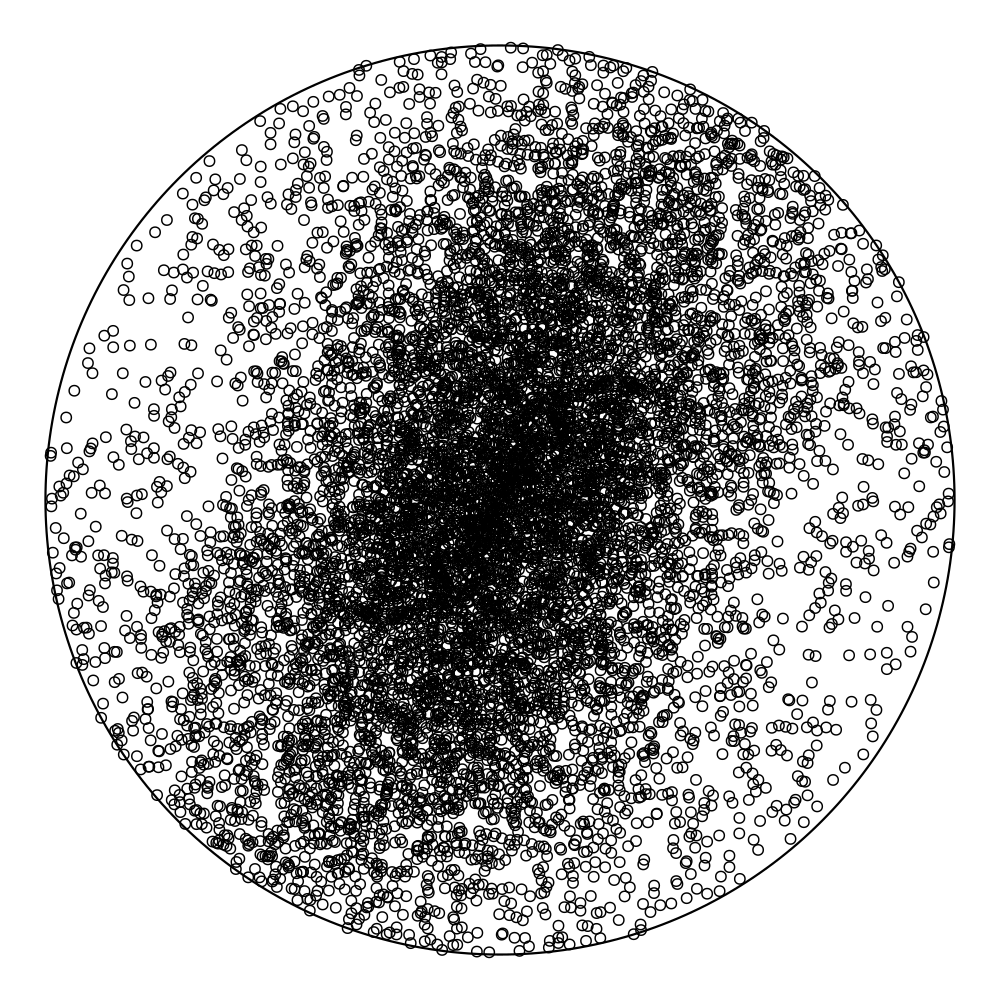}}
	\caption{One realization of the Matérn-like cluster point process with elliptical clusters for the parameter setting $a=1$, $R=10$, $\gamma=0.05$ (left) and $a=0.7$, $R=10$, $\gamma=0.05$ (right). Realizations have around $500$ points. In both Fry plots we observe a non-circular cluster in the center. The proposed test with group-wise rotation and integral ordering rejects the null hypothesis of isotropy in both cases with a $p$-value of $0.01$.}
	\label{fig:matern-discretization}
\end{figure}

\begin{figure}
    \centering
    \includegraphics[width=0.9\textwidth]{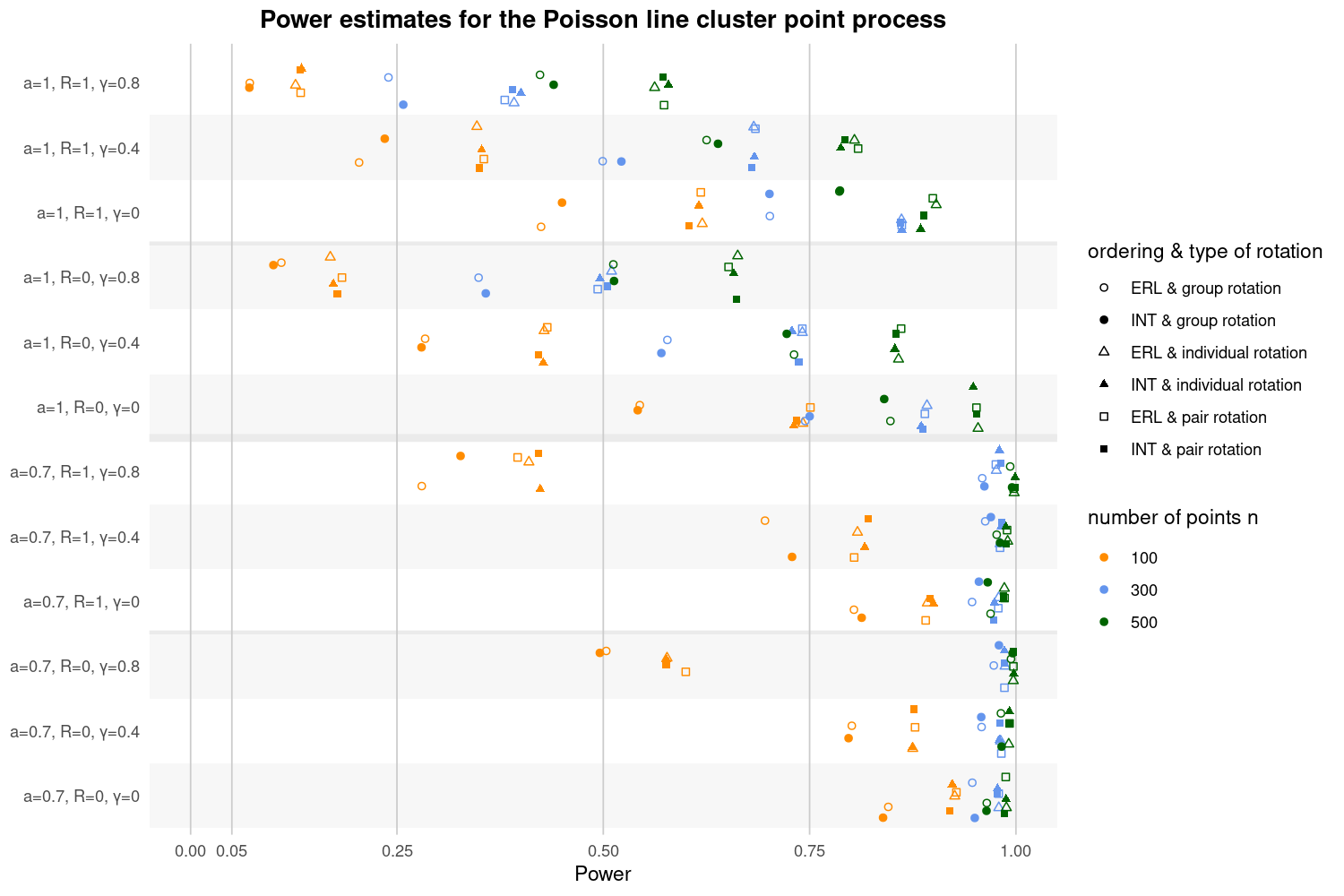}
    \caption{Power estimates for the Poisson line cluster point process. Each parameter combination corresponds to one horizontal block. The upper six combinations represent isotropic models, while the lower six are anisotropic.}
    \label{fig:linecluster-results-all}
\end{figure}

\begin{figure}
    \centering
    \includegraphics[width=0.9\textwidth]{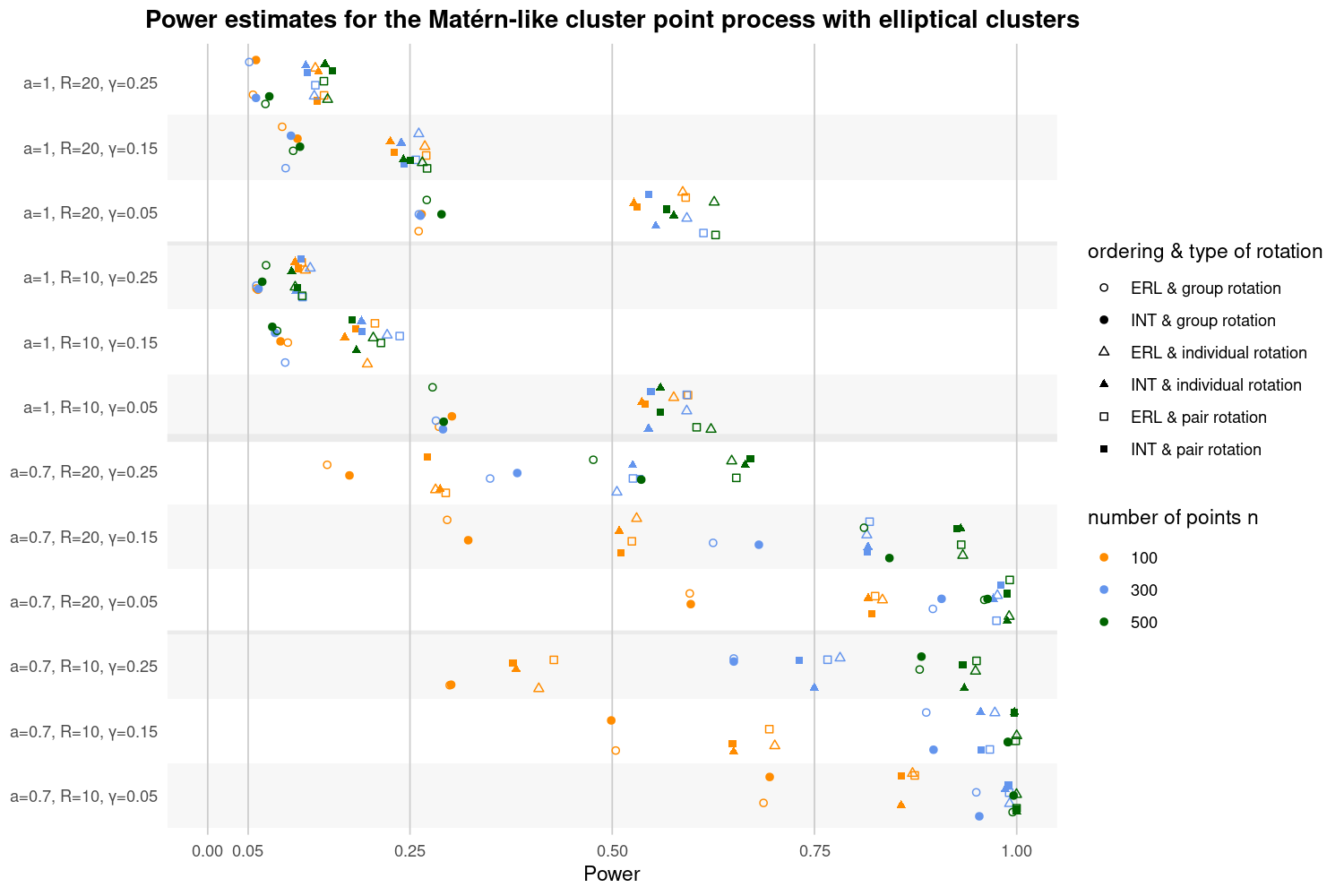}
    \caption{Power estimates for the Matérn-like cluster process with elliptical clusters. Each parameter combination corresponds to one horizontal block. The upper six combinations represent isotropic models, while the lower six are anisotropic.}
    \label{fig:matern-results-all}
\end{figure}

\FloatBarrier

% %%%%%%%%%%%%%  Real data %%%%%%%%%%%%%%%%%%%%%%%%%%%%%%%%%%%%%%%%%%%%
\section{Example: Amacrine cells data} \label{sec:real-data}

In this section, we will briefly discuss how the nonparametric isotropy test by random rotations can be applied to real data. As we have seen that the test worked well for regular patterns, we decided to use the \emph{amacrine cells} data set that was introduced in \citet{baddeley_modelling_2006} and which is available in the R-package \texttt{spatstat} \citep{spatstat}. The same data set has also been investigated with regard to isotropy in \citet{wong_isotropy_2016} and \citet{rajala_tests_2022} which allows for a comparison with our proposed resampling approach.

The data consists of a marked point pattern in the observation window $ [0, 1.6012085] \times [0,1]$ where one unit corresponds to $662~\mu\textrm{m}$. In total there are $294$ points, $152$ of them marked as \emph{on} and the other $142$ as \emph{off}. \citet{rajala_tests_2022} investigated only the pattern containing the \emph{off} cells, while \citet{wong_isotropy_2016} considered the total unmarked pattern as well as the patterns for each of the two marks.

As in \citet{rajala_tests_2022}, we choose the two directions for the contrast statistic of the sector $K$-function based on the shape of the elliptical central void in the Fry plots. This yields for the unmarked pattern the angles $\alpha_1 = -45$° and $\alpha_2 = 45$°, for the \emph{on} cells $\alpha_1 = -10$° and $\alpha_2 = 80$° and for the \emph{off} cells $\alpha_1 = 60$° and $\alpha_2 = 150$°. As in the simulation study, we keep the half-opening angle of $\varepsilon = \frac{\pi}{4}$ for the sector $K$-function. Figure~\ref{fig:amacrine} shows the marked point pattern as well as the Fry plots overlaid with the sectors. Based on these plots, we use several $r_{\max}$ values between $0.08$ and $0.12$.
 
\begin{figure}[th]
\centering
    \subfloat[Amacrine cell data with cells marked \emph{on} as filled circles $\bullet$ and \emph{off} cells as unfilled circles $\circ$.]{\includegraphics[width=0.4\textwidth, valign=c]{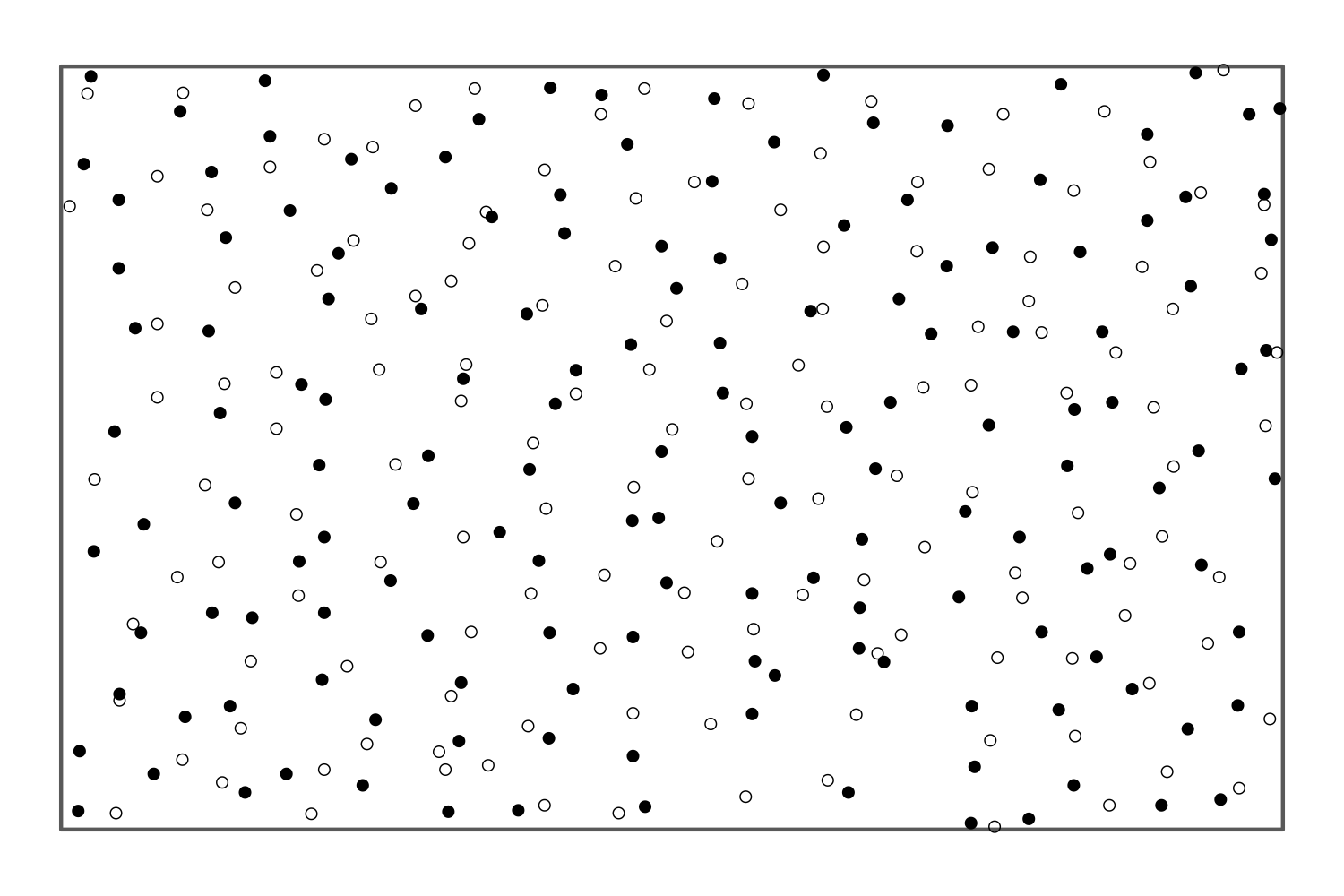}}\hfill
     \subfloat[both cells]{\includegraphics[width=0.19\textwidth, valign=c]{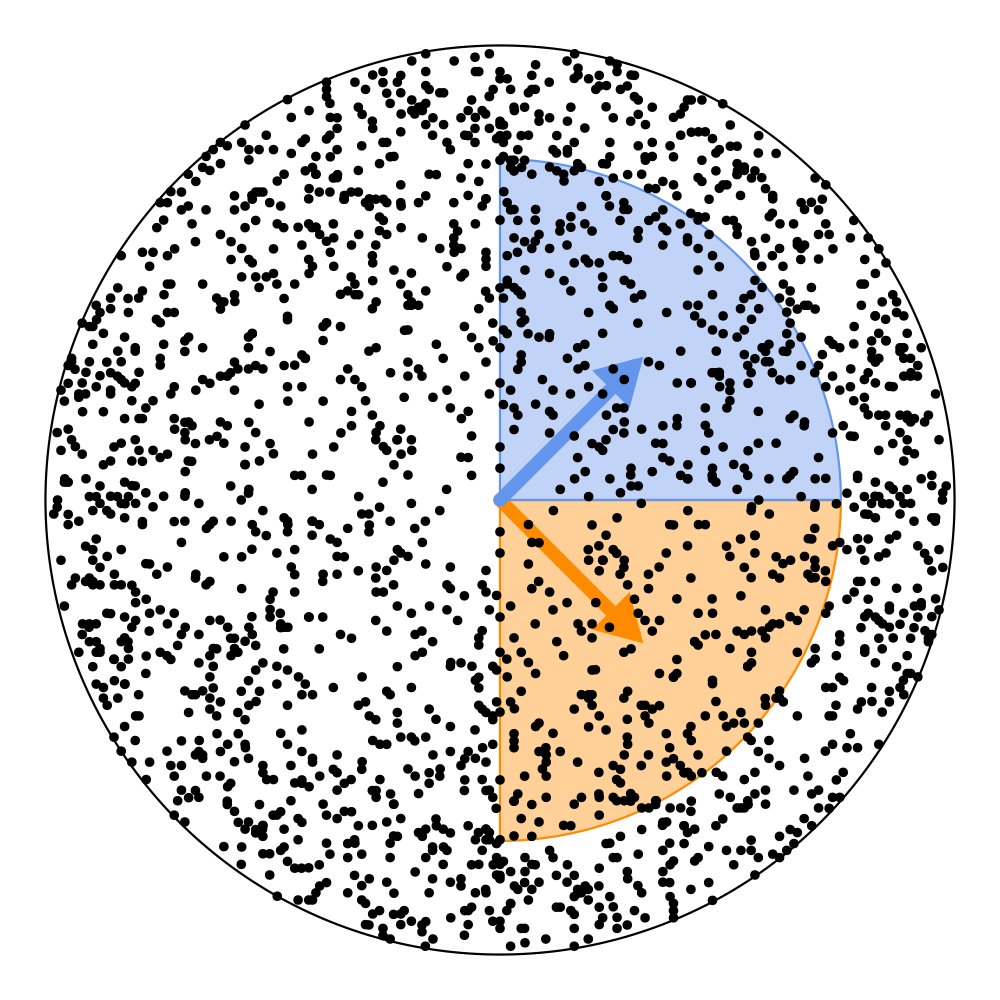}
     \vphantom{\includegraphics[width=0.4\textwidth, valign=c]{img/real-data/amacrine-marked.png}}}\hfill
 \subfloat[\emph{on} cells]{\includegraphics[width=0.19\textwidth, valign=c]{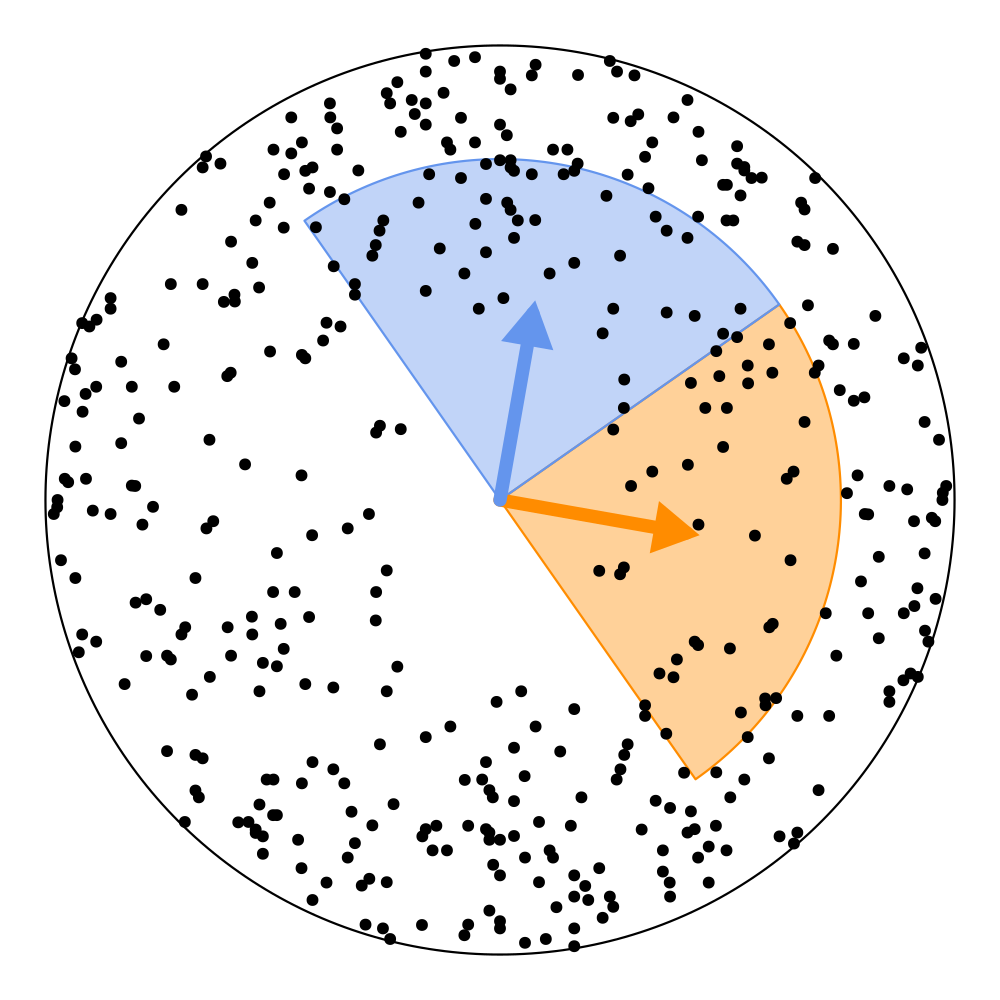}\vphantom{\includegraphics[width=0.4\textwidth, valign=c]{img/real-data/amacrine-marked.png}}}\hfill
 \subfloat[\emph{off} cells]{\includegraphics[width=0.19\textwidth, valign=c]{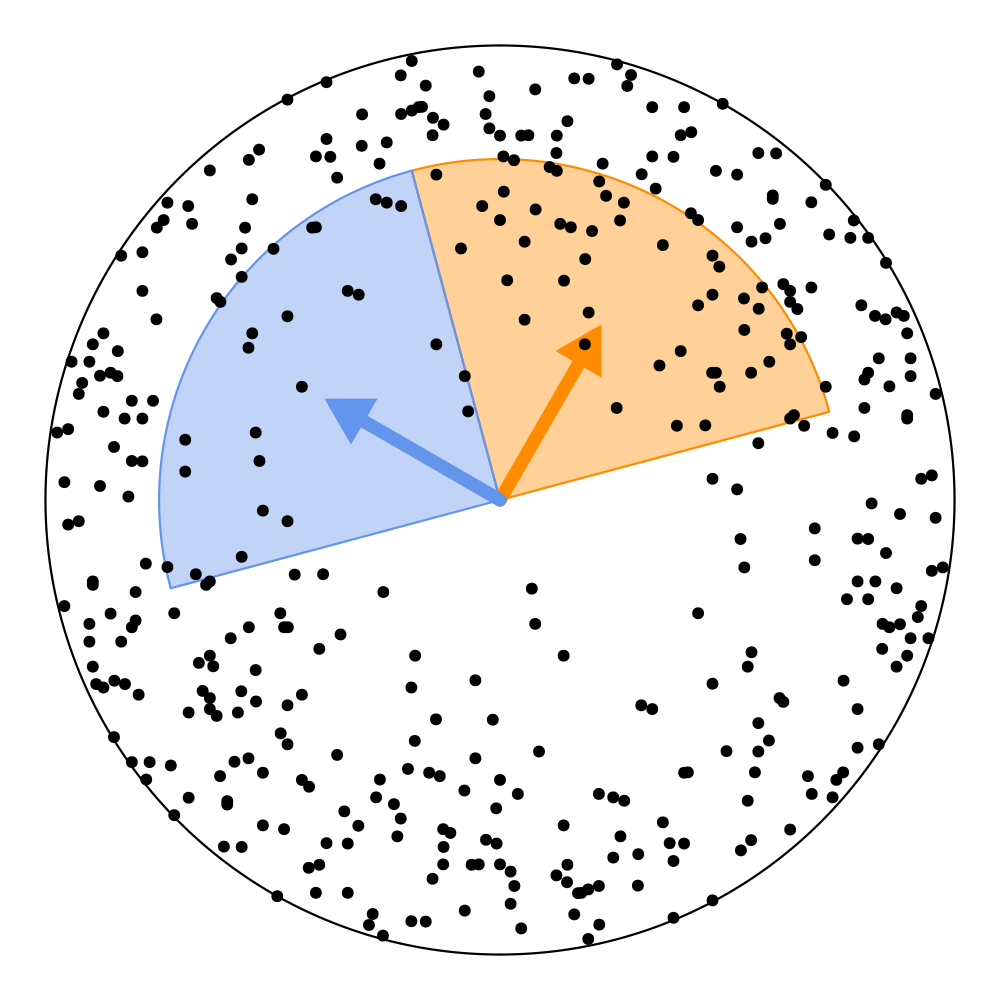}\vphantom{\includegraphics[width=0.4\textwidth, valign=c]{img/real-data/amacrine-marked.png}}}
    \caption{Overview of the amacrine cell data set. (a) Marked point pattern. (b) - (d) Fry points with norm bounded by $0.12$. The sector $S(\alpha_1, \frac{\pi}{4}, 0.09)$ is shown in orange and in  $S(\alpha_2, \frac{\pi}{4}, 0.09)$ in blue. Values of $\alpha_1, \alpha_2$ are given in the text.}
    \label{fig:amacrine}
\end{figure}

Due to the stochasticity in the Monte Carlo test, we run all tests $1000$ times. Table~\ref{tab:amacrine-results} reports the average $p$-values including the standard deviation.

\begin{table*}[th]\centering
\renewcommand{\arraystretch}{1.5}
\begin{tabularx}{\textwidth}{Xrccccc} \toprule
& & \multicolumn{5}{c}{$r_{\max}$} \\ \cmidrule(l{8pt}){3-7}
& $M$ & $0.08$ & $0.09$ & $0.10$ & $0.11$ & $0.12$ \\ \midrule
\multicolumn{7}{@{}l}{\textbf{\emph{on} cells}} \\
integral ordering & $99$ & $0.292$ ($0.046$) & $0.416$ ($0.047$) &$ 0.400$ ($0.049
$) & $0.296$ ($0.045$) &$0.108$ ($0.030$)\\
& $499$ & $0.288$ ($0.020$) & $0.414$ ($0.022$) & $0.393$ ($0.022$)& $0.291$ ($0.020$)& $0.101$ ($0.013$) \\
ERL ordering & $99$& $0.187$ ($0.096$) & $0.199$ ($0.106$) & $0.216$ ($0.115$) & $0.183$ ($0.075$) & $0.074$ ($0.044$)\\
& $499$ & $0.151$ ($0.056$) & $0.160$ ($0.057$) & $0.179$ ($0.063$)& $0.145$ ($0.049$) & $0.046$ ($0.028$)\\
\multicolumn{7}{@{}l}{\textbf{\emph{off} cells}} \\
integral ordering  & $99$& $0.178$ ($0.038$) & $0.051$ ($0.019$)& $0.017$ ($0.008
$)& $0.012$ ($0.004$)& $0.011$ ($0.003$) \\
& $499$ & $0.170$ ($0.017$)& $0.043$ ($0.009$) & $0.009$ ($0.004$) & $0.004$ ($0.002$) & $0.003$ ($0.001$) \\
ERL ordering & $99$& $0.215$ ($0.080$) & $0.081$ ($0.053$) & $0.029$ ($0.024$) &$0.016$ ($0.011$) & $0.013$ ($0.013$)\\
& $499$ & $0.142$ ($0.049$) & $0.048$ ($0.024$) & $0.012$ ($0.011$) & $0.007$ ($0.008$)& $0.004$ ($0.006$)\\
\multicolumn{7}{@{}l}{\textbf{unmarked pattern}} \\
integral ordering  & $99$& $0.249$ ($0.041$)& $0.257$ ($0.047$)& $0.236$ ($0.043$) & $0.229$ ($0.042$)& $0.183$ ($0.036$)\\
& $499$ & $0.243$ ($0.020$) & $0.251$ ($0.020$) & $0.229$ ($0.018$) & $0.222$ ($0.018$)& $0.174$ ($0.016$)\\
ERL ordering & $99$ & $0.614$ ($0.117$)& $0.637$ ($0.121$) & $0.655$ ($0.127$) & $0.615$($0.129$) & $0.545$ ($0.129$)\\
& $499$ & $0.634$ ($0.063$)& $0.651$ ($0.064$)& $0.699$ ($0.062$) & $0.679$ ($0.059$) & $0.569$ ($0.074$)\\
\bottomrule 
\end{tabularx}
\caption{Results of the nonparametric isotropy tests using the random group-wise rotation for the amacrine cell data set. Shown are the average $p$-values and in parentheses the empirical standard deviations of $1000$ Monte Carlo tests. Every individual test is based on $M$ bootstrap samples.}
\label{tab:amacrine-results}
\end{table*}

Using the group-wise rotation resampling scheme in combination with the pointwise contrast of the sector $K$-functions, we obtain a series of small $p$-values for the pattern containing the \emph{off} cells. When $r_{\max} \geq 0.10$, there is strong evidence that this pattern is anisotropic. This conclusion coincides with the analysis of \citet{wong_isotropy_2016} and \citet{rajala_tests_2022} who both identified the pattern as anisotropic. 

In contrast to \citet{wong_isotropy_2016} we obtain large $p$-values for both the \emph{on} cells and the total unmarked pattern. Using our resampling scheme there is no evidence to reject the null hypothesis of isotropy in these two cases. The conclusion is also consistent over all five upper bounds that we considered.

Table~\ref{tab:amacrine-results} also highlights that the Monte Carlo $p$-value estimator obtained from $M$ bootstrap samples using the extreme rank length ordering has a higher variance compared to the one using the integral ordering. This is the case for both numbers of bootstrap samples that we investigated here. Consequently, we should use in general a higher number of samples $M$ when using the extreme rank length ordering.

% %%%%%%%%%%%%%  Discussion %%%%%%%%%%%%%%%%%%%%%%%%%%%%%%%%%%%%%%%%%%%
\section{Discussion}\label{sec:discussion}

Previous works on isotropy tests for spatial point processes often rely on the assumption that a parametric isotropic null model is available. Since this is not always the case in practice, we aimed for a new nonparametric isotropy test that can be used without any parametric model. The focus of this paper was to propose a new resampling scheme for a Monte Carlo test. Compared to the stochastic reconstruction approach by \citet{wong_isotropy_2016} we do not generate entire new point patterns but only Fry-type point patterns containing the resampled pairwise difference vectors. Our proposed random rotation approaches are also free of parameters to be selected and computationally cheap. The only parameters that one has to choose for the derived isotropy tests are the functional contrast summary statistic including its parameters, and the upper bound for and the type of ordering used for the comparison in the Monte Carlo test. These choices need to be made in any type of null hypothesis testing for point processes.

We performed an extensive simulation study investigating the empirical power and size of the tests when using the contrast of two sector $K$-functions as the contrast statistic. We compared all three proposed rotation schemes as well as two types of orderings.

In case of the regular Strauss process with geometric anisotropy mechanism, all tests performed well even though the tests with either the individual rotation or the pairwise rotation approach were sometimes conservative. The empirical sizes met the nominal level for all parameter combinations when choosing the integral ordering and the group-wise rotation. In particular the anisotropy for hard-core point processes is found already for point patterns with around $300$ points. This is an improvement compared to the results stated in \citet{wong_isotropy_2016} where it was difficult to detect anisotropic Gibbs hard-core point process.

For the two clustered processes that we considered we observed clear differences between the three rotation techniques. In particular, only with the group-wise rotation it is possible to control the type I error. The Matérn-like cluster process with elliptical clusters represents the harder problem as in this case the anisotropy is induced by a non-uniform distribution for the cluster directions. The number of clusters in our study was sometimes too small to correctly identify the uniform distribution of these directions when having an isotropic pattern. For the geometric anisotropic Thomas-like cluster process, the influence of the number of clusters is less severe in view of the empirical sizes. This is due to the fact that all clusters are aligned and only the shape of the clusters determines whether the process is isotropic or anisotropic. 

None of our proposed tests is valid if only a few linear structures are visible in the point pattern as it was the case for the Poisson line cluster point process in our study. The null hypothesis of isotropy was rejected for far too many isotropic patterns. As for the Matérn-like cluster process this was due to having too few lines, i.e. samples of the orientation distribution, in order to be able to identify the uniform distribution. Consequently, we cannot trust the almost perfect empirical powers as the reason of the rejection is not necessarily the different directional distribution. 

The same problem also occurred in the nonparametric isotropy tests of \citet{sormani_second_2020} which are based on classical uniformity tests on $\So$. The spherical tests assume an i.i.d. sample of the directional distribution which is not met in case of Fry points, in particular for clustered processes. With our proposed group-wise rotation approach we were still able to obtain valid tests for geometrically transformed anisotropic cluster processes which is again an improvement as the spherical tests are not applicable for these types of processes.

In conclusion, our resampling scheme for spatial point processes using random group-wise rotations is applicable for many types of point processes and different types of anisotropy. The obtained bootstrap samples can be used in the very general isotropy testing framework discussed in \citet{rajala_tests_2022}. Therefore one can use any summary function derived from the second-order characteristic of the point process instead of only using the sector $K$-function that was used in our study. Since our bootstrap samples from the group-wise rotation are no longer point symmetric around the origin we advise to use only directed sets which are not point symmetric when the point pattern is not very large. Otherwise possible artifacts in the observed Fry points coming from the small number of samples of the orientation distribution are enhanced. This results in higher type I error rates.

Not only other contrast statistics but also other ordering can be used. For example one can incorporate the variance of the second-order characteristic estimators. With these adjustments we expect to obtain even more powerful isotropy tests.

% %%%%%%%%%%%%%  References and acknowledgements %%%%%%%%%%%%%%%%%%%%%%
\section*{Acknowledgements}

We thank Saskia Ostmann for the implementation of the group-wise rotation scheme.

\bibliographystyle{abbrvnat}
\bibliography{references} 

\section{Appendix. Additional Figures}

\begin{figure}[th]
    \captionsetup[subfloat]{aboveskip=0.25pt}
	\subfloat[Observed Fry plot]{\includegraphics[width=0.245\textwidth]{img/fry-points/iso-cluster-fry.png}} \hfill
	\subfloat[]{\includegraphics[width=0.245\textwidth]{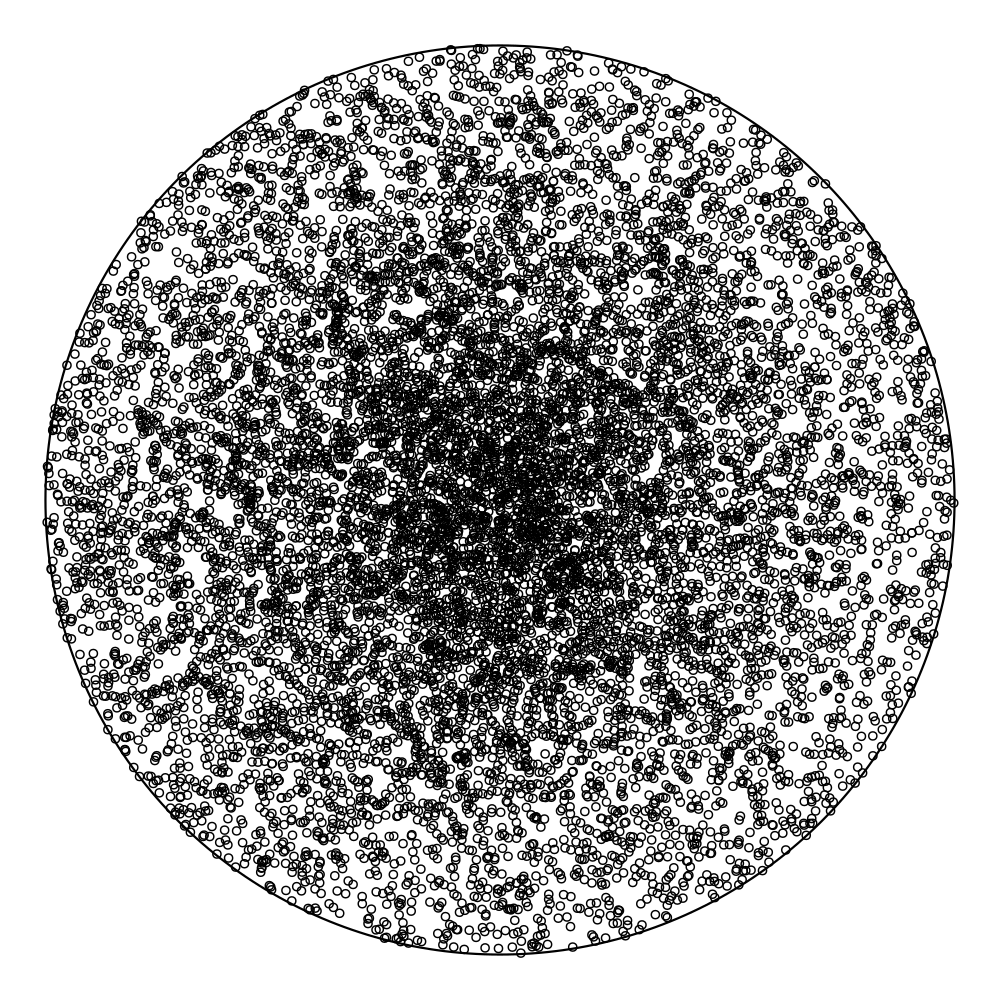}}\hfill 
	\subfloat[]{\includegraphics[width=0.245\textwidth]{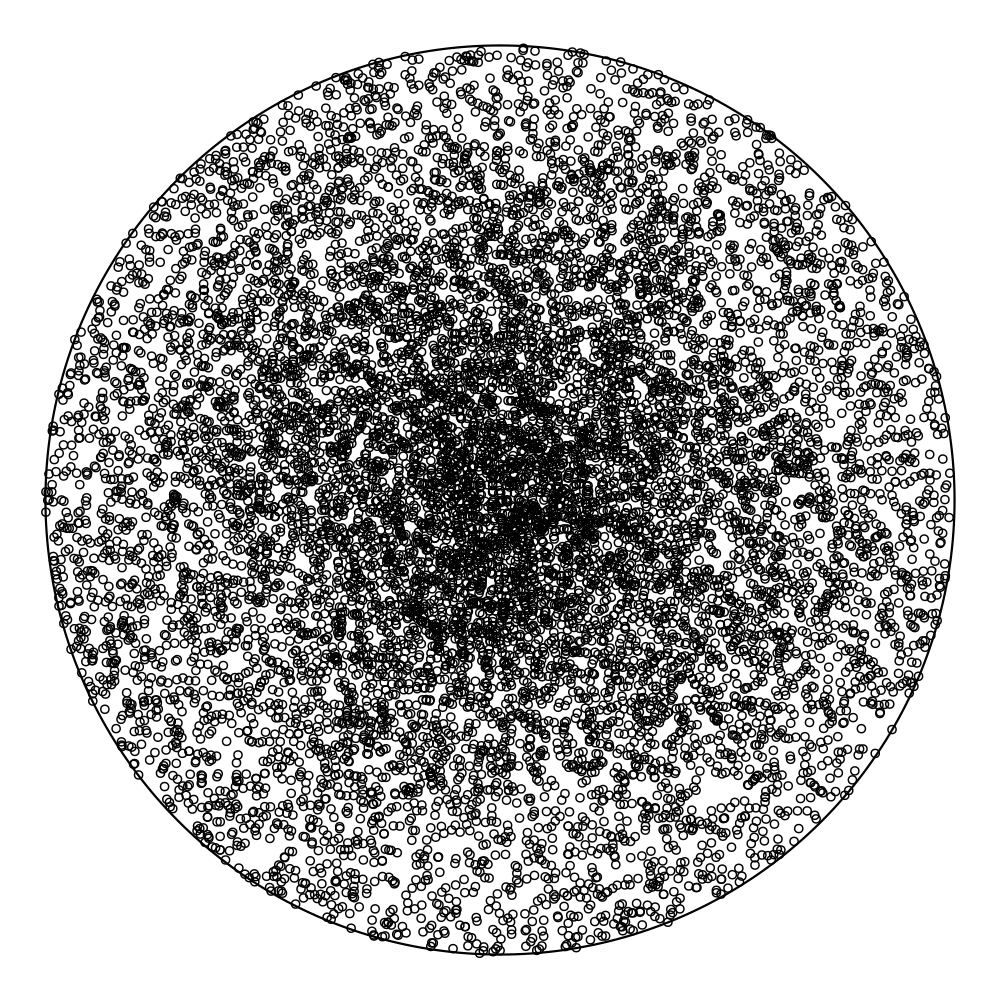}}\hfill
	\subfloat[]{\includegraphics[width=0.245\textwidth]{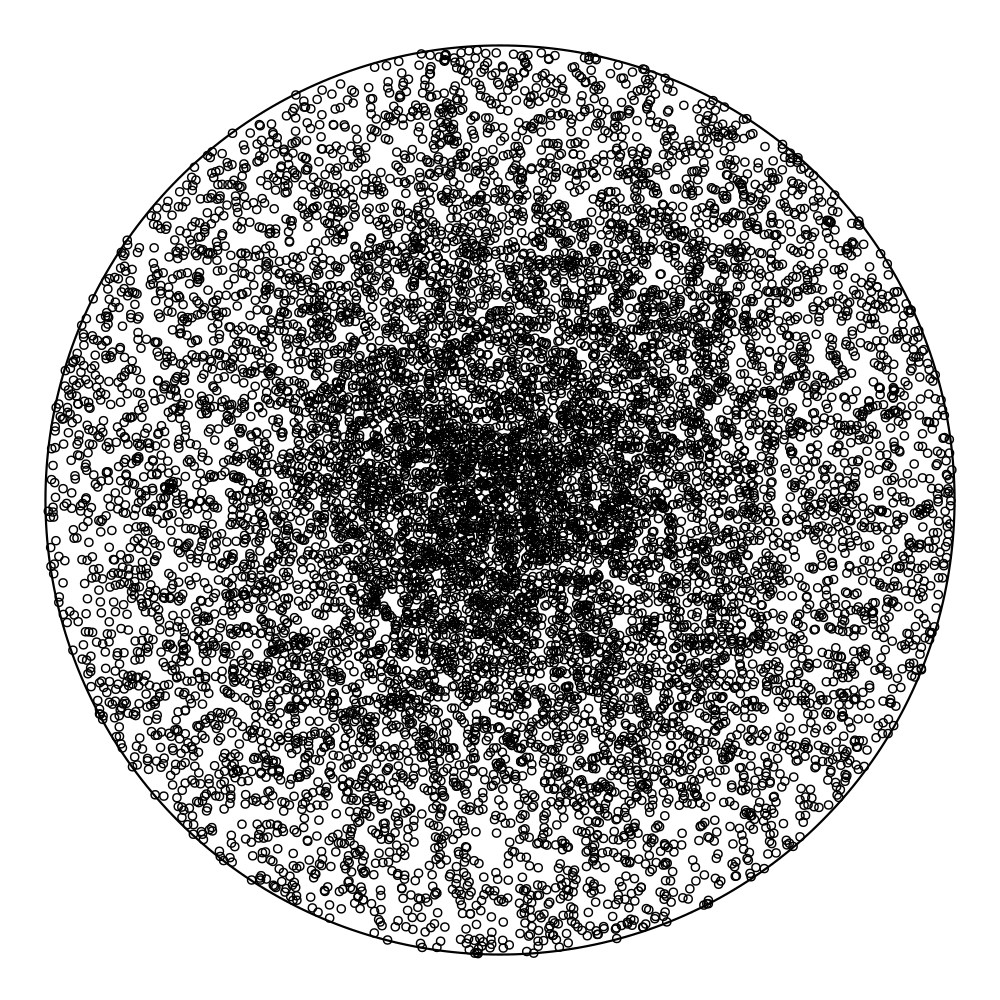}}\\
 
    \subfloat[Observed Fry plot]{\includegraphics[width=0.245\textwidth]{img/fry-points/aniso-cluster-fry.png}} \hfill
	\subfloat[]{\includegraphics[width=0.245\textwidth]{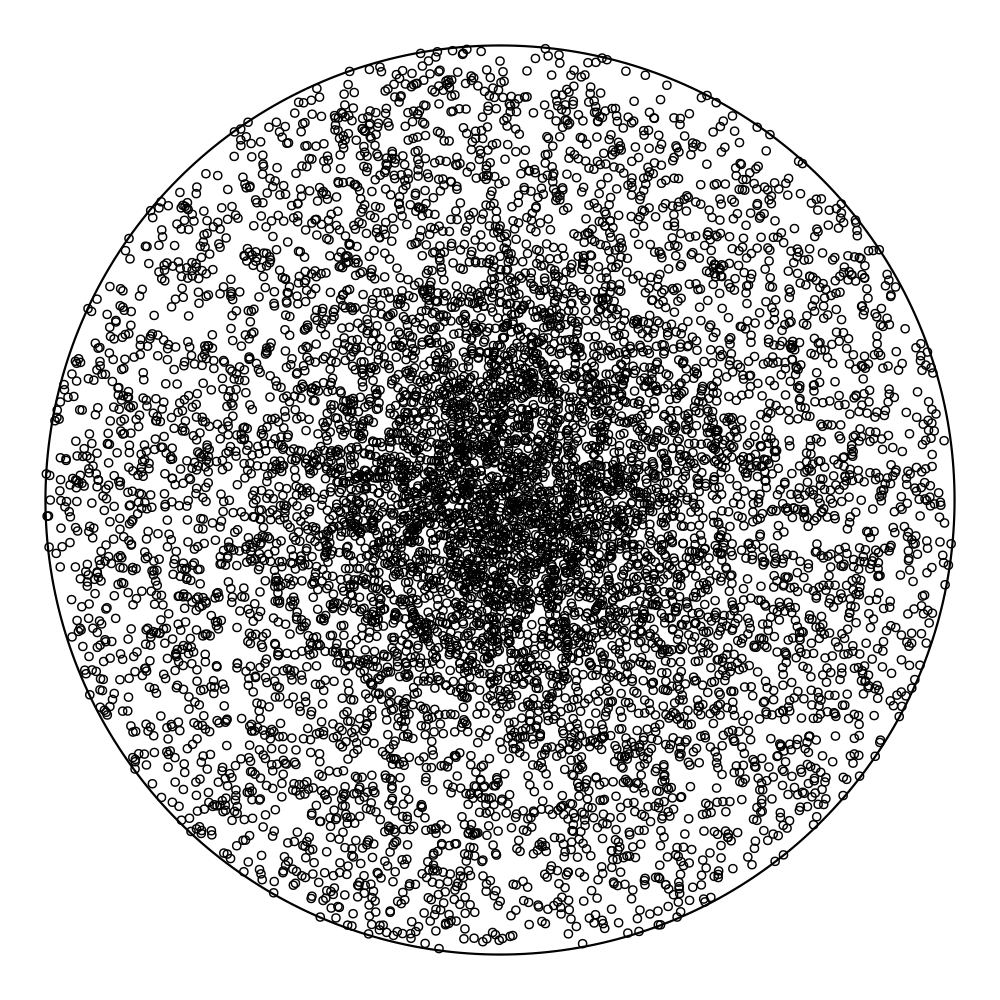}}\hfill 
	\subfloat[]{\includegraphics[width=0.245\textwidth]{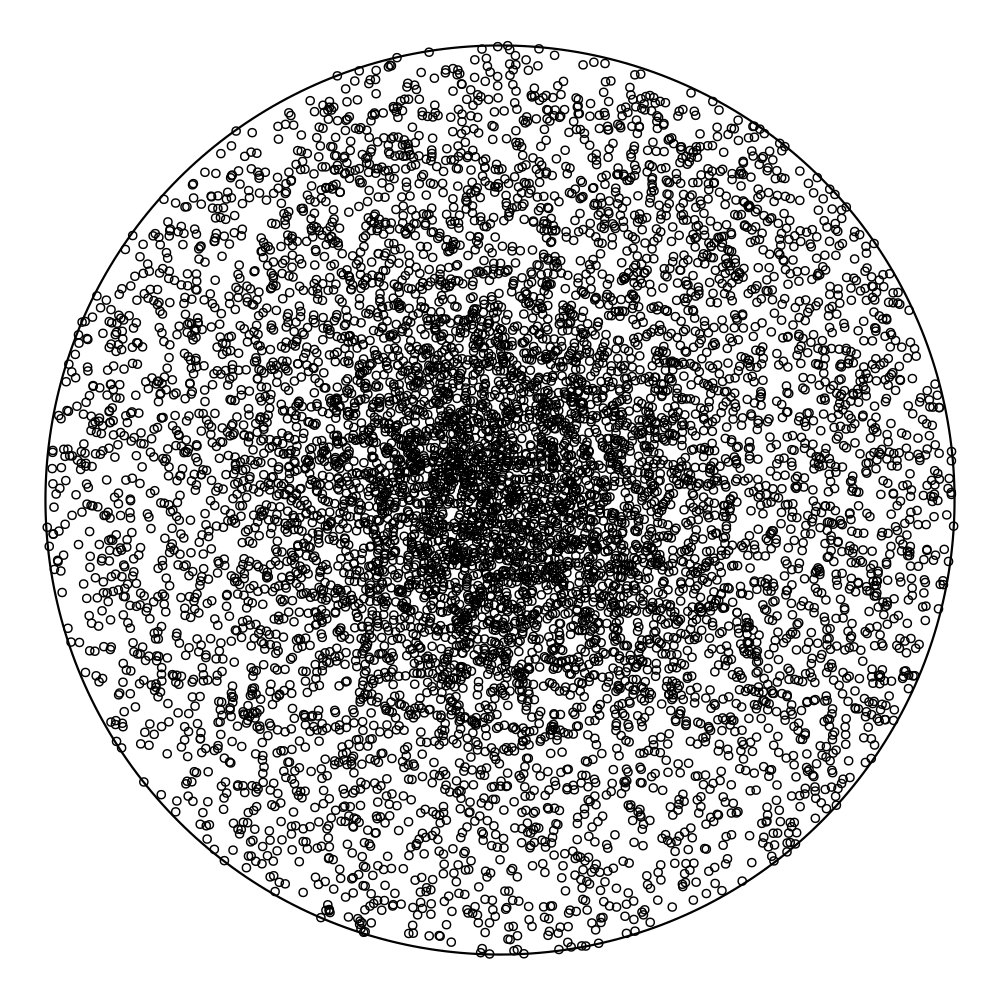}}\hfill
	\subfloat[]{\includegraphics[width=0.245\textwidth]{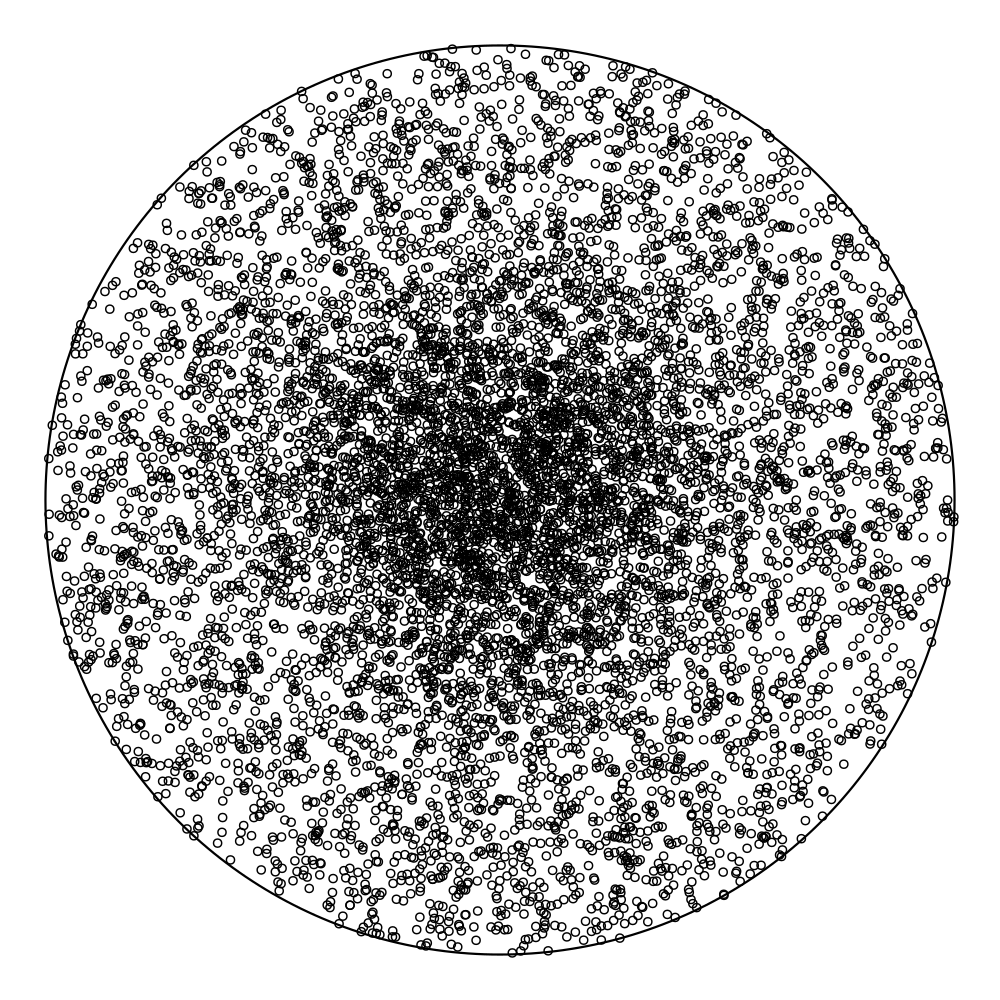}}
	\caption{Examples of Fry plots generated for the Monte Carlo test. (a) and (e) display the observed Fry plots from an isotropic cluster process and one with the geometric anisotropy mechanism, respectively. (b) - (d) and (f) - (h) show simulations obtained by random group-wise rotation of the observed Fry points.}
	\label{fig:fry-group-rot-cluster}
\end{figure}

\begin{figure}[th]
    \captionsetup[subfloat]{aboveskip=0.25pt}
	\subfloat[Observed Fry plot]{\includegraphics[width=0.245\textwidth]{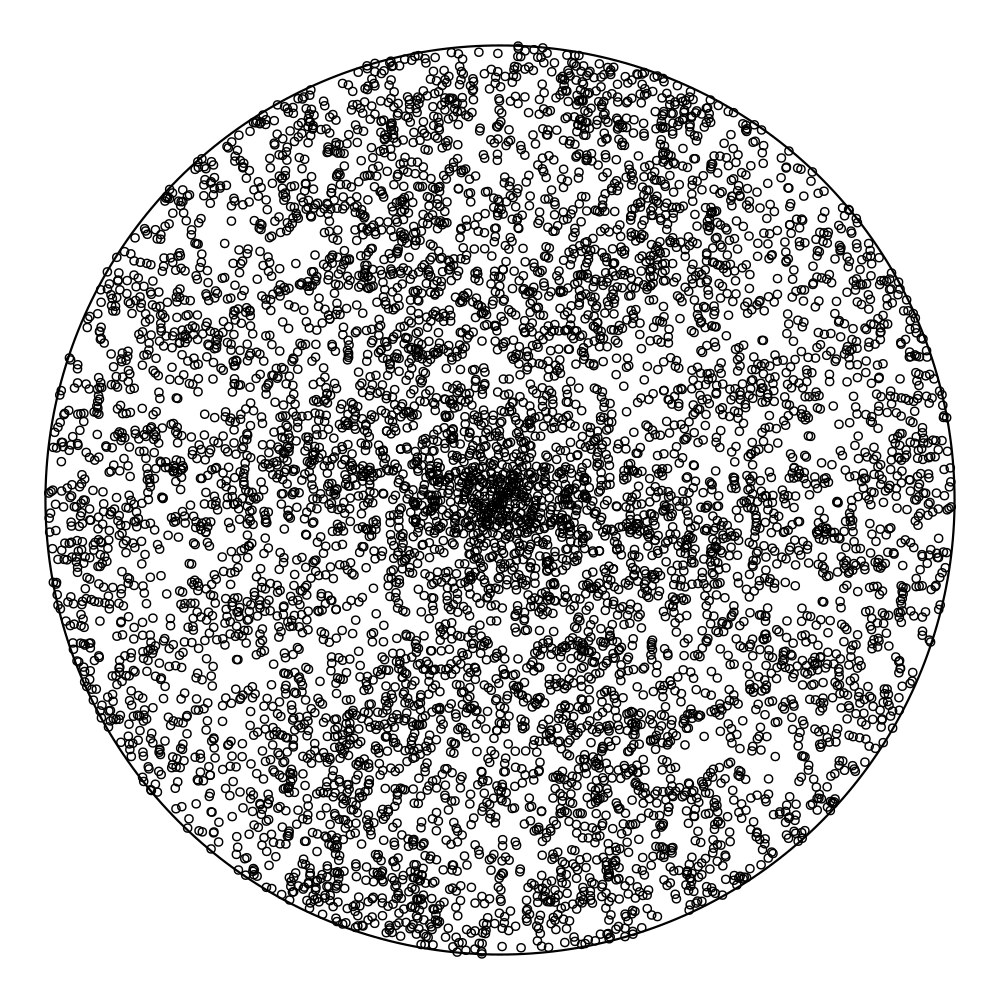}} \hfill
	\subfloat[]{\includegraphics[width=0.245\textwidth]{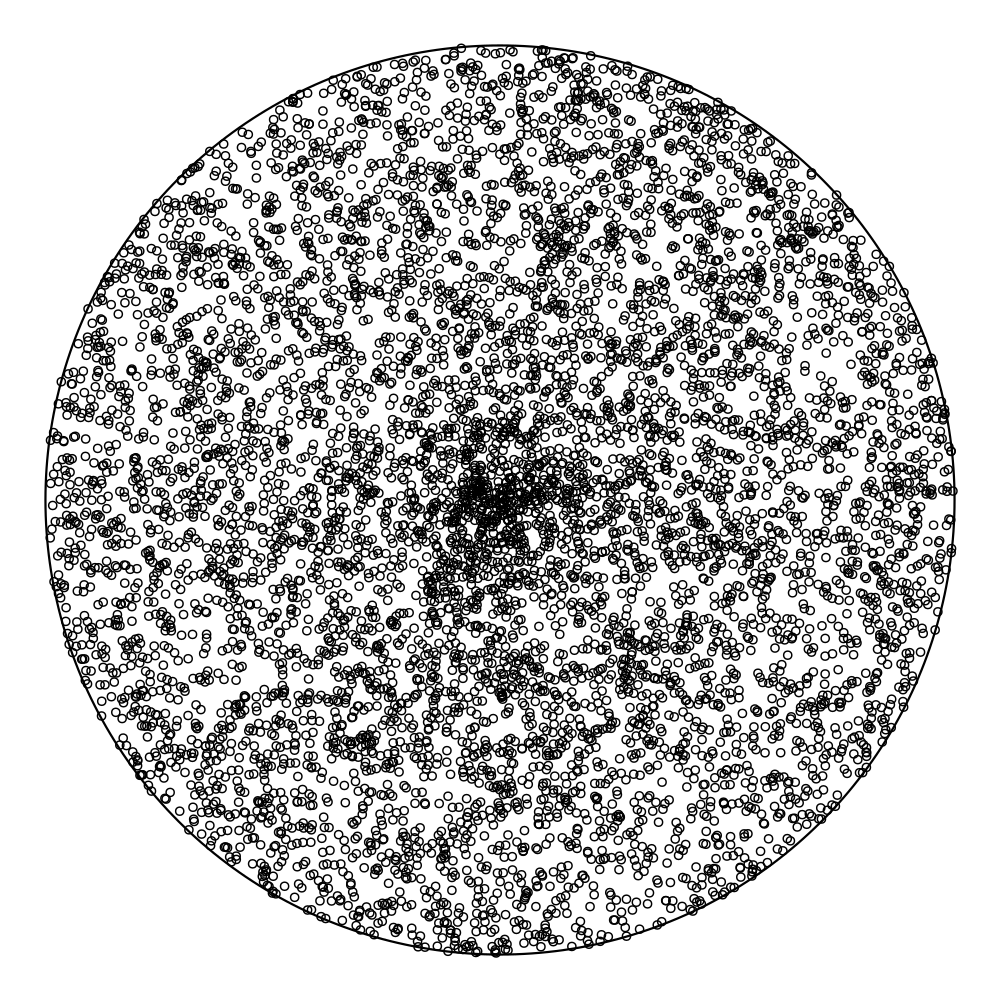}}\hfill 
	\subfloat[]{\includegraphics[width=0.245\textwidth]{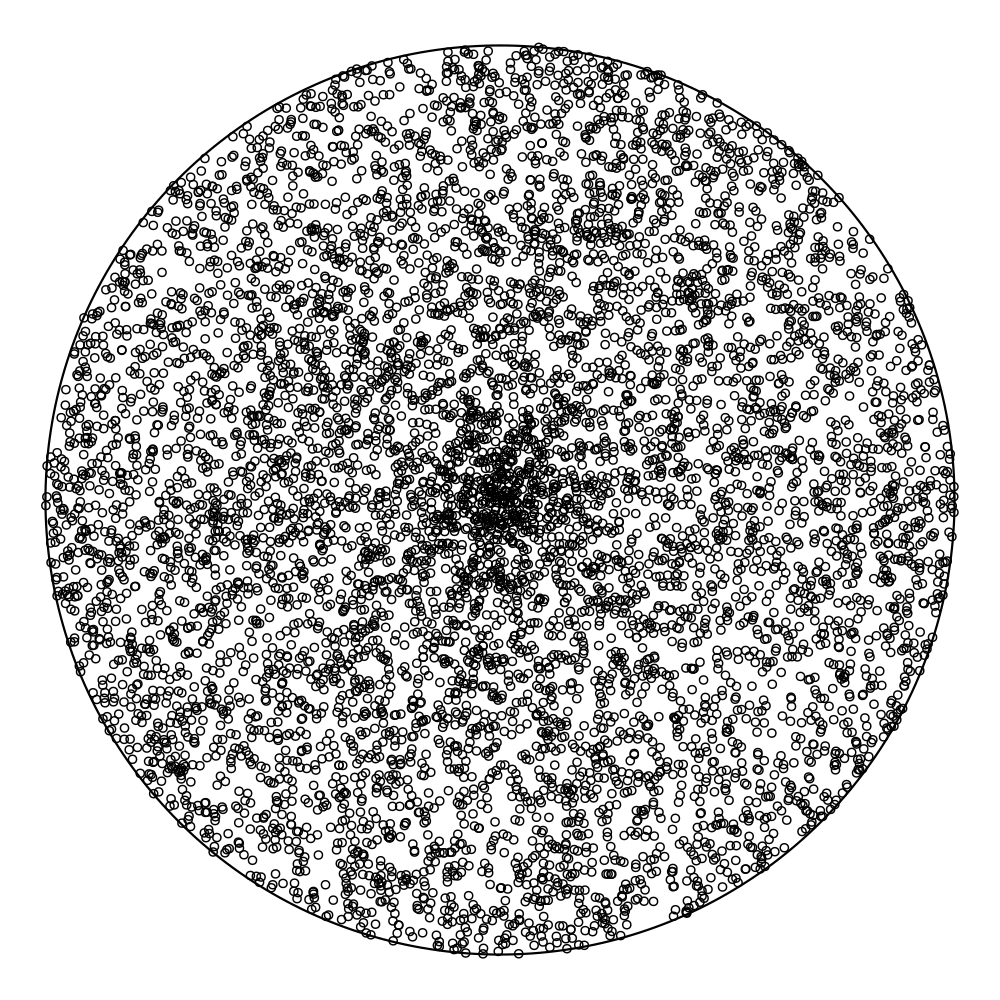}}\hfill
	\subfloat[]{\includegraphics[width=0.245\textwidth]{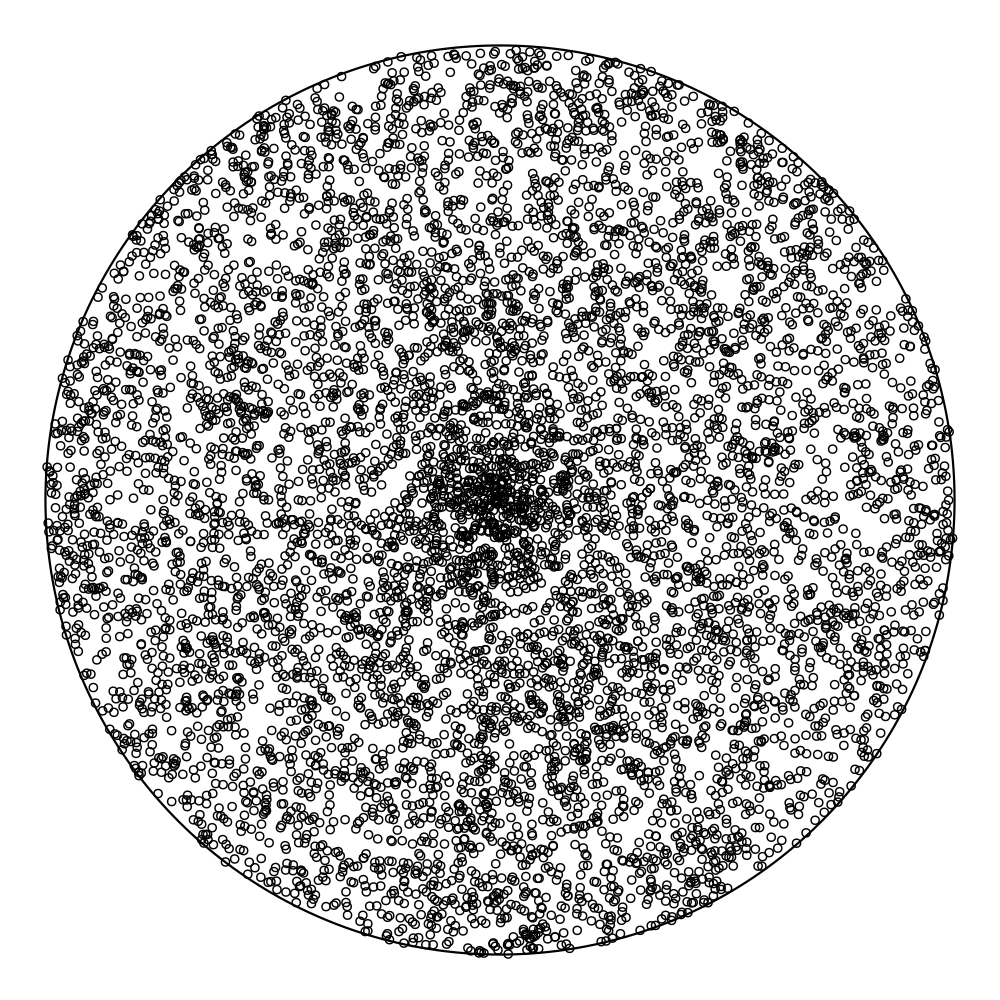}}\\
 
    \subfloat[Observed Fry plot]{\includegraphics[width=0.245\textwidth]{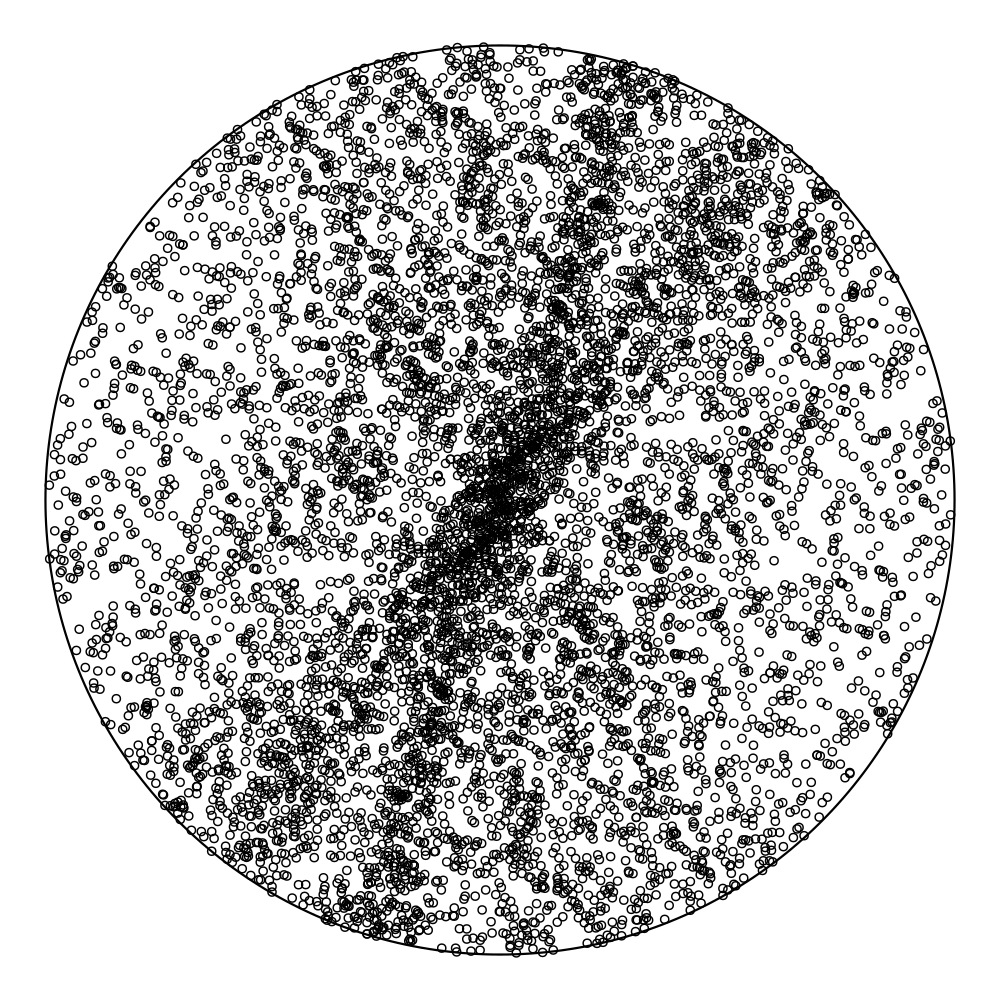}} \hfill
	\subfloat[]{\includegraphics[width=0.245\textwidth]{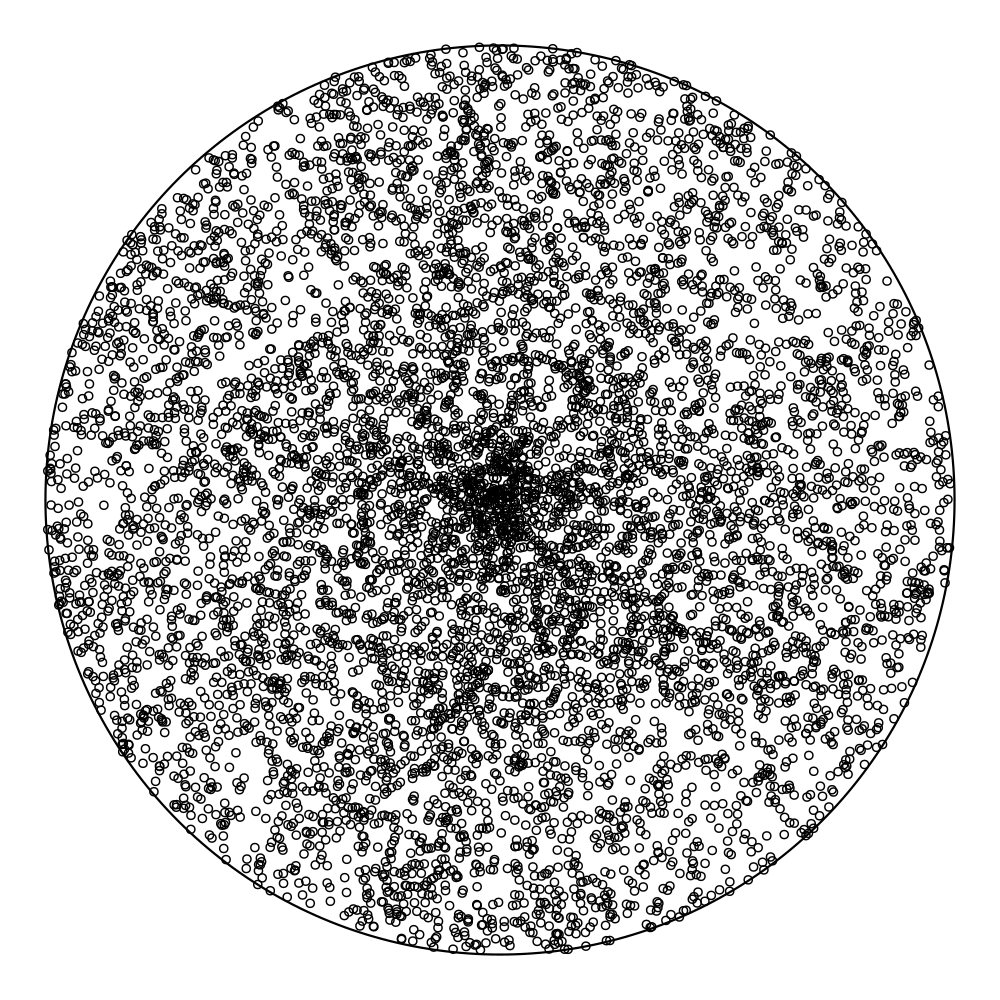}}\hfill 
	\subfloat[]{\includegraphics[width=0.245\textwidth]{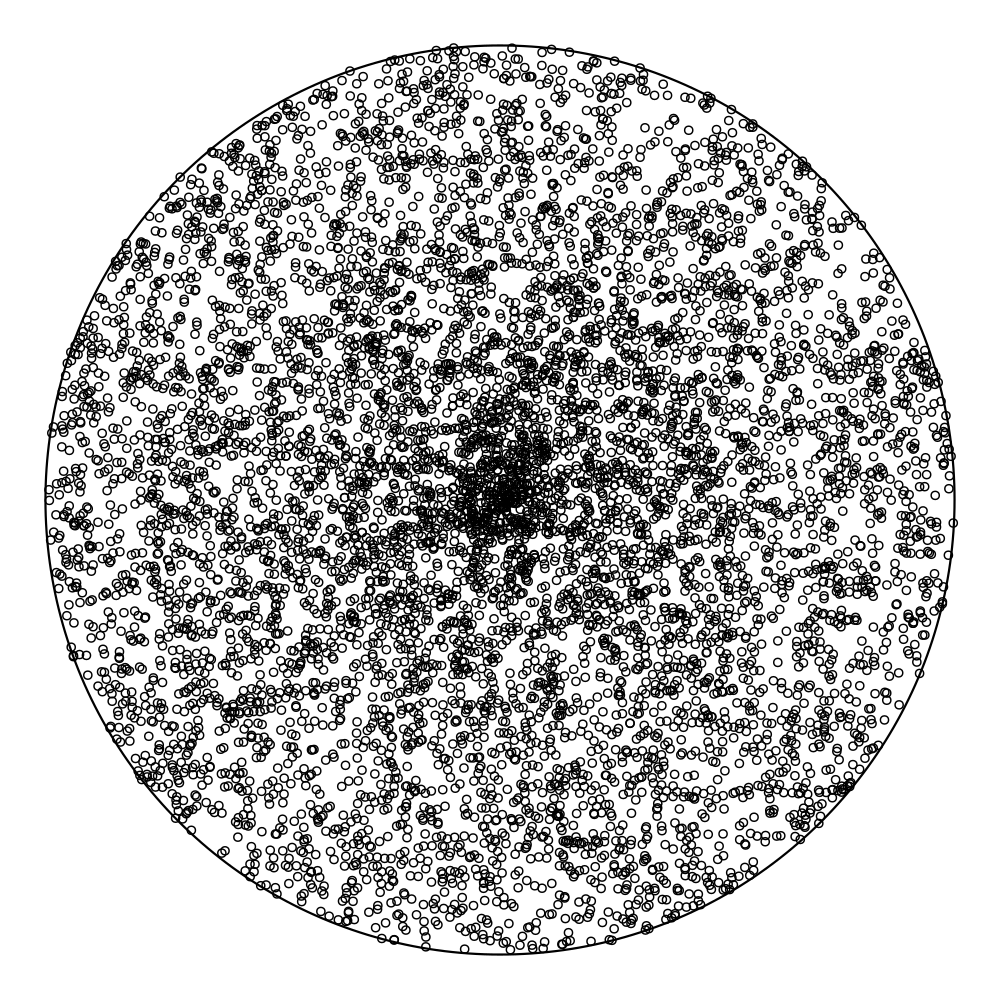}}\hfill
	\subfloat[]{\includegraphics[width=0.245\textwidth]{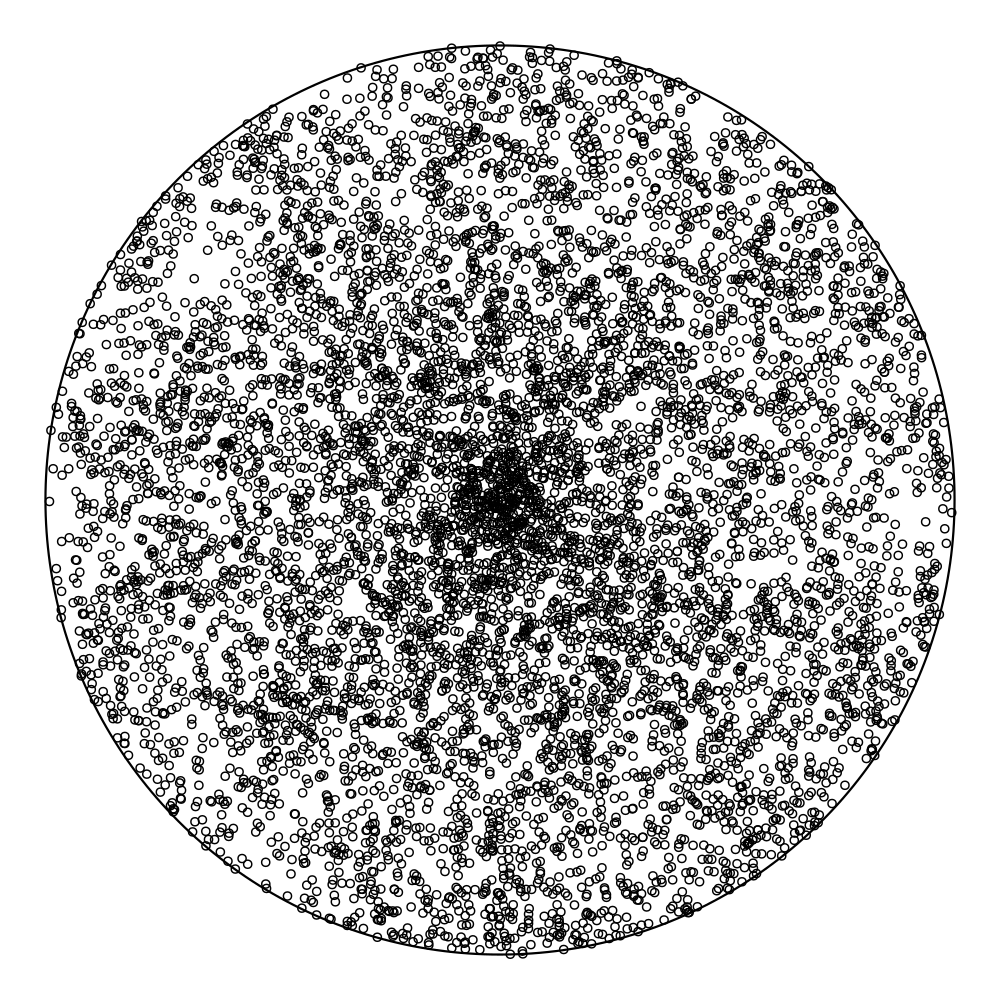}}
	\caption{Examples of Fry plots generated for the Monte Carlo test. (a) and (e) display the observed Fry plots from an isotropic and an anisotropic line cluster process, respectively. (b) - (d) and (f) - (h) show simulations obtained by random group-wise rotation of the observed Fry points.}
	\label{fig:fry-group-rot-line}
\end{figure}

\end{document}